\documentclass[11pt]{article}
\usepackage{subcaption}

\usepackage{pgfplots,pgfplotstable}
\usepgfplotslibrary{groupplots}
\pgfplotsset{compat=1.18}
\usepackage[font=small,labelfont=bf]{caption}
\usepackage{float}

\usepackage{amsbsy}
\usepackage{amssymb}
\usepackage{amsmath}
\usepackage{amsthm}
\usepackage{bm}
\usepackage{algorithm,algpseudocode}
\usepackage{paralist}

\usepackage{epsfig}
\usepackage{url}
\usepackage{multirow}
\usepackage{authblk}
\usepackage{stmaryrd}

\usepackage{booktabs,array,pdflscape}

\usepackage[pagebackref]{hyperref}

\newcommand{\sym}[1]{{\sf #1}}
\newcommand{\ceil}[1]{\mbox{$\lceil #1\rceil$}}
\newcommand{\floor}[1]{\mbox{$\lfloor #1\rfloor$}}

\newtheorem{lemma}{\bf Lemma}
\newtheorem{theorem}{\bf Theorem}

\newtheorem{proposition}{\bf Proposition}
\newtheorem{remark}{\bf Remark}
\newtheorem{construction}{\bf Construction}

\title{Vectorised Hashing Based on Bernstein-Rabin-Winograd Polynomials over Prime Order Fields}
\author[1]{Kaushik Nath}
\author[1]{Palash Sarkar}
\affil[1]{Indian Statistical Institute, 203, B.T. Road, Kolkata, India 700108}
\affil[ ]{Emails: {\tt kaushik.nath@yahoo.in, palash@isical.ac.in}}

\begin{document}
\maketitle

\begin{abstract}
	We introduce the new AXU hash function $c\mbox{-}\sym{decBRWHash}$, which is parameterised by the positive integer $c$ and
	is based on Bernstein-Rabin-Winograd (BRW) polynomials. 
	Choosing $c>1$ gives a hash function which can be implemented using $c$-way single instruction multiple data (SIMD) instructions. 
	We report a set of very comprehensive hand optimised assembly implementations of $4\mbox{-}\sym{decBRWHash}$ using {\tt avx2} SIMD instructions available on modern 
	Intel processors. For comparison, we also report similar carefully optimised {\tt avx2} assembly implementations of $\sym{polyHash}$, an AXU hash function based on 
	usual polynomials. Our implementations are over prime order fields, specifically the primes $2^{127}-1$ and $2^{130}-5$. 
	For the prime $2^{130}-5$, for {\tt avx2} implementations, compared to the famous Poly1305 hash function, $4\mbox{-}\sym{decBRWHash}$ 
	is faster for messages which are a few hundred bytes long and achieves a speed-up of about 16\% for message lengths in a few kilobytes range and improves to a speed-up of 
	about 22\% for message lengths in a few megabytes range. 
\end{abstract}

\section{Introduction\label{sec-intro}}

Authentication and authenticated encryption are two of the major functionalities of modern symmetric key cryptography. Almost XOR universal (AXU) hash functions
play an important role in both of these tasks. One of the most famous AXU hash functions is Poly1305~\cite{DBLP:conf/fse/Bernstein05}, and in combination
with XChaCha20~\cite{chacha} provides one of the most used authenticated encryption algorithm. Poly1305 is based on usual polynomials with arithmetic
done modulo the prime $2^{130}-5$. 

Present generation processors provide support for single input multiple data (SIMD) instructions. These instructions permit implementation of vectorised
algorithms. In a vectorised algorithm, at every step a single instruction is applied to a number of data items. Such vectorised algorithms have the potential
to provide significant efficiency improvements over conventional sequential algorithms. However, to apply SIMD instructions it is required to rewrite the
basic algorithm in vectorised form. For Poly1305, SIMD implementation was earlier reported in~\cite{Goll-Gueron}.

A class of polynomials was introduced in~\cite{Be07} to construct AXU hash functions, and later these polynomials were named the BRW polynomials~\cite{Sa09}.
An important theoretical advantage of hash functions based on BRW polynomials is that the number of field multiplications required by such hash functions
is about half the number of field multiplications required by hash functions based on usual polynomials~\cite{Be07} (see~\cite{DBLP:journals/dcc/GhoshS19,BNS2025} for 
further complexity improvements). This feature makes BRW polynomials an attractive option for constructing AXU hash function. An extensive study of both
BRW polynomials based hash functions (named $\sym{BRWHash}$) and usual polynomial based hash functions (named $\sym{polyHash}$)
for the primes $2^{127}-1$ and $2^{130}-5$ was carried out 
in~\cite{BNS2025} (see~\cite{cryptoeprint:2025/1224} for an update on~\cite{BNS2025}). The implementations reported in~\cite{BNS2025,cryptoeprint:2025/1224} were sequential, i.e. not SIMD, and
for such implementations it was observed that BRW based AXU hash functions indeed provides significant speed improvements (though less than what is theoretically
predicted) for both of the primes that were considered. The natural question that arises is whether a similar speed improvement can be achieved for 
SIMD implementation.

Efficient algorithms for computing the value of a BRW polynomial at a particular  point were reported in~\cite{DBLP:journals/dcc/GhoshS19,BNS2025}. Unfortunately,
there is no good way to rewrite these algorithms in vectorised form. There is a certain amount of parallelism present in the computation of 
BRW polynomials, which has been exploited for hardware implementation~\cite{DBLP:journals/tc/ChakrabortyMRS13}. This parallelism, however, does not permit
vector computation. 

The main theoretical contribution of the present work is to present a new AXU hash function based on BRW polynomials which permits an efficient vector implementation. 
We define the hash function $c\mbox{-}\sym{decBRWHash}$ which is parameterised by the positive integer $c$. Suppose $c=4$, which is the case that we implement. The basic idea is to
decimate the input stream into four parallel streams of the same length and perform independent BRW polynomial computation on each of the streams. Finally the 
four outputs of the four streams are combined using usual polynomial hashing. So usual polynomial hashing is applied only to a small and fixed number of blocks. 
Since the four BRW computations are independent and are on the sequences of the same length, it is possible to
apply the algorithm from~\cite{BNS2025} in a vectorised manner to perform simultaneous computations of the four BRW polynomials. As a result, the entire algorithm
becomes amenable to SIMD implementation. 

The above strategy is the natural way to build an SIMD friendly hashing strategy based on BRW polynomials. While the value of $c$
is a design parameter and is fixed, we show that for any $c^\prime$ dividing $c$, it is possible to implement $c\mbox{-}\sym{decBRWHash}$ using
$c^\prime$-way SIMD instructions. For example, $4\mbox{-}\sym{decBRWHash}$ can be implemented both sequentially and also using 2-way SIMD along with
4-way SIMD instructions. 

The hash function $c\mbox{-}\sym{decBRWHash}$ is a generalisation of $\sym{BRWHash}$ in the sense that for $c=1$, 
$c\mbox{-}\sym{decBRWHash}$ becomes exactly $\sym{BRWHash}$. While the idea behind $c\mbox{-}\sym{decBRWHash}$ is simple, there is a subtlety
in the choice of key used for combining the outputs of the $c$ streams. We prove that 
$c\mbox{-}\sym{decBRWHash}$ is indeed an AXU hash function, whose AXU bound is almost the same as the AXU bound of $\sym{BRWHash}$ for small $c$. 

From a practical point of view, we report implementations of $4\mbox{-}\sym{decBRWHash}$ for both the primes $2^{127}-1$ and $2^{130}-5$ using the {\tt avx2}
SIMD instructions available on modern Intel processors. For completeness and for the sake of comparison, we also report new implementations of $\sym{polyHash}$ using
{\tt avx2} instructions for both $2^{127}-1$ and $2^{130}-5$. Our implementations are comprehensive in the sense that we consider all feasible values of the
various implementation parameters for both $4\mbox{-}\sym{decBRWHash}$ and $\sym{polyHash}$. The implementations that we report are in assembly language 
targeting the {\tt avx2} and {\tt mulx}/{\tt adcx}/{adox} instructions of the Intel architecture and were
meticulously hand optimised. Our hand optimised {\tt avx2} implementation in assembly language of $\sym{polyHash}1305$ (i.e. $\sym{polyHash}$ based on the prime
$2^{130}-5$) is of independent interest, since if the message length is a multiple of eight and an appropriate key clamping is used, then $\sym{polyHash}1305$ is exactly 
the well known hash function Poly1305. 
To the best of our knowledge, there is no previous hand optimised {\tt avx2} assembly language implementation of Poly1305 which systematically considers all
feasible values of the implementation parameters. The independent work~\cite{DBLP:journals/tches/PegolottiGPP26} considers SIMD implementation
of Poly1305 using {\tt avx2}, {\tt avx512}, and {\tt avx512\_ifma} instruction sets with feasible values of implementation parameters, but the
code is in C with intrinsics, and not assembly. Further, our {\tt avx2} implementation outperforms the {\tt avx2} implementation reported
in~\cite{DBLP:journals/tches/PegolottiGPP26}.

We obtained extensive timing results for all our implementations. These results show that for {\tt avx2} implementations, the prime $2^{127}-1$ is a slower option
than the prime $2^{130}-5$ for both the hash functions $4\mbox{-}\sym{decBRWHash}$ and $\sym{polyHash}$. We provide a detailed explanation for this observation. 
In view of this observation, in this work we present only the timing results for the prime $2^{130}-5$. For the case of $\sym{polyHash}1305$, the timing results
show that for files which are longer than a few hundred bytes, the {\tt avx2} implementation is faster than the previously reported sequential implementation~\cite{BNS2025}.
Of more interest in the present context is the comparison between the hash functions $4\mbox{-}\sym{decBRWHash}$ and $\sym{polyHash}$. The timings results for the
prime $2^{130}-5$ show that for 
{\tt avx2} implementations, $4\mbox{-}\sym{decBRWHash}$ is faster than $\sym{polyHash}$ for messages which are a few hundred bytes long, and achieves a speed-up of about 
16\% (for kilobyte size messages) to 22\% (for megabyte size messages). Since Poly1305 and $\sym{polyHash}$ over the prime $2^{130}-5$ have the same speed, the
previous statement for $2^{130}-5$ also applies to the speed-up of $4\mbox{-}\sym{decBRWHash}$ over Poly1305.
In a typical file system, text files are usually about a few kilobytes long, while media files such as high resolution pictures, audio and video files, 
are a few megabytes long (see~\cite{DBLP:conf/asist/DinneenN21}). 
Compared to Poly1305, the new hash function $4\mbox{-}\sym{decBRWHash}1305$ provides a faster option for authentication. 
For authenticated encryption, the speed gain will be smaller since the hash function will have to be coupled with encryption based on a block or a stream cipher,
and so the speed improvement of only the hash function will be amortised over the time taken by the encryption module. 

\paragraph*{Previous and related works.} AXU hash functions are a generalisation of the notion of universal hash functions~\cite{DBLP:journals/jcss/CarterW79}. 
Research over the last
few decades have resulted in a sizeable literature on AXU hash functions. Overviews of the literature can be found 
in~\cite{Be04,DBLP:conf/fse/Bernstein05,Be07,DBLP:journals/dcc/Sarkar11,DBLP:journals/dcc/Sarkar13,DBLP:conf/sp/DegabrieleGGP24}. We mention only the works which are relevant
to the present paper. 

Polynomial hash functions were proposed independently in three 
papers~\cite{DBLP:journals/jcs/Boer93,DBLP:conf/crypto/Taylor93,DBLP:conf/crypto/BierbrauerJKS93}. 
The prime $2^{127}-1$ for use in polynomial hashing was first proposed in~\cite{DBLP:conf/crypto/Taylor93}, and was later used in~\cite{Be04,DBLP:conf/fse/KohnoVW04,BNS2025}. 
BRW polynomials were proposed by Bernstein~\cite{Be07} based on earlier work by Rabin and Winograd~\cite{RW72}. Implementations of BRW polynomials in both software
and hardware over binary extension fields were reported 
in~\cite{DBLP:journals/tc/ChakrabortyMRS13,DBLP:journals/tc/ChakrabortyMS15,DBLP:journals/tosc/ChakrabortyGS17,DBLP:journals/dcc/GhoshS19}. For prime order fields,
sequential software implementations of BRW polynomials were reported in~\cite{BNS2025,cryptoeprint:2025/1224}. 

Two recent works, namely~\cite{DBLP:conf/sp/DegabrieleGGP24} and~\cite{DBLP:journals/tches/PegolottiGPP26}, study implementation of polynomial
hashing. Sequential implementation was considered in~\cite{DBLP:conf/sp/DegabrieleGGP24}, while~\cite{DBLP:journals/tches/PegolottiGPP26} considered
SIMD implementations. With respect to code, both~\cite{DBLP:conf/sp/DegabrieleGGP24} and~\cite{DBLP:journals/tches/PegolottiGPP26} focus on developing
code generators, rather than direct hand written assembly implementation. We note that the present work was done independently 
of~\cite{DBLP:journals/tches/PegolottiGPP26} --
the ePrint version (available at \url{https://eprint.iacr.org/archive/2025/1262/20250708:155337}) of the present paper is dated 8 July 2025,
whereas the received date of~\cite{DBLP:journals/tches/PegolottiGPP26} is 15 July 2025. 
The main thrust of the present paper is to study SIMD friendly hash function based on BRW polynomials, which was not considered 
in~\cite{DBLP:conf/sp/DegabrieleGGP24,DBLP:journals/tches/PegolottiGPP26}. For Poly1305, we report hand written {\tt avx2} assembly code, 
whereas~\cite{DBLP:journals/tches/PegolottiGPP26} provides generated C with intrinsics code for {\tt avx2}, {\tt avx512}, and {\tt avx512\_ifma}
instruction sets, and also Jasmin code for {\tt avx2} instruction set. Our {\tt avx2} implementation of $\sym{polyHash}1305$ is faster than
the {\tt avx2} implementation of Poly1305 in~\cite{DBLP:journals/tches/PegolottiGPP26}.

\paragraph*{Overview of the paper.} Section~\ref{sec-prelim} provides the preliminaries. The new construction of decimated BRW hash function
is presented in Section~\ref{sec-SIMD-BRW-Horner}. The various aspects of implementation are given in Sections~\ref{sec-fld-arith} and~\ref{sec-vec-algo}. 
Descriptions of the implementations and the timing results are given in Section~\ref{sec-times}. Finally Section~\ref{sec-conclu} concludes the paper. 
For the ease of reference, in Appendix~\ref{sec-BRW-algo} we provide the algorithm from~\cite{BNS2025} for computing BRW polynomials. 

\section{Preliminaries\label{sec-prelim}}
The cardinality of a finite set $S$ will be denoted as $\#S$. Logarithms to the base two will be denoted by $\lg$. 
For a positive integer $i$, and $0\leq j<2^i$, by $\sym{bin}_i(j)$ we will denote the $i$-bit binary representation of $j$. For example, $\sym{bin}_4(3)=0011$ 
and $\sym{bin}_4(13)=1101$.

Let $\mathcal{D}$ be a non-empty set, $(\mathcal{R},+)$ be a finite group and $\mathcal{K}$ be a finite non-empty set. 
Let $\{\sym{Hash}_{\tau}\}_{\tau\in\mathcal{K}}$ be a family of functions, where for each $\tau\in\mathcal{K}$, $\sym{Hash}_{\tau}:\mathcal{D}\rightarrow \mathcal{R}$. 
The sets $\mathcal{D}$, $\mathcal{K}$ and $\mathcal{R}$ are called the message, key and tag (or digest) spaces respectively.
Let $a,a^{\prime}\in \mathcal{D}$ with $a\neq a^\prime$ and $b\in \mathcal{R}$.
The differential probability corresponding to the triple $(a,a^{\prime},b)$ is defined to be $\Pr_{\tau}[\sym{Hash}_{\tau}(a)-\sym{Hash}_{\tau}(a^{\prime})=b]$, where 
the probability is taken over a uniform random choice of $\tau$ from $\mathcal{K}$. 
If for every choice of distinct $a,a^{\prime}$ in $\mathcal{D}$ and $b\in \mathcal{R}$, the differential probability corresponding to $(a,a^{\prime},b)$ is at most 
$\epsilon$, then we say that the family $\{\sym{Hash}_{\tau}\}_{\tau\in\mathcal{K}}$ is $\epsilon$-almost XOR universal ($\epsilon$-AXU). 

Let $\mathbb{F}$ be a finite field.  Given a non-zero polynomial $P(x)\in \mathbb{F}[x]$, $\sym{deg}(P(x))$ denotes the degree of $P(x)$. 
Given $l\geq 0$ elements $M_1,\ldots,M_l$ in $\mathbb{F}$, we define two polynomials $\sym{Poly}(x;M_1,\ldots,M_l)$ and 
$\sym{BRW}(x;M_1, \allowbreak M_2, \ldots, \allowbreak M_{l})$ in $\mathbb{F}[x]$ with indeterminate $x$ and parameters $M_1,\ldots,M_l$ as follows.
\begin{eqnarray}
	\sym{Poly}(x;M_1,\ldots,M_l)
	& = & \left\{
		\begin{array}{ll} 
			0, & \mbox{if } l=0; \\
			M_1x^{l-1}+M_2x^{l-2}+\cdots+M_{l-1}x+M_l, & \mbox{if } l>0,
		\end{array} \right.
			\label{eqn-poly} 
\end{eqnarray}
and
\begin{tabbing}
\ \ \ \ \= $\bullet$ \= \ \ \ \ \=\ \ \ \ \= \kill
 \> $\bullet$ \> $\sym{BRW}(x;) = 0$; \\
 \> $\bullet$ \> $\sym{BRW}(x;M_1) = M_1$; \\
 \> $\bullet$ \> $\sym{BRW}(x;M_1,M_2) = M_1x+M_2$; \\
 \> $\bullet$ \> $\sym{BRW}(x;M_1,M_2,M_3) = (x + M_1)(x^2 + M_2) + M_3$; \\
 \> $\bullet$ \> $\sym{BRW}(x;M_1,M_2, \ldots, M_{i})$ \\
	\> \> \> \> $= \sym{BRW}(x;M_1,\ldots,M_{2^r-1})(x^{2^r} + M_{2^r}) + \sym{BRW}(x;M_{2^r+1},\ldots,M_{l})$; \\
	\> \> \> if  $2^r \in \{4,8,16,32,\ldots\}$ and $2^r\leq l < 2^{r+1}$, i.e. $2^r$ is the largest power of 2 such that $l \geq 2^r$.
\end{tabbing}
The BRW polynomials were introduced in~\cite{Be07} and named in~\cite{Sa09}. Note that for $l\geq 3$, $\sym{BRW}(x;M_1,\allowbreak M_2,\allowbreak \ldots,\allowbreak M_{l})$ 
is a monic polynomial. 

For $\tau\in \mathbb{F}$, using Horner's rule $\sym{Poly}(\tau;M_1,\ldots,M_l)$ can be evaluated using $l-1$ multiplications and same number of additions. 
For the BRW polynomials the following was proved in~\cite{Be07}.
\begin{theorem}\cite{Be07}\label{thm-BRW-basic} \ 
\begin{enumerate}
\item For every $l\geq 0$, the map from $\mathbb{F}^{l}$ to $\mathbb{F}[x]$ given by $$(M_1,\ldots,M_{l})\mapsto \sym{BRW}(x;M_1,\ldots,M_{l})$$ is injective. 
\item For $l\geq 1$, let $\mathfrak{d}(l)$ denote $\sym{deg}(\sym{BRW}(x;M_1,\ldots,M_l))$. For $l\geq 3$, 
$\mathfrak{d}(l)=2^{\lfloor\lg l\rfloor+1}-1$ and so $\mathfrak{d}(l)\leq 2l-1$; the bound is achieved if and only if $l=2^a$; 
and $\mathfrak{d}(l)=l$ if and only if $l=2^{a}-1$ for some integer $a\geq 2$. 
\item For $\tau\in \mathbb{F}$ and $l\geq 3$, $\sym{BRW}(\tau;M_1,\ldots,M_{l})$ can be computed using 
$\lfloor l/2\rfloor$ field multiplications (i.e. a multiplication over the field $\mathbb{F}$) and $\lfloor\lg l\rfloor$ additional field squarings to 
compute $\tau^2,\tau^4,\ldots,\tau^{2^{\lfloor\lg l\rfloor}}$.
\end{enumerate}
\end{theorem}
A field multiplication in $\mathbb{F}$ has two steps, namely a multiplication over the underlying ring (either the ring of integers, or the ring of polynomials), followed by
a reduction step. So $\lfloor l/2\rfloor$ field multiplications amounts to $\lfloor l/2\rfloor$ ring multiplications and $\lfloor l/2\rfloor$ reductions. 
By an unreduced multiplication we mean the ring multiplication with possibly a partial reduction. The following complexity improvement in computing
$\sym{BRW}(\tau;M_1,\ldots,M_{l})$ was proved in~\cite{BNS2025,DBLP:journals/dcc/GhoshS19}.
\begin{theorem}\cite{BNS2025,DBLP:journals/dcc/GhoshS19}\label{prop-BRW-comp}
For $\tau\in \mathbb{F}$ and $l\geq 3$, computing $\sym{BRW}(\tau;M_1,\ldots,M_{l})$ requires $\lfloor l/2\rfloor$ unreduced multiplications, $1+\floor{l/4}$ reductions, and 
	additionally requires $\lfloor\lg l\rfloor$ field squarings to compute $\tau^2,\tau^4,\ldots,\tau^{2^{\lfloor\lg l\rfloor}}$. 
\end{theorem}
The algorithm for evaluating $\sym{BRW}(\tau;M_1,\ldots,M_{l})$ given in~\cite{BNS2025} is
provided in Appendix~\ref{sec-BRW-algo}. The algorithm uses a parameter $t$ which is a small integer. The values $t=2,3,4$ and $5$ were considered in~\cite{BNS2025}
and the same values of $t$ will also be considered in the present work. 
\begin{proposition}[From Theorem~5.2 of~\cite{BNS2025}] \label{prop-stack-sz}
	Applying Algorithm~\ref{algo-B} in Appendix~\ref{sec-BRW-algo} to compute $\sym{BRW}(\tau;M_1,\ldots,M_{l})$ requires the stack size to be at most
	$\lfloor\lg l\rfloor-t+1$.
\end{proposition}

\subsection{Hash Functions $\sym{polyHash}$ and $\sym{BRWHash}$ \label{subsec-construct}}
Let $p$ be a prime and $\mathbb{F}_p$ be the finite field of order $p$. Our primary focus will be $2^{130}-5$ which is the prime underlying the hash function
Poly1305. We will also consider the prime $2^{127}-1$ which has turned out to be quite important (see~\cite{BNS2025,cryptoeprint:2025/1224}).
Given the prime $p$, we define the integers $m$, $n$, $k$ and $\mu$ as shown in Table~\ref{tab-params}.
\begin{table}
	\centering
		\begin{tabular}{|l|c|c|c|c|c|}
			\hline
			\multicolumn{1}{|c|}{$p$}       & $m$ & $n$ & $k$ & $\mu$ \\ \hline
			$2^{127}-1$ & 127 & 120 & 126 & 126 \\ \hline
			$2^{130}-5$ & 130 & 128 & 128 & 128 \\ \hline
		\end{tabular}
		\caption{The parameters $m$, $n$, $k$ and $\mu$ for the primes $2^{127}-1$ and $2^{130}-5$. \label{tab-params} }
\end{table}
Elements of $\mathbb{F}_p$ can be represented as $m$-bit strings. Since $n$, $k$ and $\mu$ are less than $m$, we will consider $n$-bit, $k$-bit and $\mu$-bit strings to represent 
elements of $\mathbb{F}_p$, where the most significant $m-n$, $m-k$, and $m-\mu$ bits respectively are set to 0. \\

\noindent {\bf Formatting and padding:} 
A binary string $X$ of length $L\geq 0$ is formatted (or partitioned) into $\ell$ blocks $X_1,\ldots,X_{\ell}$, where the length of $X_i$ is $n$
for $1\leq i\leq \ell-1$, the length of $X_{\ell}$ is $s$ with $1\leq s\leq n$, and $X=X_1||X_2||\cdots||X_{\ell}$.
Note that if $X$ is the empty string, i.e. if $L=0$, then $\ell=0$. 
We call each $X_i$ to be a block. If the length of a block is $n$, then we call it a full block, otherwise we call it a partial block.
By $\sym{format}(X)$ we will denote the list $(X_1,\ldots,X_{\ell})$ obtained from $X$ using the above described procedure.
The following two padding schemes were described in~\cite{BNS2025}. 
\begin{itemize}
\item $\sym{pad}1(X_1,\ldots,X_{\ell})$ returns $(M_1,\ldots,M_{\ell})$, where 
$M_i=0^{m-n-2}||1||X_i$, for $i=1,\ldots,\ell-1$, and $M_{\ell}=0^{m-s-2}||1||X_\ell$. 
\item $\sym{pad}2(X_1,\ldots,X_{\ell})$ returns $(M_1,\ldots,M_{\ell}, \allowbreak \sym{bin}_{m-1}(L))$, where
$M_i=0^{m-n-1}||X_i$, for $i=1,\ldots,\ell-1$, and $M_{\ell}=0^{m-s-1}||X_\ell$.
\end{itemize}
For both the padding schemes, the length of each $M_i$, $i=1,\ldots,\ell$, is $m-1$ and we consider $M_i$ to be an element of $\mathbb{F}_p$.
For the padding scheme $\sym{pad}1$, there is no restriction on the value of $L$. On the other hand, for $\sym{pad}2$, the value of $L$ has to be less
than $2^{m-1}$. From Table~\ref{tab-params}, the values of $m-1$ for the two primes are 127 and 130, and so the restriction on $L$ is a non-issue in practice. In fact,
in our implementations we consider $L$ to be less than $2^{64}$, so that the binary representation of $L$ can be stored as a 64-bit quantity. This is sufficient for
all conceivable applications.



The hash functions $\sym{polyHash}$ and $\sym{BRWHash}$ were introduced in~\cite{BNS2025}. In particular, the hash function $\sym{polyHash}$ is based on the idea behind
the design of the hash function Poly1305.
The key space and digest space for both the families $\sym{polyHash}$ and $\sym{BRWHash}$ are $\{0,1\}^k$; $\tau$ denotes the $k$-bit key which is 
considered to be an element of $\mathbb{F}_p$.
The digest space is the group $(\mathbb{Z}_{2^\mu},+)$, and so the digest can be represented using a $\mu$-bit string.
The message space for $\sym{polyHash}$ is the set of all binary strings. The message space for $\sym{BRWHash}$ is the set of all binary strings of
lengths less than $2^{m-1}$; as mentioned above in our implementations we considered messages of lengths less than $2^{64}$. 

\begin{remark}\label{rem-tag-sp}
	The descriptions of $\sym{polyHash}$ and $\sym{BRWHash}$ in~\cite{BNS2025} did not include the parameter $\mu$. Instead both the key and tag spaces were defined to be 
	$\{0,1\}^k$. In this paper, we make the formal distinction between the key and the tag spaces by introducing the additional parameter $\mu$ to denote the size of tags. 
	This generalises the descriptions of the $\sym{polyHash}$ and $\sym{BRWHash}$, and we restate the result on the
	AXU bounds proved in~\cite{BNS2025} in terms of $\mu$ and $k$. 
\end{remark}

In the descriptions of $\sym{polyHash}$ and $\sym{BRWHash}$ given below, $X$ denotes a message which is a binary string of length $L\geq 0$. \\

\begin{construction}\label{cons-polyhash}
Given a binary string $X$, let $(M_1,\ldots,M_{\ell})$ be the output of $\sym{pad}1(\sym{format}(X))$. We define
\begin{eqnarray}\label{eqn-polyHash}
	\sym{polyHash}_{\tau}(X) & = & ( P_1(\tau;M_1,\ldots,M_{\ell}) \bmod p ) \bmod 2^\mu,
\end{eqnarray}
where $P_1(x;M_1,\ldots,M_{\ell})$ is a polynomial in $\mathbb{F}_p[x]$ defined as follows.
\begin{eqnarray}\label{eqn-polyHashpoly}
	P_1(x;M_1,\ldots,M_{\ell}) & = & x \cdot \sym{Poly}(x;M_1,\ldots,M_{\ell}).
\end{eqnarray}
\end{construction}
Note that if $X$ is the empty string, then $L=\ell=0$ and so $\sym{polyHash}_{\tau}(X)=0$.
The family $\sym{polyHash}$ is motivated by the design of Poly1305~\cite{DBLP:conf/fse/Bernstein05} for the prime $2^{130}-5$. The differences between
Poly1305 and $\sym{polyHash}$ are as follows.
\begin{enumerate}
	\item Poly1305 considers $X$ to be a sequence of bytes, whereas $\sym{polyHash}$ considers $X$ to be a sequence of bits.
	\item In Poly1305, certain bits of the key $\tau$ are ``clamped'', i.e. they are set to 0. In~\cite{DBLP:conf/fse/Bernstein05} the clamping of key bits
		helped in efficient floating point implementation. On the other hand, however, clamping reduces security. Since we are not interested in 
		floating point implementation, we do not include clamping of key bits in the specification of $\sym{polyHash}$.
	\item Poly1305 is defined only for the prime $2^{130}-5$, whereas $\sym{polyHash}$ can be instantiated by any appropriate prime. In~\cite{BNS2025},
		instantiations of $\sym{polyHash}$ were proposed using both $2^{130}-5$ and $2^{127}-1$.
\end{enumerate}
Suppose $\sym{pad}1(\sym{format}(X))$ returns $(M_1,\ldots,M_{\ell})$. Computing $\sym{polyHash}_\tau(X)$ requires $\ell$ field multiplications.
A delayed reduction strategy was proposed in~\cite{DBLP:conf/ches/BernsteinS12} and developed further in~\cite{Goll-Gueron} for computing 
$\sym{Poly}(\tau;M_1,\ldots,M_{\ell})$. For a parameter $g\geq 1$, the idea is to 
perform a sequence of $g$ unreduced multiplications and additions and then perform a single reduction. This strategy requires pre-computing the powers
$\tau,\tau^2,\tau^3,\ldots,\tau^g$. 
Using this strategy, it is possible to compute $\sym{polyHash}(X)$ using $\ell$ unreduced multiplications, $\ceil{\ell/g}$ reductions, and additionally $g-1$ 
field multiplications~\cite{Goll-Gueron,BNS2025}. The key powers $\tau,\tau^2,\tau^3,\ldots,\tau^g$ are required to be pre-computed before the
actual computation of $\sym{polyHash}$. See Table~\ref{tab-oc}.

\begin{construction}\label{cons-BRWHash}
Given a binary string $X$, let the output of $\sym{pad}2(\sym{format}(X))$ be $(M_1,\allowbreak\ldots,\allowbreak M_{\ell},\allowbreak \sym{bin}_{m-1}(L))$. 
	We define
\begin{eqnarray}\label{eqn-BRWHash}
	\sym{BRWHash}_{\tau}(X) & = & ( P_2(\tau; M_1,\ldots,M_{\ell},\sym{bin}_{m-1}(L)) \bmod p ) \bmod 2^\mu,
\end{eqnarray}
where $P_2(x;M_1,\ldots,M_{\ell},\sym{bin}_{m-1}(L))$ is a polynomial in $\mathbb{F}_p[x]$ defined as follows.
\begin{eqnarray}\label{eqn-BRWHashpoly}
	P_2(x;M_1,\ldots,M_{\ell},\sym{bin}_{m-1}(L)) & = & x (x\cdot \sym{BRW}(x;M_1,\ldots,M_\ell) + \sym{bin}_{m-1}(L) ). 
\end{eqnarray}
\end{construction}
Note that if $X$ is the empty string, then $L=\ell=0$ and so $\sym{BRWHash}_{\tau}(X)=0$.
Suppose $\sym{pad}2(\sym{format}(X))$ returns $(M_1,\ldots,M_{\ell},\sym{bin}_{m-1}(L))$. 
Computing $\sym{BRWHash}(X)$ requires $2+\floor{\ell/2}$ unreduced multiplications, $2+\floor{\ell/4}$ reductions, and additionally
$\floor{\lg\ell}$ field squarings~\cite{BNS2025}. The key powers $\tau,\tau^2,\tau^{2^2},\ldots,\tau^{2^{\floor{\lg\ell}}}$ are required to be pre-computed before the actual
computation of $\sym{BRWHash}$. See Table~\ref{tab-oc}.

The following two results were proved in~\cite{BNS2025} for the case $k=\mu$. Below we state the results for the more general case of separate $k$ and $\mu$. The proofs
are essentially the same as the proofs of the case $k=\mu$ given in~\cite{BNS2025}.
\begin{lemma}[Based on Lemma~4.1 of~\cite{BNS2025}] \label{lem-basic}
	Let $p=2^m-\delta$ be a prime and $\mu$ be a positive integer such that $\mu<m$ and $\delta < 2^\mu-1$. 
	Let $\alpha\in \mathbb{Z}_{2^\mu}$, and $P(x)$ and $P^{\prime}(x)$ be distinct polynomials in $\mathbb{F}_p[x]$ satisfying $P(0)=P^{\prime}(0)=0$. 
	The number of distinct $\tau\in \mathbb{F}_p$ such that 
	\begin{eqnarray}\label{eqn-basic1}	
		((P(\tau)\bmod p)\bmod 2^\mu) - ((P^{\prime}(\tau)\bmod p) \bmod 2^\mu) & \equiv  & \alpha \pmod {2^\mu}
	\end{eqnarray}
	is at most $2^{m-\mu+1}$ times the degree of the polynomial $P(x)-P^{\prime}(x)$. 

	Consequently, for $\tau$ chosen uniformly at random from a set $\emptyset \neq S\subseteq \mathbb{F}_p$, the probability
	that~\eqref{eqn-basic1} holds is at most $(2^{m-\mu+1}\cdot \sym{deg}(P(x)-P^{\prime}(x)))/\#S$.
\end{lemma}
In our applications of Lemma~\ref{lem-basic}, the set $S$ will be the key space of the hash functions. For the hash functions that we construct, the key space
(and hence $S$) will be the set $\{0,1\}^k$ for some $k$ such that $2^k<p$, and the $k$-bit keys are considered to be elements of $\mathbb{F}_p$. 
Lemma~\ref{lem-basic} reduces the problem of determining the probability that a uniform random $k$-bit string $\tau$ satisfies~\eqref{eqn-basic1} to the simpler 
problem of determining the degree of the non-zero polynomial $P(x)-P^{\prime}(x)\in \mathbb{F}_p[x]$. 
The values of $p$, $m$, $k$ and $\mu$ given in Table~\ref{tab-params} satisfy the conditions stated in Lemma~\ref{lem-basic}.


\begin{theorem}[Based on Theorem~4.7 of~\cite{BNS2025}]\label{thm-AXUbndPolyHash}
	Let $p=2^m-\delta$ be a prime and $\mu$ be a positive integer such that $\mu<m$ and $\delta < 2^\mu-1$.
	Let $X$ and $X^{\prime}$ be two distinct binary strings
	of lengths $L$ and $L^{\prime}$ respectively with $L\geq L^{\prime} \geq 0$, and $\alpha$ be an element of $\mathbb{Z}_{2^\mu}$. Let $\ell=\lceil L/n\rceil$.
	Suppose $\tau$ is chosen uniformly at random from $\{0,1\}^k$. Then
	\begin{eqnarray*}
		\Pr[\sym{polyHash}_{\tau}(X)-\sym{polyHash}_{\tau}(X^{\prime})=\alpha] & \leq & \ell\cdot 2^{m-k-\mu+1}, \\
		\Pr[\sym{BRWHash}_{\tau}(X)-\sym{BRWHash}_{\tau}(X^{\prime})=\alpha] & \leq & (1+2\ell)\cdot 2^{m-k-\mu+1}, \\
	\end{eqnarray*}
\end{theorem}

	Apart from $\sym{polyHash}$ and $\sym{BRWHash}$, two other hash functions, named $t\mbox{-}\sym{BRWHash}$ and $d\mbox{-}\sym{2LHash}$,
	as well as the hash function $d\mbox{-}\sym{Hash}$ (which is a combination of $\sym{polyHash}$ and $d\mbox{-}\sym{2LHash}$)
	were defined in~\cite{BNS2025}. Timing results from the sequential implementations reported in~\cite{BNS2025} indicated that among all the hash functions, for short messages
	$\sym{polyHash}$ is the fastest, while $d\mbox{-}\sym{2LHash}$ is the fastest for longer messages. However, fresh implementations of $\sym{BRWHash}$
	(and also $t\mbox{-}\sym{BRWHash}$) reported in~\cite{cryptoeprint:2025/1224} showed that among all the hash functions considered in~\cite{BNS2025}, 
	$\sym{BRWHash}$ is the fastest among
	all the four hash functions $\sym{polyHash}$, $\sym{BRWHash}$, $t\mbox{-}\sym{BRWHash}$ and $d\mbox{-}\sym{2LHash}$ for
	all message lengths and for both the primes $2^{127}-1$ and $2^{130}-5$. 
	In view of the fact that $\sym{BRWHash}$ is faster than both $t\mbox{-}\sym{BRWHash}$ and $d\mbox{-}\sym{2LHash}$, we do not consider the hash functions 
	$t\mbox{-}\sym{BRWHash}$ and $d\mbox{-}\sym{2LHash}$ (and also $d\mbox{-}\sym{Hash}$) in this work.

\section{Decimated BRW Hash \label{sec-SIMD-BRW-Horner}}
We describe the hash function family $c\mbox{-}\sym{decBRWHash}$. The key space is $\{0,1\}^k$; $\tau$ denotes the $k$-bit key which is considered to be an element of $\mathbb{F}_p$. 
The digest space is the group $(\mathbb{Z}_{2^\mu},+)$. The message space is the set of all binary strings of lengths less than $2^{m-1}$. 
For concreteness we refer to the primes and the parameters $m$, $n$, $k$ and $\mu$ given in Table~\ref{tab-params}.

\begin{construction}\label{cons-decBRWHash}
	The hash function is parameterised by a positive integer $c$.
Given a binary string $X$ of length $L\geq 0$, let $(M_1,\ldots,M_{\ell},\sym{bin}_{m-1}(L))$ be the output of $\sym{pad}2(\sym{format}(X))$. We define
\begin{eqnarray}\label{eqn-SIMDBRWHash}
	c\mbox{-}\sym{decBRWHash}_{\tau}(X) & = & ( Q(\tau; M_1,\ldots,M_{\ell},\sym{bin}_{m-1}(L)) \bmod p ) \bmod 2^\mu,
\end{eqnarray}
where $Q(x;M_1,\ldots,M_{\ell},\sym{bin}_{m-1}(L))$ is the polynomial in $\mathbb{F}_p[x]$ defined in the following manner.

Let $\mathfrak{n}=\lceil \ell/c\rceil$ and $\mathfrak{m}=c\mathfrak{n}$. Define $M_{\ell+1}=\cdots=M_{\mathfrak{m}}=0^{m-1}$. Let 
\begin{eqnarray*}
	Q_1(x) & = & \sym{BRW}(x;M_1,M_{c+1},M_{2c+1},\ldots,M_{\mathfrak{m}-c+1}), \\
	Q_2(x) & = & \sym{BRW}(x;M_2,M_{c+2},M_{2c+2}\ldots,M_{\mathfrak{m}-c+2}), \\
	\cdots \\
	Q_c(x) & = & \sym{BRW}(x;M_{c},M_{2c},M_{3c}\ldots,M_{\mathfrak{m}}).
\end{eqnarray*}
Note that each of the $Q_i$'s is a BRW polynomial on $\mathfrak{n}$ blocks. If $L=0$, let $d=1$ and if $L>0$, let $d=2^{\lfloor\lg\mathfrak{n}\rfloor+1}$. Define 
\begin{eqnarray}\label{eqn-Q}
	Q_{c+1}(x) & = & \sym{Poly}(x^d;Q_1(x),Q_2(x),\ldots,Q_c(x)) \nonumber \\
	& = & x^{(c-1)d}Q_1(x) + x^{(c-2)d}Q_2(x) + \cdots + x^dQ_{c-1}(x) + Q_c(x). 
\end{eqnarray}
Finally,
\begin{eqnarray}\label{eqn-decBRWHashpoly}
	Q(x;M_1,\ldots,M_{\ell},\sym{bin}_{m-1}(L)) & = & x (x\cdot Q_{c+1}(x) + \sym{bin}_{m-1}(L) ). 
\end{eqnarray}
\end{construction}
When the quantities $M_1,\ldots,M_{\ell},\sym{bin}_{m-1}(L)$ are clear from the context, we will write $Q(x)$ instead of $Q(x;M_1,\ldots,M_{\ell},\sym{bin}_{m-1}(L))$. 
Note that if $X$ is the empty string, then $L=\ell=0$, and $c\mbox{-}\sym{decBRWHash}_{\tau}(X)=0$.

The idea behind the construction of $\sym{decBRWHash}$ is to decimate the message blocks into $c$ independent streams, process each stream using BRW and then 
combine the outputs of the streams using Horner with an appropriate power of the key $\tau$. Choosing the proper power of $\tau$ for the Horner evaluation is important
to ensure security. We prove later that the choice of $\tau^d$ is appropriate. Further, the key powers, $\tau,\tau^2,\ldots,\tau^{d/2}$ are required for the
BRW computations. So the key power $\tau^d$ for the Horner computation is obtained from the last key power $\tau^{d/2}$ required for the BRW computation by one squaring.

\begin{remark}\label{rem-brw-gen}
	Suppose $c=1$. Then $Q_1(x)=\sym{BRW}(x;M_1,\ldots,M_\ell)$, $Q_2(x)=Q_1(x)$, and $Q(x)=x(x\cdot Q_1(x)+\sym{bin}_{m-1}(L))$. So with $c=1$, 
	the hash function $c\mbox{-}\sym{decBRWHash}$ becomes exactly the hash function $\sym{BRWHash}$. Consequently, $c\mbox{-}\sym{decBRWHash}$ is a generalisation of 
	$\sym{BRWHash}$. Note that for $c=1$, since $Q_2(x)=Q_1(x)$, $x^d$ is not required.
\end{remark}

\begin{remark}\label{rem-design-c}
	The parameter $c$ in the definition of $c\mbox{-}\sym{decBRWHash}$ is a \textit{design} parameter, i.e. changing the value of $c$ changes the
	hash function. For example, on the same input, the outputs of $2\mbox{-}\sym{decBRWHash}$ and $4\mbox{-}\sym{decBRWHash}$ are not necessarily the
	same. The AXU bound that we prove below is for a fixed value of $c$. 
\end{remark}


The complexity of computing $c\mbox{-}\sym{decBRWHash}_{\tau}(X)$ for $c>1$ is stated in the following result.
\begin{proposition}\label{prop-comp-decBRW}
	Let $c>1$ be an integer and $X$ be a binary string. 
	Suppose $\sym{pad}2(\sym{format}(X))$ returns $(M_1,\allowbreak \ldots,\allowbreak M_{\ell},\allowbreak \sym{bin}_{m-1}(L))$. 
	Computing $c\mbox{-}\sym{decBRWHash}_{\tau}(X)$ 
	requires $c\floor{\mathfrak{n}/2}$ unreduced multiplications, $c(1+\floor{\mathfrak{n}/4})$ reductions, $c+1$ field multiplications, and $\floor{\lg\mathfrak{n}}$
field squarings. 
\end{proposition}
\begin{proof}
	From Theorem~\ref{prop-BRW-comp}, computing each $Q_i$, $i=1,\ldots,c$, requires $\lfloor \mathfrak{n}/2\rfloor$ unreduced multiplications and 
	$1+\floor{\mathfrak{n}/4}$ reductions. The key powers $\tau^2,\tau^4,\ldots,\tau^{2^{\lfloor\lg \mathfrak{n}\rfloor}}$ are required in the computation
	of all the $Q_i$'s, and are computed only once using $\floor{\lg\mathfrak{n}}$ field squarings. Computing $Q_{c+1}$ from $Q_1,\ldots,Q_c$ requires
	$c-1$ field multiplications, and computing $Q$ from $Q_{c+1}$ requires two additional field multiplications. 
\end{proof}
For $c>1$, the key powers $\tau,\tau^2,\tau^4,\ldots,\tau^{2^{\floor{\lg\mathfrak{n}}+1}}$ are required to be stored. See Table~\ref{tab-oc} which compares the
operation counts and storage requirements for $\sym{polyHash}$ and $c\mbox{-}\sym{BRWHash}$. 
Compared to $\sym{BRWHash}$, for a small value of $c>1$, the hash function $c\mbox{-}\sym{BRWHash}$ requires a few extra unreduced multiplications and reductions, and a 
little less storage. Computed sequentially, both $\sym{BRWHash}$ and $c\mbox{-}\sym{BRWHash}$ have similar efficiencies for message which are longer than a few blocks 
(for short messages $\sym{BRWHash}$ will be faster than $c\mbox{-}\sym{BRWHash}$).
The main advantage of $c\mbox{-}\sym{BRWHash}$ is that can be implemented using SIMD operations, as we describe later.

\begin{table}
\centering
{\scriptsize
\begin{tabular}{|l|c|c|c|c|}
\cline{2-5}
\multicolumn{1}{c|}{} & unred mult & red & storage & pre-comp (mult) \\ \hline
	$\sym{polyHash}$ & $\ell$ & $\lceil\ell/g\rceil$ & $g$ & $g-1$ \\ \hline
$\sym{BRWHash}$ & $2+\lfloor \ell/2\rfloor$ & $3+\lfloor \ell/4\rfloor$ & $1+\lfloor \lg \ell\rfloor$ & $\lfloor \lg \ell\rfloor$ \\ \hline
$c\mbox{-}\sym{decBRWHash}$ 
	& $1+c(1+\lfloor \lceil\ell/c\rceil/2\rfloor)$ 
	& $1+c(2+\lfloor \lceil\ell/c\rceil/4\rfloor)$ 
	& $2+\lfloor \lg \lceil\ell/c\rceil \rfloor$ 
	& $1+\lfloor \lg \lceil\ell/c\rceil \rfloor$ \\ \hline
\end{tabular}
	\caption{Operation counts and storage requirement for the hash functions for $\ell$ message blocks. For $\sym{polyHash}$, the parameter $g$ is a positive integer. 
	For $c\mbox{-}\sym{decBRWHash}$, $c>1$.  \label{tab-oc} }
}
\end{table}

\paragraph{Naming convention.}
We adopt the following naming convention. For all the hash functions considered in this paper, there are two possible sets of parameters in Table~\ref{tab-params}. The choice
of the prime $p$ determines the values of $m$, $k$, $n$ and $\mu$. So for each of the hash functions, by specifying the value of $p$, we obtain two different instantiations.
If $p$ is chosen to be $2^{127}-1$, we append 1271 to the name of the hash function, and if $p$ is chosen to be $2^{130}-5$, we append 1305 to the name of the hash function. 

\subsection{Computing $c\mbox{-}\sym{decBRWHash}$ with $c^\prime$-way SIMD Instructions, where $c^\prime$ Divides $c$ \label{subsec-divides}}
The definition of $c\mbox{-}\sym{decBRWHash}$ shows a natural $c$-way SIMD implementation is possible. We show that in fact, it is possible
to compute $c\mbox{-}\sym{decBRWHash}$ using $c^\prime$-way SIMD for any $c^\prime$ which divides $c$. 

Let $c^{\prime\prime}=c/c^\prime$. 
To start with consider a $c$-way SIMD implementation. This requires an overall data-structure $D$ such that it is possible to apply $c$-way SIMD instructions
on $D$. In particular, $D$ will consist of several constituent sub-structures each of which has a $c$-way vector representation.  
On such $c$-way vectors, individual $c$-way SIMD instructions are applied.

Now consider each of these $c$-way vectors to be partitioned into $c^{\prime\prime}$ $c^\prime$-way vectors. As a result, the overall data structure
$D$ is partitioned into $c^{\prime\prime}$ data structures $D_1,\ldots,D_{c^{\prime\prime}}$. 

A loop of the $c$-way SIMD algorithm operating on $D$, updates $D$ in each iteration using $c$-way SIMD instructions. Let
$\sym{Update}_c$ denote the operation which applies one or more $c$-way SIMD instructions to its input. Then a loop which iteratively applies
$c$-way instructions to $D$ can be written in the following form.
\begin{tabbing}
	\ \ \ \ \=\ \ \ \ \=\ \ \ \ \kill
\> for $i=1,2,3,\ldots$ \\
	\> \> $D\leftarrow \sym{Update}_c(D)$ \\
\> end for
\end{tabbing}

The above loop of $c$-way SIMD instructions can be replaced by a loop of $c^\prime$-way SIMD instructions in the following manner.
\begin{tabbing}
        \ \ \ \ \= \ \ \ \ \=\ \ \ \ \=\ \ \ \ \kill
\> for $i=1,2,3,\ldots$ \\
	\> \> for $j\leftarrow 1$ to $c^{\prime\prime}$ \\
	\> \> \> $D_j\leftarrow \sym{Update}_{c^\prime}(D_j)$ \\
	\> \> end for \\
\> end for
\end{tabbing}

The essential idea is that at every step, instead of applying $c$-way SIMD instructions to $D$, one applies $c^\prime$-way SIMD instructions
to each $D_j$, for $j=1,\ldots,c^{\prime\prime}$. The correctness of the procedure is immediate. 

So $c\mbox{-}\sym{decBRWHash}$ can be computed using $c^\prime$-way SIMD instructions for any $c^\prime$ which divides $c$. For example, 
if we choose $c=4$, then $4\mbox{-}\sym{decBRWHash}$ can be computed sequentially (corresponding to $c^\prime=1$), using 2-way SIMD instructions 
(corresponding to $c^\prime=2$), and using 4-way SIMD instructions (corresponding to $c^\prime=4$). Similarly, 
$8\mbox{-}\sym{decBRWHash}$ can be computed either sequentially, or using 2-way, 4-way, or 8-way SIMD instructions.

\subsection{AXU Bounds\label{sec-AXU-bnd}}

The following result from~\cite{BNS2025} states the basic property of $\sym{pad}2$.
\begin{lemma}[Lemma 4.3 of~\cite{BNS2025}]\label{lem-str2str}
Let $X$ be a binary string of length $L\geq 0$. Then the map $X\mapsto \sym{pad}2(\sym{format}(X))$ is an injection.
\end{lemma}

\begin{lemma}\label{lem-str2poly}
	Let $X$ be a binary string of length $L\geq 0$. Let $\ell=\lceil L/n\rceil$.
	Suppose $(M_1,\allowbreak \ldots,\allowbreak M_{\ell},\allowbreak \sym{bin}_{m-1}(L))$ is the output of $\sym{pad}2(\sym{format}(X))$ and
	$Q(x;M_1,\ldots,M_{\ell},\sym{bin}_{m-1}(L))$ is the polynomial constructed from $X$ as in~\eqref{eqn-decBRWHashpoly}. Then 
	$X \mapsto Q(x;M_1,\ldots,M_{\ell},\sym{bin}_{m-1}(L))$ is an injection.
\end{lemma}
\begin{proof}
Let $X$ and $X^\prime$ be two distinct binary strings of lengths $L$ and $L^\prime$ respectively. We assume without loss of generality that
	$L\geq L^\prime\geq 0$. 
	Let $\ell = \ceil{L/n}$, $\mathfrak{n}=\ceil{\ell/c}$, $\mathfrak{m}=c\mathfrak{n}$, $d=2^{\floor{\lg \mathfrak{n}}+1}$, 
	and $\ell^\prime=\ceil{L^\prime/n}$, $\mathfrak{n}^\prime=\ceil{\ell^\prime/c}$, $\mathfrak{m}^\prime=c\mathfrak{n}^\prime$, $d^\prime=2^{\floor{\lg \mathfrak{n}^\prime}+1}$.
	Let $(M_1,\ldots,M_\ell,\sym{bin}_{m-1}(L))$ be the output of $\sym{pad}2(\sym{format}(X))$, and
	let $(M_1^\prime,\ldots,M_\ell^\prime,\sym{bin}_{m-1}(L^\prime))$ be the output of $\sym{pad}2(\sym{format}(X^\prime))$.

	Let $Q_1(x),\ldots,Q_c(x),Q_{c+1}$ and $Q(x)$ be the polynomials arising from $X$, and 
	$Q_1^\prime(x),\allowbreak \ldots,\allowbreak Q_c^\prime(x),\allowbreak Q_{c+1}^\prime(x)$ and $Q^\prime(x)$ be the polynomials arising from $X^\prime$. By construction, 
	the coefficient of $x$ in $Q(x)$ is $\sym{bin}_{m-1}(L)$ and 
	the coefficient of $x$ in $Q^\prime(x)$ is $\sym{bin}_{m-1}(L^\prime)$. So if $L\neq L^\prime$, then $Q(x)\neq Q^\prime(x)$. 

	Now suppose that $L=L^\prime$, which implies $\sym{bin}_{m-1}(L)=\sym{bin}_{m-1}(L^\prime)$, 
	$\ell=\ell^\prime$, $\mathfrak{n}=\mathfrak{n}^\prime$, $\mathfrak{m}=\mathfrak{m}^\prime$, and $d=d^\prime$. Since there is exactly one string
	of length $0$, $L=L^\prime$ and $X\neq X^\prime$ implies that both the strings $X$ and $X^\prime$ are non-empty and so $\ell=\ell^\prime>0$.
	Since $X\neq X^\prime$, by the injectivity of $\sym{pad}2$ (see Lemma~\ref{lem-str2str}), 
	$(M_1,\ldots,M_\ell,\sym{bin}_{m-1}(L))\neq (M_1^\prime,\ldots,M_\ell^\prime,\sym{bin}_{m-1}(L^\prime))$. Since $\sym{bin}_{m-1}(L)=\sym{bin}_{m-1}(L^\prime)$,
	it follows that $(M_1,\ldots,M_\ell)\neq (M_1^\prime,\ldots,M_\ell^\prime)$. Let $\imath$ be such that $M_\imath\neq M_\imath^\prime$, and suppose that
	$\imath =\jmath +cj$, for some $\jmath\in\{1,\ldots,c\}$. By construction
	\begin{eqnarray*}
		Q_\jmath(x) & = & \sym{BRW}(x;M_\jmath,M_{\jmath+c},\ldots,M_{\jmath+c(j-1)},M_{\jmath+cj},M_{\jmath+c(j+1)},\ldots,M_{\mathfrak{m}-c+\jmath}), \\
		Q^\prime_\jmath(x) 
		& = & \sym{BRW}(x;M^\prime_\jmath,M^\prime_{\jmath+c},\ldots,M^\prime_{\jmath+c(j-1)},
		       M^\prime_{\jmath+cj},M^\prime_{\jmath+c(j+1)},\ldots,M^\prime_{\mathfrak{m}-c+\jmath}).
	\end{eqnarray*}
	Since $M_{\jmath+cj}=M_\imath\neq M_\imath^\prime=M^\prime_{\jmath+cj}$, by the injectivity of BRW polynomials (first point of Theorem~\ref{thm-BRW-basic}),
	$Q_\jmath(x)\neq Q^\prime_\jmath(x)$.

	For each $i=1,\ldots,c$, both $Q_i(x)$ and $Q_i^\prime(x)$ are BRW polynomials built from $\mathfrak{n}$ blocks (where $\mathfrak{n}>0$ since $\ell>0$) and 
	hence from the second point of Theorem~\ref{thm-BRW-basic}, 
	the degree of both $Q_i(x)$ and $Q_i^\prime(x)$ is $2^{\floor{\lg \mathfrak{n}}+1}-1=d-1$. From the definition of $Q_{c+1}(x)$ in~\eqref{eqn-Q}, the
	coefficients of $Q_{c+1}(x)$ are exactly the coefficients of $Q_i(x)$, $i=1,\ldots,c$, and similarly the coefficients of $Q_{c+1}^\prime(x)$ are exactly the coefficients of
	$Q_i^\prime(x)$, $i=1,\ldots,c$. Since $Q_\jmath(x)\neq Q_\jmath^\prime(x)$, it follows that $Q_{c+1}(x)\neq Q_{c+1}^\prime(x)$ and hence $Q(x)\neq Q^\prime(x)$.
\end{proof}

\begin{lemma}\label{lem-deg} 
Let $X$ be a binary string of length $L\geq 1$ and $n$ be a positive integer. Let $\ell=\lceil L/n\rceil$. 
Let $(M_1,\ldots,M_{\ell},\sym{bin}_{m-1}(L))$ be the output of $\sym{pad}2(\sym{format}(X))$. Then the following holds.
\begin{compactenum}
	\item If $c\mid \ell$, then $\ell+1 < \sym{deg}(Q(x;M_1,\ldots,M_{\ell},\sym{bin}_{m-1}(L))) \leq 2\ell+1$.
	\item If $c\nmid \ell$, then $\ell+1 < \sym{deg}(Q(x;M_1,\ldots,M_{\ell},\sym{bin}_{m-1}(L))) < 2\ell+2c+1$.
\end{compactenum}
\end{lemma}
\begin{proof}
	As argued in the proof of Lemma~\ref{lem-str2poly} the degree of $Q_i(x)$ is $d-1$ for $i=1,\ldots,c$. So from~\eqref{eqn-Q}, the degree of 
	$Q_{c+1}(x)$ is $cd-1$ and hence the degree of $Q(x)$ is $cd+1$, where $d=2^{\floor{\lg \mathfrak{n}}+1}$, and $\mathfrak{n}=\ceil{\ell/c}$. 
	Suppose $\floor{\lg \mathfrak{n}}=\rho$, i.e. $2^\rho\leq \mathfrak{n}=\ceil{\ell/c} <2^{\rho+1}$. So the degree of $Q(x)$ is $c2^{\rho+1}+1$. 

	If $c\mid \ell$, then $c2^\rho \leq \ell < c2^{\rho+1}$ from which we obtain the first point.
	If $c\nmid \ell$, then write $\ell/c = a+f$, where $a$ is an integer and $0<f<1$. So $\ceil{\ell/c}=a+1$ and 
	$2^\rho \leq a+1 < 2^{\rho+1}$. Using $a=\ell/c-f$, we obtain $c2^\rho \leq \ell + c(1-f) < c2^{\rho+1}$. This yields
	$\ell+1+c(1-f) < c2^{\rho+1}+1 \leq 2\ell + 2c(1-f) + 1$. Since $0<f<1$, we obtain the second point.
\end{proof}

\begin{theorem}\label{thm-AXUbnddecBRWHash}
	Let $p=2^m-\delta$ be a prime and $\mu$ be a positive integer such that $\mu<m$ and $\delta < 2^\mu-1$.
	Let $X$ and $X^{\prime}$ be two distinct binary strings
	of lengths $L$ and $L^{\prime}$ respectively with $L\geq L^{\prime} \geq 0$, and $\alpha$ be an element of $\mathbb{Z}_{2^\mu}$. Let $\ell=\lceil L/n\rceil$.
	Suppose $\tau$ is chosen uniformly at random from $\{0,1\}^k$. Then the following holds.
	\begin{compactenum}
	\item If $c=1$, then $\Pr[c\mbox{-}\sym{decBRWHash}_{\tau}(X)-c\mbox{-}\sym{decBRWHash}_{\tau}(X^{\prime})=\alpha] \leq (2\ell+1)\cdot 2^{m-k-\mu+1}$.
	\item If $c>1$, then $\Pr[c\mbox{-}\sym{decBRWHash}_{\tau}(X)-c\mbox{-}\sym{decBRWHash}_{\tau}(X^{\prime})=\alpha] < (2\ell+2c+1)\cdot 2^{m-k-\mu+1}$.
	\end{compactenum}
\end{theorem}
\begin{proof}
	Since $X$ and $X^\prime$ are distinct, by Lemma~\ref{lem-str2poly}, the corresponding polynomials $Q(x)$ and $Q^\prime(x)$ are also distinct. 
	Further, by construction the constant terms of both $Q(x)$ and $Q^\prime(x)$ are zero. Using Lemma~\ref{lem-basic}, the required probability
	is at most $2^{m-k-\mu+1}$ times the degree of $Q(x)$. From Lemma~\ref{lem-deg}, the degree of $Q(x)$ is at most $2\ell+1$ if $c=1$, and is less than $2\ell+2c+1$ if $c>1$.
\end{proof}
Note that for $c=1$, the AXU bound of $c\mbox{-}\sym{decBRWHash}$ is the same as that of $\sym{BRWHash}$. This is a consequence of the fact that
$c\mbox{-}\sym{decBRWHash}$ is a generalisation of $\sym{BRWHash}$. Note that the AXU bound degrades with $c$. In practice, $c$ will be a small 
constant, and hence the degradation will be insignificant. In our implementations, we take $c=4$, and we do not envisage a future situation where it will
be useful to take $c$ to be greater than 8.

\section{Field Arithmetic \label{sec-fld-arith}}
The focus of our implementation is SIMD operations. In particular, we focus on the {\tt avx2} instructions of Intel processors.
The presently available {\tt avx2} instructions determine how the elements of the field $\mathbb{F}_p$ are represented and the field arithmetic is performed. 
First we describe the representation of individual field elements and arithmetic for a pair of field elements. Later we describe how the representation of a single
field element can be lifted to a vector of 4 field elements, and how simultaneous arithmetic is performed on 4 pairs of field elements. 
We note that previous works~\cite{Goll-Gueron,BNS2025,DBLP:conf/sp/DegabrieleGGP24} and also the independent work~\cite{DBLP:journals/tches/PegolottiGPP26} 
provide descriptions of field arithmetic. The descriptions in~\cite{BNS2025,DBLP:conf/sp/DegabrieleGGP24} are for 64-bit sequential implementations,
while the descriptions in~\cite{Goll-Gueron,DBLP:journals/tches/PegolottiGPP26} are for SIMD implementation. For {\tt avx2} implementations, compared
to previous works, we provide more details regarding overflow analysis based on limb sizes, partial reduction, delayed reduction, and reduction
chain. In particular, in comparison to~\cite{DBLP:journals/tches/PegolottiGPP26}, we provide the complete details specific to the primes 
$2^{130}-5$ and $2^{127}-1$, while~\cite{DBLP:journals/tches/PegolottiGPP26} sets up an optimisation problem for overflow analysis corresponding to 
pseudo-Mersenne primes to derive such details; the actual implementation details for $2^{130}-5$ and $2^{127}-1$ are not explicitly discussed.

The {\tt avx2} instructions allow applying the same operation simultaneously on four different pairs of operands. The basic data type is a 256-bit quantity which is
considered to be 4 64-bit words. Given two such 256-bit quantities, it is possible to simultaneously add or multiply the four pairs of 64-bit operands that arise
from the same 64-bit positions of the two 256-bit quantities. In particular, the instruction {\tt vpmuludq} performs 4 simultaneous multiplications and
{\tt vpaddq} performs 4 simultaneous additions; two other relevant instructions are {\tt vpand} (which performs 4 simultaneous bitwise AND operations),
{\tt vpsllq} (which performs 4 simultaneous left shifts), and {\tt vpsrlq} (which performs 4 simultaneous right shifts).

There is, however, no scope for handling overflow (i.e. the result of an arithmetic instruction is greater than or equal to $2^{64}$) with {\tt avx2}
instructions. So to ensure correctness of the computation, the result of the addition and multiplication instructions
must also fit within a 64-bit word. In particular, the add-with-carry operation is not available with {\tt avx2} instructions. 
Since there is no scope for handling overflow, to ensure the correctness of the results of addition and multiplication, the whole 64 bits of the operands 
cannot be information bits. For addition, at most the 63 least significant bits of the operands can contain information, so that the result of the addition is
at most a 64-bit quantity. For multiplication, at most the 32 least significant bits of the operands can contain information, so that the result of the 
multiplication is at most a 64-bit quantity. So in effect the {\tt avx2} instructions support 32-bit multiplication. 

\begin{remark}\label{rem-maax}
Intel processors also provide support for 64-bit integer multiplication. In particular, from the Haswell processor onwards three instructions,
	namely {\tt mulx}, {\tt adcx}, and {\tt adox}, are provided which allow a double carry chain multiplication and squaring to be 
	performed~\cite{OGVW12,OGV13,DBLP:journals/amco/0001022}. Implementations which utilise these instructions have been called {\tt maax} 
	implementations~\cite{DBLP:journals/amco/0001022}. For both $\sym{polyHash}$ and $\sym{BRWHash}$, {\tt maax} implementations were
	reported in~\cite{BNS2025,cryptoeprint:2025/1224} for both the primes $2^{130}-5$ and $2^{127}-1$. Later we compare the speeds of these {\tt maax} implementations
	with the speeds of the new {\tt avx2} implementations that are reported in this paper.
\end{remark}

Below we describe the representation and field arithmetic separately for the primes $2^{130}-5$ and $2^{127}-1$.

\subsection{Case of $p=2^{130}-5$ \label{subsec-fld-arith-p130-5}}
Elements in $\mathbb{F}_p$ are represented using 130-bit quantities.
An element $f$ of the field $\mathbb{F}_p$ is represented as a 5-limb quantity, where each limb is a 26-bit quantity, i.e.
\begin{eqnarray*}
	f & = & f_0 + f_1 2^{26} + f_2 2^{26\cdot 2} + f_3 2^{26\cdot 3} + f_4 2^{26\cdot 4},
\end{eqnarray*}
where each $f_i$ is a 26-bit quantity. We call the coefficients of the powers of $2^{26}$ to be the limbs of $f$. Sometimes we write 
$f$ as $(f_0,f_1,f_2,f_3,f_4)$. The reason for choosing base $2^{26}$ representation is that {\tt avx2} supports only 32-bit multiplications, so that
multiplication of two 26-bit operands results in a 52-bit operand which fits within a 64-bit word. Using 32 bits to store a 26-bit limb 
is an example of unsaturated representation; see~\cite{DBLP:journals/amco/0001022} for explanations of saturated and unsaturated representations
which introduced these terminology (as far as we are aware).

Suppose $e$ is another field element whose limbs are $e_0,e_1,e_2,e_3,e_4$. The product $ef\bmod p$ can be written as a 5-limb quantity 
$h=h_0+h_12^{26}+h_22^{26\cdot 2}+h_32^{26\cdot 3}+h_42^{26\cdot 4}$ as follows (see~\cite{Goll-Gueron}).
\begin{eqnarray} \label{eqn-5-limb-mult}
	\begin{array}{rcl}
		h_0 & = & e_0f_0 + 5(e_1f_4 + e_2f_3 + e_3g_2 + e_4g_1) \\
		h_1 & = & e_0f_1 + e_1f_0 + 5(e_2f_4 + e_3f_3 + e_4f_2) \\
		h_2 & = & e_0f_2 + e_1f_1 + e_2f_0 + 5(e_3f_4 + e_4f_3) \\
		h_3 & = & e_0f_3 + e_1f_2 + e_2f_1 + e_3f_0 + 5e_4f_4 \\
		h_4 & = & e_0f_4 + e_1f_f + e_2f_2 + e_3f_1 + e_4f_0.
	\end{array}
\end{eqnarray}
In the above we have used $2^{130}\equiv 5\bmod p$. 
Consider $h_0=e_0f_0 + 5(e_1f_4+e_2f_3+e_3f_2+e_4f_1)=(e_0f_0+e_1f_4+e_2f_3+e_3f_2+e_4f_1)+4(e_1f_4+e_2f_3+e_3f_2+e_4f_1)=u+v$, where $u=e_0f_0+e_1f_4+e_2f_3+e_3f_2+e_4f_1$ 
and $v=4(e_1f_4+e_2f_3+e_3f_2+e_4f_1)$.
Each of the cross product terms $e_if_j$ is 52-bit long; the sum of four such quantities is at most 54-bit long; the multiplication by 4 increases the length
by 2 bits, so $v$ is at most 56-bit long; by a similar reasoning $u$ is at most 55-bit long; so the sum $h_0=u+v$ is at most 57-bit long.
By a similar argument, the lengths of the other $h_j$'s are also at most 57 bits. So the limbs of $h$ are (at most) 57-bit quantities. 
\textit{By an unreduced multiplication we mean obtaining $(h_0,\ldots,h_4)$ from $(e_0,\ldots,e_4)$ and $(f_0,\ldots,f_4)$ as given in~\eqref{eqn-5-limb-mult}.}

Further reduction of the limbs of $h$ to 26-bit quantities are not immediately done. 
Recall that both grouped Horner and BRW evaluation support lazy or delayed reduction. The limbs of $h$ are stored as 64-bit quantities. If we perform limb-wise addition of at most
64 5-limb quantities all of whose limbs are 57 bits long, then the limbs of the final sum are at most 63 bits long, and so there is no overflow. So delayed reduction
strategy can be applied up to the sum of 64 quantities. 
Note that instead of being 26-bit quantities, if at most one of the $e_i$'s and at most one of the $f_j$'s were 27-bit quantities, then the limb sizes
would be 58 bits (instead of 57 bits), and there would be
no overflow when delayed reduction is employed up to the addition of 32 quantities. We take advantage of this observation during the reduction step, where we
allow one limb to be a 27-bit quantity. We found that for grouped Horner implementing delayed reduction beyond group size 4 did not lead to speed improvement.
For BRW evaluation on $\mathfrak{n}=\ceil{\ell/4}$ vector blocks, the maximum number of additions required by Algorithm~\ref{algo-B} in Appendix~\ref{sec-BRW-algo} due to delayed 
reduction is the size of the stack which by Proposition~\ref{prop-stack-sz} is at most $\floor{\lg\mathfrak{n}}-t+1$. So if the number of blocks $\ell$ is at most about
$2^{31+t}$, then there is no problem with delayed reduction.

\begin{remark}\label{rem-5R}
	In computing $h$ using~\eqref{eqn-5-limb-mult}, suppose $e$ is a fixed quantity, while $f$ varies. In such a situation, given $e=(e_0,e_1,\ldots,e_4)$
	it is advantageous to pre-compute $\tilde{e}=(5\cdot e_1, 5\cdot e_2, 5\cdot e_3, 5\cdot e_4)$, as then $h_0$ can be computed as
	$h_0=e_0f_0+(5\cdot e_1)f_4+(5\cdot e_2)f_3+(5\cdot e_3)f_2+(5\cdot e_4)f_1$, and similarly for $h_1$, $h_2$ and $h_3$. This saves the four multiplications by 5. This technique was used in~\cite{Goll-Gueron}.
\end{remark}

\begin{remark} \label{rem-karatsuba}
	The product form shown in~\eqref{eqn-5-limb-mult} corresponds to schoolbook multiplication. We have also experimented with Karatsuba multiplication,
	but did not find it to provide better performance than schoolbook multiplication.
\end{remark}

Next we consider the reduction. Suppose $h=h_0+h_12^{26}+h_22^{26\cdot 2}+h_32^{26\cdot 3}+h_42^{26\cdot 4}$, where we assume (due to possible delayed reduction)
that each $h_i$ is a 63-bit quantity. The goal of the reduction is to reduce (modulo $p$) so that each limb is a 26-bit quantity. This is a complete reduction and is
done once at the end of the computation. For intermediate reductions, we reduce all limbs other than $h_1$ to 26-bit quantitites and reduce $h_1$ to a 27-bit quantity. Such a 
partial reduction is faster than a complete reduction. 
As mentioned above, this does not cause an overflow when implementing delayed reduction.

The basic idea of the reduction for a limb is to retain the 26 least significant bits in the limb and add the other bits (which we call the carry out) to the next limb. 
In other words, for any $i\in \{0,\ldots,3\}$, write $h_i=h_{i,0}+h_{i,1}2^{26}$ with $h_{i,0}=h_i\bmod 2^{26}$, update $h_i$ to $h_{i,0}$ and add $h_{i,1}$ to $h_{i+1}$;
write $h_4=h_{4,0}+h_{4,1}2^{26}$ with $h_{4,0}=h_4\bmod 2^{26}$, update $h_4$ to $h_{4,0}$ and add $5h_{4,1}$ (using once again $2^{130}\equiv 5\bmod p$) to $h_0$.
Since $h_{i+1}$, $i=0,\ldots,3$ is a 63-bit quantity, adding $h_{i,1}$ to $h_{i+1}$ does not cause an overflow.
It is important to note that the reduction procedure can start from any $i\in\{0,\ldots,3\}$, and in particular the procedure does not have to start from $h_0$.
For the reduction, we use the reduction chain 
$h_0\!\!\shortrightarrow\!\! h_1\!\!\shortrightarrow\!\! h_2\!\!\shortrightarrow\!\! h_3\!\!\shortrightarrow\!\! h_4\!\!\shortrightarrow\!\! h_0\!\!\shortrightarrow\!\! h_1$, 
which is a chain having 6 steps. The first five steps reduce
$h_0,\ldots,h_4$ to 26-bit quantities, and the carry out of $h_4$ is at most a 37-bit quantity. Multiplying this carry out by 5 creates at most a 40-bit quantity, and
adding it to the 26-bit $h_0$ makes $h_0$ at most a 41-bit quantity. The last step $h_0\!\!\shortrightarrow\!\! h_1$ reduces $h_0$ to 26 bits and adds the at most 15-bit carry 
out to the 26-bit $h_1$ to make the new $h_1$ a 27-bit quantity.

\paragraph{SIMD implementation.}
Suppose $a_0,a_1,a_2$ and $a_3$ are four field elements, and for $i=0,\ldots,3$, suppose the limbs of $a_i$ are $a_{i,0},a_{i,1},\ldots,a_{i,4}$, where each $a_{i,j}$
is a 26-bit (or 27-bit) quantity. The total of 20 limbs of $a_0,a_1,a_2$ and $a_3$ are packed into 5 256-bit words $U_0,\ldots,U_4$ in the following manner. 
Consider $U_j$ to be $U_{j,0}||U_{j,1}||U_{j,2}||U_{j,3}$, where each $U_{j,i}$ is a 64-bit word. For $i=0,\ldots,3$, $a_{i,j}$ is
stored in the 26 least significant bits of $U_{j,i}$. See Figure~\ref{fig-SIMD-pack} for an illustration. This idea of packing is from~\cite{Goll-Gueron}
and was also used in~\cite{DBLP:journals/tches/PegolottiGPP26}.

Similarly, suppose $b_0,b_1,b_2$ and $b_3$ are four field elements
which are packed into 5 256-bits words $V_0,\ldots,V_4$. Let $c_i=a_ib_i\bmod p$, $i=0,\ldots,3$, where the 5-limb representation of $c_i$ is obtained
from the 5-limb representations of $a_i$ and $b_i$ in a manner similar to~\eqref{eqn-5-limb-mult}. Then the 5-limb representations of $c_0,c_1,c_2,c_3$
are obtained in 5 256-bit words $W_0,W_1,\ldots,W_4$. Using the {\tt avx2} instructions {\tt vpmuludq} and {\tt vpaddq}, it is possible to obtain 
$W_0,W_1,\ldots,W_4$ from $U_0,U_1,\ldots,U_4$ and $V_0,V_1,\ldots,V_4$. 
In particular, we note that 25 {\tt vpmuludq} instructions are required to obtain
all the cross product terms and additionally 4 {\tt vpmuludq} instructions are required to perform the multiplications by 5. The multiplications by 5 are not
required if one of the operands in each of the four multiplications is fixed (see Remark~\ref{rem-5R}).

The computation of grouped Horner and BRW proceeds using the delayed reduction strategy. So the limbs of the $W_j$'s grow to at most $63$-bit quantities.
Then the reduction strategy described above is applied using SIMD instructions to the 5 words $W_0,W_1,\ldots,W_4$. 

\begin{figure}
\setlength{\unitlength}{3mm}
\begin{center}
\begin{picture}(15,8)
\put(0,0){\makebox(1.6,1.5)}
\put(0,1.6){\makebox(1.6,1.5)}
\put(0,3.2){\makebox(1.6,1.5)}
\put(0,4.8){\makebox(1.6,1.5)}
\put(0,6.4){\makebox(1.6,1.5)}

\multiput(1.6,0)(3.2,0){4}{\framebox(3.2,1.5){}}
\multiput(1.6,1.6)(3.2,0){4}{\framebox(3.2,1.5){}}
\multiput(1.6,3.2)(3.2,0){4}{\framebox(3.2,1.5){}}
\multiput(1.6,4.8)(3.2,0){4}{\framebox(3.2,1.5){}}
\multiput(1.6,6.4)(3.2,0){4}{\framebox(3.2,1.5){}}

\put(0.75,0.75){\makebox(0,0){$U_4$}}
\put(0.75,2.35){\makebox(0,0){$U_3$}}
\put(0.75,3.95){\makebox(0,0){$U_2$}}
\put(0.75,5.55){\makebox(0,0){$U_1$}}
\put(0.75,7.15){\makebox(0,0){$U_0$}}

\put(3.75,0.75){\makebox(0,0){$a_{3,4}$}}
\put(6.75,0.75){\makebox(0,0){$a_{2,4}$}}
\put(9.75,0.75){\makebox(0,0){$a_{1,4}$}}
\put(12.75,0.75){\makebox(0,0){$a_{0,4}$}}

\put(3.75,2.35){\makebox(0,0){$a_{3,3}$}}
\put(6.75,2.35){\makebox(0,0){$a_{2,3}$}}
\put(9.75,2.35){\makebox(0,0){$a_{1,3}$}}
\put(12.75,2.35){\makebox(0,0){$a_{0,3}$}}

\put(3.75,3.95){\makebox(0,0){$a_{3,2}$}}
\put(6.75,3.95){\makebox(0,0){$a_{2,2}$}}
\put(9.75,3.95){\makebox(0,0){$a_{1,2}$}}
\put(12.75,3.95){\makebox(0,0){$a_{0,2}$}}

\put(3.75,5.55){\makebox(0,0){$a_{3,1}$}}
\put(6.75,5.55){\makebox(0,0){$a_{2,1}$}}
\put(9.75,5.55){\makebox(0,0){$a_{1,1}$}}
\put(12.75,5.55){\makebox(0,0){$a_{0,1}$}}

\put(3.75,7.15){\makebox(0,0){$a_{3,0}$}}
\put(6.75,7.15){\makebox(0,0){$a_{2,0}$}}
\put(9.75,7.15){\makebox(0,0){$a_{1,0}$}}
\put(12.75,7.15){\makebox(0,0){$a_{0,0}$}}

\end{picture}
\caption{Packing of four field elements $a_0,\ldots,a_3$ into 5 256-bit words $U_0,\ldots,U_4$.  \label{fig-SIMD-pack} }
\end{center}
\end{figure}

\paragraph{A different SIMD representation.}
While the above method of packing was suggested in~\cite{Goll-Gueron}, another
method of packing four field elements into 256-bit words was actually used in~\cite{Goll-Gueron}. Each of the four field elements is represented using
5 26-bit quantities and hence can fit in 5 32-bit words. So the four field elements together require 20 32-bit words to be stored. Three 256-words provide
a total of 24 32-bit words. So the 20 32-bit words representing the four field elements can be stored in three 256-bit words. This is the representation
of 4 field elements that was used in~\cite{Goll-Gueron}.
Stored in this manner, it is not
possible to directly perform the 4-way SIMD multiplication, and requires more instructions for unpacking and repacking. The rationale for adopting such a
strategy is that using 3 instead of 5 256-bit words to store operands frees up some of the 256-bit registers for performing the actual arithmetic, and this compensates
for the penalty incurred due to packing and unpacking. In our implementations, on the other hand, we have used 5 256-words to represent the operands so that
the multiplication operations can be directly applied. By carefully managing register allocation, we did not encounter the problem of unavailable registers. 
This was possible since we implemented directly in assembly, whereas the implementation in~\cite{Goll-Gueron} is in Intel intrinsics which is at a higher level.
Since in our approach the problem of unavailable registers does not arise, using the packed representation of field elements used in~\cite{Goll-Gueron} would
incur an unnecessary penalty. So we chose not to use that strategy.

\subsection{Case of $p=2^{127}-1$ \label{subsec-fld-arith-p127-1}}
Elements in $\mathbb{F}_p$ are represented using 128-bit quantities. Keeping in mind the fact that SIMD supports 32-bit multiplication, there are
two possible representations of elements of $\mathbb{F}_p$, namely a 4-limb, or a 5-limb representation. 

\paragraph{5-limb representation.}
The 5-limb representation is almost the same as that of the 5-limb representation for $2^{130}-5$ described in Section~\ref{subsec-fld-arith-p130-5}, 
i.e. a base $2^{26}$ representation can be used.
The only difference in the multiplication procedure shown in~\eqref{eqn-5-limb-mult} is that the constant 5 is replaced by 8, since $2^{130}\equiv 8 \bmod (2^{127}-1)$. 
For the reduction algorithm, we use the chain 
$h_3\!\!\shortrightarrow\!\! h_4\!\!\shortrightarrow\!\! h_0\!\!\shortrightarrow\!\! h_1\!\!\shortrightarrow\!\! h_2\!\!\shortrightarrow\!\! h_3\!\!\shortrightarrow\!\! h_4$, 
i.e. we start the chain at $h_3$ instead of starting at $h_0$.
The chain consists of 6 steps as in the case for $2^{130}-5$. For the step $h_4\!\!\shortrightarrow\!\! h_0$, we reduce $h_4$ to 23 bits (note that $4\times 26+23=127$) and produce
a carry out of at most 40 bits which is then added to $h_0$ (since $2^{127}\equiv 1\bmod (2^{127}-1)$, there is no need to multiply the carry out by any constant).
The chain finally stops at $h_4$ which results in $h_4$ being at most a 24-bit quantity. The number of operations required for multiplication and reduction
using the 5-limb representation of $2^{127}-1$ is almost the same as the number of operations required for multiplication and reduction using the 5-limb
representation of $2^{130}-5$. 

The SIMD implementation of the 5-limb representation for $2^{127}-1$ is also very similar to the SIMD implementation of the 5-limb representation for $2^{130}-5$.
Four field elements are stored in 5 256-bit words as shown in Figure~\ref{fig-SIMD-pack} and multiplication is done using {\tt avx2} instructions.
In particular, with the 5-limb representation, multiplication requires requires 25 {\tt vpmuludq} instructions to compute the cross product terms, plus
4 {\tt vpsllq} instructions for the multiplications by 8.

\paragraph{4-limb representation.}
For the 4-limb representation, an element $f$ of the field $\mathbb{F}_p$ is represented as a 4-limb quantity, where each limb is a 32-bit quantity, i.e.
\begin{eqnarray*}
	f & = & f_0 + f_1 2^{32} + f_2 2^{32\cdot 2} + f_3 2^{32\cdot 3}, 
\end{eqnarray*}
where each $f_i$ is a 32-bit quantity. 
Suppose $e$ is another field element whose limbs are $e_0,e_1,e_2,e_3$. The product $ef\bmod p$ can be written as a 4-limb quantity 
$h=h_0+h_12^{32}+h_22^{32\cdot 2}+h_32^{32\cdot 3}$ as follows.
\begin{eqnarray} \label{eqn-4-limb-mult}
	\begin{array}{rcl}
	h_0 & = & e_0f_0 + 2(e_1f_3+e_2f_2+e_3f_1) \\
	h_1 & = & e_0f_1 + e_1f_0 + 2(e_2f_3 + e_3f_2) \\
	h_2 & = & e_0f_2 + e_1f_1 + e_2f_0 + 2e_3f_3 \\
	h_3 & = & e_0f_3 + e_1f_2 + e_2f_1 + e_3f_0.
	\end{array}
\end{eqnarray}
In the above, we have used $2^{128}\equiv 2\bmod p$. 
Each of the cross product terms $e_if_j$ is 64-bit long. Adding together such terms increases the size of the sum beyond 64 bits. Since {\tt avx2}
instructions do not provide any mechanism to handle the carry arising out of additions, one cannot directly add the cross product terms. An alternative procedure
needs to be used. Consider $h_0$. 
Write $e_if_j$ as $u+v2^{32}$, where both $u$ and $v$ are 32-bit quantities. The $u$'s arising from the terms $e_1f_3$, $e_2f_2$ and $e_3f_1$ are 
added, the sum multiplied by 2, and the result added to the $u$ arising from the the term $e_0f_0$. This results in $h_0$ being a 36-bit quantity. 
Similarly the $v$'s arising from the terms $e_1f_3$, $e_2f_2$ and $e_3f_1$ are added, the sum multiplied by 2, and the result added to the $v$ arising from
the term $e_0f_0$, giving a value $v^\prime$. The computation of $h_1$ starts with the initial value $v^\prime$, and is updated by adding the $u$'s arising from the cross product
terms in the expression for $h_1$, while the $v$'s arising from the cross product terms in the expression for $h_1$ contribute to the value of $h_2$. 

For the reduction algorithm, the chain is 
$h_0\!\!\shortrightarrow\!\! h_1\!\!\shortrightarrow\!\! h_2\!\!\shortrightarrow\!\! h_3\!\!\shortrightarrow\!\! h_0\!\!\shortrightarrow\!\! h_1\!\!\shortrightarrow\!\! h_2\!\!\shortrightarrow\!\! h_3$ which consists of 7 steps 
(which is one step more than the
chain for the 5-limb representation). Note that the chain makes two full iterations over the limbs and reduces all the limbs to 32-bit quantities. This is required,
since if any limb is greater than 32 bits, then subsequent multiplication with such limbs will cause an overflow.

SIMD implementation packs four field elements with 32-bit limbs into 4 256-bit words in much the same as the packing of four field elements with 26-bit limbs into 
5 256-bit words. Using this packed representation, multiplication of four pairs of field elements is done using {\tt avx2} instructions following the description
given earlier. An advantage of using the 4-limb representation is that the number of {\tt vpmuludq} instructions required to compute the multiplication comes down
to 16 from 25. However, there is a significant increase in the number additions and shifts. Multiplication using the 4-limb representation requires
16 {\tt vpmuludq}, 40 {\tt vpaddq}, 16 {\tt vmovdqa}, 16 {\tt vpand} and 16 {\tt vpsrlq} operations (plus 3 {\tt vpsllq} instructions for the multiplications by
2). 
In contrast,
multiplication using the 5-limb representation requires 25 {\tt vpmuludq} and 20 {\tt vpaddq} instructions (plus 4 {\tt vpsrlq} instructions for the multiplications
by 8). The latencies of the various instructions on the Skylake processor are as follows:
{\tt vpmuludq} - 5, {\tt vpaddq} -1, {\tt vpand} - 1, {\tt vpsrlq} - 1, {\tt vmovdqa} - 7 for load and 5 for store. 
This shows that the penalty due to the additional instructions required for the 4-limb representation more than cancels the benefit of requiring a less
number of {\tt vpmuludq} instructions.

\paragraph{Non-availability of carry handling instructions.}
A major reason that 4-limb representation turns out to be slower is that SIMD instructions do not provide any mechanism to handle the carry out
of an addition operation. Since the cross product terms in~\eqref{eqn-4-limb-mult} are all 64-bit quantities, without the availability of an 
efficient carry handling mechanism, many more operations are required to prevent an overflow condition. If in the future SIMD instructions provide some mechanism
for obtaining the carry out of an addition, and/or the add-with-carry operation, then it is likely that the 4-limb representation will provide
a significantly faster multiplication algorithm than the 5-limb representation. We note that for the {\tt maax} implementation using
64-bit arithmetic there is excellent support for carry operations (see Remark~\ref{rem-maax}).

\paragraph{Non-availability of 64-bit multiplications.}
SIMD operations presently do not support full 64-bit multiplication. If in the future, full 64-bit
SIMD multiplication is supported (with 2 such simultaneous multiplications using 256-bit words, or 4 such simultaneous multiplications using
512-bit words), then for the prime $2^{130}-5$, a 3-limb
representation can be used, while for $2^{127}-1$, either a 2-limb, or a 3-limb representation can be used. The speed of multiplication using such a 2-limb representation
for $2^{127}-1$ has the potential to significantly outperform the speed of multiplication of a 3-limb representation of $2^{130}-5$, especially if along
with 64-bit multiplication, carry handling SIMD operations are used. This observation arises from the comparative speed performances of multiplication
algorithms for $2^{127}-1$ and $2^{130}-5$ for non-SIMD implementation using {\tt maax} instructions (see~\cite{BNS2025,cryptoeprint:2025/1224}).

We note that both {\tt avx512} and {\tt avx512\_ifma} work on 512-bit registers, but neither instruction set support 4 simultaneous full 64-bit 
multiplications. 

\section{Vectorised Algorithms \label{sec-vec-algo} }
We describe algorithms for vectorised computation of $\sym{polyHash}$ and $\sym{decBRWHash}$.

\subsection{Vectorised Computation of $\sym{polyHash}$ \label{subsec-vec-poly}}
One way to exploit parallelism in the computation of $\sym{polyHash}$ is to divide the sequence of blocks $(M_1,\ldots,M_{\ell})$ into $c\geq 2$ subsequences
and apply Horner's rule to each of the subsequence. For $c=4$, such a strategy was used in~\cite{Goll-Gueron} for vectorised implementation of Poly1305. 
In~\cite{DBLP:journals/tosc/ChakrabortyGS17}, this strategy was called $c$-decimated Horner evaluation. 

From~\eqref{eqn-polyHash} and~\eqref{eqn-polyHashpoly}, the computation of $\sym{polyHash}_{\tau}(X)$ requires the computation of 
$\tau\cdot \sym{Poly}(\tau;M_1,\allowbreak \ldots,\allowbreak M_\ell)$, 
where $(M_1,\ldots,M_\ell)$ is the output of $\sym{pad}1(\sym{format}(X))$. 
Let $\rho=\ell\bmod c$ and $\ell^{\prime}=(\ell-\rho)/c=\floor{\ell/c}$. The computation of $\tau\cdot \sym{Poly}(\tau;M_1,\ldots,M_{\ell})$ can be done
in the following manner. First compute
\begin{eqnarray}\label{eqn-4SIMD}
	P & = & 
	\left\{\begin{array}{llll}	
		  & \tau^{c}   & \cdot & \sym{Poly}({\tau^c};M_1,  M_{c+1},M_{2c+1},\ldots, M_{c\ell^{\prime}-c+1}) \\
		+ & \tau^{c-1} & \cdot & \sym{Poly}({\tau^c};M_{2},M_{c+2},M_{2c+2},\ldots,M_{c\ell^{\prime}-c+2}) \\
		+ & \cdots \\
		+ & \tau       & \cdot & \sym{Poly}({\tau^c};M_{c},M_{2c},M_{3c},\ldots,M_{c\ell^{\prime}}).
	\end{array}\right.
\end{eqnarray}
Then
\begin{eqnarray*}
	\begin{array}{rcll}
		\tau\cdot \sym{Poly}(\tau;M_1,\ldots,M_{\ell}) & = & P & \mbox{if } \rho=0, \\
		\tau\cdot \sym{Poly}(\tau;M_1,\ldots,M_{\ell}) & = & \tau \cdot \sym{Poly}(\tau;P+M_{\ell-\rho+1}, M_{\ell-\rho+2}, \ldots, M_{\ell}) & \mbox{if } \rho > 0.
	\end{array}
\end{eqnarray*}
We focus on the computation of $P$. For $i=1,\ldots,c$, let
\begin{eqnarray*}
	C_i & = & \tau^{c+1-i}\cdot \sym{Poly}(\tau^c;M_i,M_{c+i},\ldots,M_{c\ell^\prime-c+i}).
\end{eqnarray*}
Then $P=C_1+C_2+\cdots+C_c$. So the computation of $P$ consists of computing $C_1,\ldots,C_c$ and then adding these together. The computation of the vector 
\begin{eqnarray}\label{eqn-C-vec}
	\mathbf{C} & = & (C_1,\ldots,C_c)
\end{eqnarray}
can be done using an SIMD strategy. The computation of $P$ from $\mathbf{C}$, and the subsequent computation of $\tau\cdot \sym{Poly}(\tau;M_1,\ldots,M_{\ell})$ 
from $P$ (in the case $\rho>0$) is done using a small number of {\tt maax} operations.

For a non-negative integer $j$, let $\bm{\tau}^j=(\tau^{j},\ldots,\tau^{j})$ be a vector of length $c$. Further, let $\bm{\tau}_{\theta}=(\tau^c,\tau^{c-1},\ldots,\tau)$ also
be a vector of length $c$. For $i=1,\ldots,\ell^\prime$, let $\mathbf{M}_i=(M_{ci-c+1},\ldots,M_{ci})$. Then
\begin{eqnarray} \label{eqn-vec-poly}
	\mathbf{C} & = & \bm{\tau}_\theta \circ ( \bm{\tau}^c\circ(\cdots (\bm{\tau}^c \circ (\bm{\tau}^c \circ \mathbf{M}_1 + \mathbf{M}_2) + \mathbf{M}_3) 
	+ \cdots + \mathbf{M}_{\ell^\prime-1}) + \mathbf{M}_{\ell^\prime} ),
\end{eqnarray}
where $\circ$ and $+$ denote SIMD (i.e. component wise) field multiplication and field addition of two $c$-dimensional vectors respectively. 

The computation in~\eqref{eqn-vec-poly} requires a total of $\ell^\prime$ SIMD field multiplications and $\ell^\prime-1$ SIMD additions. Of the
$\ell^\prime$ SIMD field multiplications, $\ell^\prime-1$ SIMD field multiplications have $\bm{\tau}^c$ as one of the operands, while one SIMD field multiplication
has $\bm{\tau}_\theta$ as one of the operands. 

As mentioned in Section~\ref{sec-prelim}, a delayed (or lazy) reduction strategy was used originally in~\cite{DBLP:conf/ches/BernsteinS12} and developed 
in~\cite{Goll-Gueron} to decrease the number of reductions. This strategy was used in a number of follow-up 
works~\cite{BNS2025,DBLP:conf/sp/DegabrieleGGP24,DBLP:journals/tches/PegolottiGPP26}. The technique is important and provides significant
speed improvement (see Remark~5.1 in~\cite{BNS2025} for a detailed explanation of the speed improvement). 
We provide a detailed mathematical description of the technique, which includes the exact partitioning
technique ($g+1$ in the first group, $g$ in the second to last but one group, and the remaining in the last group), the recurrence, and the
handling of the last group. Such a description is not present 
in other works -- for example,~\cite{DBLP:journals/tches/PegolottiGPP26} illustrates the technique using a 12-block example and group size 3. 
Also, we provide precise counts 
of the number of SIMD operations required by delayed reduction (see Table~\ref{tab-simd-oc}), and the number of bytes required to store the key material
(see Table~\ref{tab-poly-avx-storage}); such details are not available in previous works. 

Let $g\geq 1$ be a parameter, and $\ell^{\prime\prime}$ and $r$ be integers such that $(\ell^\prime-g-1)=g(\ell^{\prime\prime}-1)+r$, with $\ell^{\prime\prime}\geq 1$ 
and $1\leq r\leq g$. So $\ell^\prime-1=g\ell^{\prime\prime}+r$, and since $1\leq r\leq g$, we have $\ell^{\prime\prime}=\lceil (\ell^\prime-1)/g\rceil-1$.  Let
$\bm{\gamma}=\bm{\tau}^c$. 
\begin{eqnarray*}
	\mathbf{A}_1 & = & 
	\left\{\begin{array}{ll}
		\mathbf{M}_1\circ \bm{\gamma}^{\ell^\prime-1} 
			+ \mathbf{M}_2\circ \bm{\gamma}^{\ell^\prime-2} + \cdots + \mathbf{M}_{\ell^\prime-1}\circ \bm{\gamma} + \mathbf{M}_{\ell^\prime} 
					& \mbox{if } \ell^\prime \leq g+1, \\
		\mathbf{M}_1\circ \bm{\gamma}^g + \mathbf{M}_2\circ \bm{\gamma}^{g-1} + \cdots + \mathbf{M}_g\circ \bm{\gamma} + \mathbf{M}_{g+1} & \mbox{if } \ell^\prime > g+1.
	\end{array} \right.
\end{eqnarray*}
For $i=1\ldots, \ell^{\prime\prime}-1$, let
\begin{eqnarray*}
	\mathbf{A}_{i+1} & = & \mathbf{A}_i\circ \bm{\gamma}^g + \mathbf{M}_{ig+2}\circ \bm{\gamma}^{g-1}+\cdots+\mathbf{M}_{(i+1)g}\circ \bm{\gamma} + \mathbf{M}_{(i+1)g+1},
\end{eqnarray*}
and 
\begin{eqnarray*}
	\mathbf{A}_{\ell^{\prime\prime}+1}
	& = & \mathbf{A}_{\ell^{\prime\prime}}\circ \bm{\gamma}^r 
	+ \mathbf{M}_{\ell^{\prime\prime}g+2}\circ \bm{\gamma}^{r-1} + \mathbf{M}_{\ell^{\prime\prime}g+3}\circ \bm{\gamma}^{r-2} + \cdots
	+ \mathbf{M}_{\ell^{\prime\prime}g+r}\circ \bm{\gamma} + \mathbf{M}_{\ell^{\prime\prime}g+r+1}.
\end{eqnarray*}
Then it is not difficult to verify that
\begin{eqnarray} \label{eqn-C-comp}
	\mathbf{C} & = & 
	\left\{\begin{array}{ll}	
		\bm{\tau}_{\theta}\circ \mathbf{A}_{1} & \mbox{if } \ell^\prime\leq g+1, \\
		\bm{\tau}_{\theta}\circ \mathbf{A}_{\ell^{\prime\prime}+1} & \mbox{if } \ell^\prime>g+1.
	\end{array} \right.
\end{eqnarray}
Suppose that the key powers $\bm{\gamma},\bm{\gamma}^2,\bm{\gamma}^3,\ldots,\bm{\gamma}^g$ are pre-computed. 
If $\ell^\prime=1$, then $\mathbf{A}_1=\mathbf{M}_1$ and no SIMD multiplication or SIMD reduction are required to compute $\mathbf{A}_1$.
If $1<\ell^\prime\leq g+1$, then 
computing $\mathbf{A}_1$ requires $\ell^\prime-1$ unreduced SIMD multiplications and one SIMD reduction.
If $\ell^\prime>g+1$, then computing $\mathbf{A}_1$ requires $g$ unreduced SIMD multiplications and one SIMD reduction; 
for $1\leq i\leq \ell^{\prime\prime}-1$, computing $\mathbf{A}_{i+1}$ from $\mathbf{A}_i$, $1\leq i\leq \ell^{\prime\prime}-1$, requires $g$ unreduced SIMD multiplications and
one SIMD reduction; and computing $\mathbf{A}_{\ell^{\prime\prime}+1}$ from $\mathbf{A}_{\ell^{\prime\prime}}$ requires $r$ unreduced SIMD multiplications and one reduction. So 
if $\ell^\prime>g+1$, then to compute $\mathbf{A}_{\ell^{\prime\prime}+1}$ the total number of unreduced SIMD multiplications required is equal to 
$\ell^{\prime\prime}g+r=\ell^\prime-1$, and the total number of SIMD reductions is equal to $\ell^{\prime\prime}+1=\lceil(\ell^\prime-1)/g\rceil$. In fact, for all values of
$\ell^\prime\geq 1$, the total number of unreduced SIMD multiplications required is $\ell^\prime-1$, and the total number of SIMD reductions is
$\lceil(\ell^\prime-1)/g\rceil$.

The computation of $\mathbf{C}$ in~\eqref{eqn-C-comp} from $\mathbf{A}_1$ or $\mathbf{A}_{\ell^{\prime\prime}+1}$ requires one unreduced SIMD multiplication and one SIMD reduction.
However, the SIMD reduction can be avoided. 
The unreduced SIMD multiplication by $\bm{\tau}_{\theta}$ results in $\sym{unreduced}(\mathbf{C})$. The four components $\sym{unreduced}(C_i)$, $i=1,\ldots,4$,
of $\sym{unreduced}(\mathbf{C})$ are added together and a single reduction applied to the sum using {\tt maax} operations. So for the computation of $P$,
the total number of unreduced SIMD multiplications is equal to $\ell^\prime$, and the total number of SIMD reductions is equal to $\lceil(\ell^\prime-1)/g\rceil$.

In our implementations, we have taken $c=4$, i.e. we have made 4-way SIMD implementation using {\tt avx2} instructions. 
Given the delayed reduction parameter $g$, the key power vectors $\bm{\tau}^{4},\bm{\tau}^8,\ldots,\bm{\tau}^{4g}$ are required. Additionally, the
vector $\bm{\tau}_\theta$ is required.
To compute $\bm{\tau}_\theta$ and the key power vectors $\bm{\tau}^{4i}$, $i=1,\ldots,g$, the key powers
$\tau,\tau^2,\tau^3,\tau^4,\tau^8,\tau^{12},\ldots,\tau^{4g}$ are first computed using {\tt maax} instructions 
and then appropriately organised into the required key power vectors.
The actual SIMD computation of $\mathbf{C}$ starts after the required key power vectors have been computed. The number of 
unreduced SIMD multiplications and the number of SIMD reductions for the computation of $C$ are shown in Table~\ref{tab-simd-oc}. 
Additionally, there are a small number of non-SIMD operations required at the end to compute $P$ and to compute 
$\tau\cdot \sym{Poly}(\tau;M_1,\ldots,M_\ell)$ from $P$. These operations are implemented using {\tt maax} instructions and their counts are not shown 
in Table~\ref{tab-simd-oc}.

\begin{table}
\centering
{\scriptsize
\begin{tabular}{|l|c|c|}
	\cline{2-3}
	\multicolumn{1}{c|}{} & \begin{tabular}{c} unred mult \\ (SIMD) \end{tabular}  & \begin{tabular}{c} red \\ (SIMD) \end{tabular} \\ \hline
		$\sym{polyHash}$ & $\floor{\ell/4}$ & $\lceil(\floor{\ell/4}-1)/g\rceil$ \\ \hline
	$4\mbox{-}\sym{decBRWHash}$ & $\lfloor \lceil\ell/4\rceil/2\rfloor$ & $1+\lfloor \lceil\ell/4\rceil/4\rfloor$ \\ \hline
\end{tabular}
	\caption{Counts of 4-way SIMD operations required for computing $P$ (as part of the computation of $\sym{polyHash}$) and $\sym{decBRWHash}$ for 
	$\ell$ message blocks. Full computations of $\sym{polyHash}$ and $\sym{decBRWHash}$ require additionally a small number of {\tt maax} operations
	which are not shown in the table. For $\sym{polyHash}$, the parameter $g$ is a positive integer. \label{tab-simd-oc} }
}
\end{table}


As explained after~\eqref{eqn-4SIMD} if the number of blocks is not a multiple of 4, then the last few blocks (between 1 and 3) need to be 
tackled sequentially. As a result, the entire computation cannot proceed uniformly as a 4-way SIMD computation. It is possible to pre-pend a number of zero
blocks, so that the total number of blocks becomes a multiple of 4, and the 4-way SIMD can be employed for the entire message. Such pre-pending does not
alter the hash function and is only an implementation issue. This technique was used in~\cite{DBLP:journals/iet-ifs/BhattacharyyaS20}.
There is, however, an efficiency issue which does not combine well with the technique of pre-pending zero blocks.
If the message is stored as 32-byte aligned data, then the instruction {\tt vmovdqa} (i.e. aligned move) can be used to read the data. Such aligned read is
faster than the unaligned read {\tt vmovdqu}. If the zero pre-pending technique is not used, then with successive aligned moves successive 32 bytes of data can be read.
On the other hand, if zero-prepending is used, then this is no longer possible. As a result, the reading of the data becomes slower. We have implemented 
both the technique of zero pre-pending without support of aligned moves, and not using zero pre-pending which supports aligned moves. We find no significant
difference in the timings. 

\subsection{Vectorised Computation of $c\mbox{-}\sym{decBRWHash}$ \label{subsec-vec-BRW}}
The computation of $c\mbox{-}\sym{decBRWHash}_\tau(X)$ requires the evaluation of $Q_1(\tau),Q_2(\tau),\ldots,Q_c(\tau)$.
Evaluations of the $Q_i(\tau)$'s require evaluations of $c$ BRW polynomials, where the inputs to the BRW polynomials are disjoint and the number of 
message blocks provided as input is the same for all the BRW polynomials. Consequently, the $c$ BRW polynomials can be simultaneously evaluated at $\tau$ using
an SIMD strategy. 

For $i=1,\ldots,\mathfrak{n}=\ceil{\ell/c}$, let $\mathbf{M}_i=(M_{(i-1)c+1},M_{(i-1)c+2},\ldots,M_{ic})$, where the message blocks $M_1,M_2,\ldots$
are as defined in the description of $c\mbox{-}\sym{decBRWHash}$. Algorithm~\ref{algo-B} in Appendix~\ref{sec-BRW-algo}
can be used to compute $c$ simultaneous BRW polynomials, i.e. Algorithm~\ref{algo-B} can be used in a $c$-way SIMD manner, where the input consists
of the sequence of $c$-way vectors $\mathbf{M}_1,\mathbf{M}_2,\ldots,\mathbf{M}_{\mathfrak{n}}$. The key vector is the vector $\bm{\tau}=(\tau,\ldots,\tau)$
of length $c$. The key power vectors required are $\bm{\tau}^{2^j}$, for $j=0,\ldots,\floor{\lg\mathfrak{n}}$. The key power
$\tau^{2^j}$ are computed from which the key power vector $\bm{\tau}^{2^j}$ is computed.
The algorithm uses $\sym{keyPow}$, $\sym{stack}$ and $\sym{tmp}$ as internal arrays and variables. The SIMD version of the algorithm
uses vector versions of these arrays and variables. In particular, $\sym{keyPow}[j]$ stores the vector $\bm{\tau}^{2^j}$, $\sym{stack}[j]$ stores a vector
of length $c$, and $\sym{tmp}$ is a vector of length $c$, where the $i$-th components of $\sym{stack}[j]$ and $\sym{tmp}$
correspond to the computation of the $i$-th BRW polynomial. All unreduced multiplications, reductions and unreduced additions are done component wise on the
vectors. With this strategy Algorithm~\ref{algo-B} becomes an SIMD algorithm for the simultaneous computation of the $c$ BRW polynomials. 

We have implemented the SIMD version of Algorithm~\ref{algo-B} for $c=4$. Table~\ref{tab-simd-oc} provides the number of 4-way SIMD operations required for computing
$4\mbox{-}\sym{decBRWHash}$ using the SIMD version of Algorithm~\ref{algo-B}. In addition to these SIMD operations, there are a small number of other operations which
are required to compute $Q_{5}$ from $Q_1,Q_2,Q_3,Q_4$ and to compute $Q$ from $Q_5$ (see Section~\ref{sec-SIMD-BRW-Horner}). These are implemented using {\tt maax}
operations, and the count of these operations are not shown in Table~\ref{tab-simd-oc}.

\begin{remark}\label{rem-8-decBRW}
	Taking $c=4$ gives $4\mbox{-}\sym{decBRWHash}$. Our implementation of $4\mbox{-}\sym{decBRWHash}$ was made using the {\tt avx2} instruction set 
	of Intel. Taking $c=8$, we obtain $8\mbox{-}\sym{decBRWHash}$. It is possible to
	implement $8\mbox{-}\sym{decBRWHash}$ using either the {\tt avx512} or the {\tt avx512\_ifma} instruction sets of Intel.
	Following the discussion in Section~\ref{subsec-divides}, $4\mbox{-}\sym{decBRWHash}$ can be implemented using the 64-bit {\tt maax} instructions,
	and $8\mbox{-}\sym{decBRWHash}$ can be implemented using both the {\tt avx2} instruction set, and the 64-bit {\tt maax} instructions.
	We leave these as possible future works.
\end{remark}

\subsection{Comparison between $\sym{polyHash}$ and $4\mbox{-}\sym{decBRWHash}$ \label{subsec-comp-poly-brw}}
We make a comparison between the 4-way SIMD computation of $\sym{polyHash}$ and $4\mbox{-}\sym{decBRWHash}$. There are two aspects to the
comparison, the efficiency and the storage requirement. 

\paragraph{Efficiency.} The SIMD operation counts for $\sym{polyHash}$ and $4\mbox{-}\sym{decBRWHash}$ are shown in Table~\ref{tab-simd-oc}. The number
of unreduced SIMD multiplications required for $4\mbox{-}\sym{decBRWHash}$ is about half of what is required for $\sym{polyHash}$. The number
of SIMD reductions required by $4\mbox{-}\sym{decBRWHash}$ is at most that required by $\sym{polyHash}$ for $g\leq 4$, while for $g>4$, 
$\sym{polyHash}$ requires less number of SIMD reductions. In our implementations, we found that taking $g>4$ does not provide speed improvement. 
The halving of the number of unreduced SIMD multiplications indicates that
$4\mbox{-}\sym{decBRWHash}$ should be substantially faster than $\sym{polyHash}$ for all values of $g$. Experimental results for the primes
$2^{130}-5$ and $2^{127}-1$, however, show that for {\tt avx2} implementations while there is a noticeable speed-up of $4\mbox{-}\sym{decBRWHash}$ over $\sym{polyHash}$, the 
actual speed-up obtained is less than what is theoretically predicted by operation counts. 
We explain the reasons for such an observation.

Consider the prime $2^{130}-5$. (A similar reasoning applies for the 5-limb representation based on the prime $2^{127}-1$.)
For the 4-way SIMD computation of $\sym{polyHash}1305$, the vectors of key powers $\bm{\tau}^4,\bm{\tau}^{8},\ldots,\bm{\tau}^{4g}$ are fixed. 
All the 4-way multiplications involved in the 4-way computation of $\sym{polyHash}1305$ have some key power vector $\bm{\tau}^{4i}$ as one of the arguments.
This has two effects.

First, as noted in Remark~\ref{rem-5R} when the operand $e$ to the multiplication is fixed, it is possible to pre-compute $\widetilde{e}$ so that the
4 multiplications by 5 are not required. Extending this to SIMD operations, for each key power $\bm{\tau}^{4i}$, we compute the corresponding
$\widetilde{\bm{\tau}^{4i}}$, and can avoid the SIMD multiplications by 5. Also, for the vector $\bm{\tau}_{\theta}$ we compute
$\widetilde{\bm{\tau}_{\theta}}$, so that the SIMD multiplications by 5 can be avoided while multiplying by $\bm{\tau}_\theta$. 
As a result the number of {\tt vpmuludq} instructions required for an unreduced
multiplication is 25 in the case of SIMD computation of $\sym{polyHash}$. In contrast, in the computation of $4\mbox{-}\sym{decBRWHash}$, none of the multiplications
are by a fixed element. So the number of {\tt vpmuludq} instructions required for an unreduced multiplication is 29 in the case of SIMD computation of
$4\mbox{-}\sym{decBRWHash}$. 

The second effect of multiplying with a fixed element is more generic. When a fixed element is used for repeated multiplication, during actual execution
this element is kept either on chip or in the cache memory. This significantly reduces the time for reading the element from the memory and leads to
a significant increase in speed which is not explained by simply counting the number of arithmetic operations. For the {\tt maax} implementation also,
a similar speed-up was observed and explained in details in~\cite{BNS2025}. 

As a combined result of the above two effects, the speed improvement of SIMD computation of $4\mbox{-}\sym{decBRWHash}$ over SIMD computation of 
$\sym{polyHash}$ is less than what is theoretically predicted, though a significant improvement in the speed is observed for messages which are
a few kilo bytes or longer. We provide detailed timing results in Section~\ref{sec-times}.




\paragraph{Storage.} Consider the prime $2^{130}-5$ where elements are represented using 5 limbs. 
The storage requirement based on using the 5-limb representation of the prime $2^{127}-1$ is the same as the storage requirement for the prime $2^{130}-5$. 

The 4-way SIMD computation of $\sym{polyHash}$ requires the pre-computation and storage of the key power vectors $\bm{\tau}^4,\bm{\tau}^8,\ldots,\bm{\tau}^{4g}$ as well
as the vector $\bm{\tau}_{\theta}$.  Additionally, to avoid the multiplications by 5, it is also required to store the associated vectors 
$\widetilde{\bm{\tau}^4},\widetilde{\bm{\tau}^8},\ldots,\widetilde{\bm{\tau}^{4g}}$. The associated vector $\widetilde{\bm{\tau}_\theta}$ is also required, but
it is required only at the end, and so is not pre-computed and carried forward. 
Each $\bm{\tau}^{4i}$ and also $\bm{\tau}_{\theta}$ requires 5 256-bit words to be stored, i.e. a total of 160 bytes. The associated vector $\widetilde{\bm{\tau}^{4i}}$ 
requires 4 256-bit words to be stored, i.e. a total of 128 bytes. So the total number of bytes required for storing all the key power vectors is 
$g\cdot(160+128)+160=160+288g$. Concrete values of the storage requirement for various values of $g$ are shown in Table~\ref{tab-poly-avx-storage}.

\begin{table}
	\begin{subtable}{1.0\textwidth}
		{\scriptsize
\centering
	\begin{tabular}{|l|c|c|c|c|}
		\cline{2-5}
		\multicolumn{1}{c|}{} & $g=1$ & $g=2$ & $g=3$ & $g=4$ \\ \hline
		\# bytes & 448 & 736 & 1024 & 1312 \\ \hline
	\end{tabular}
	\caption{For various values of $g$, the number of bytes of key material required to be stored for computing either $\sym{polyHash}1305$ or 
	$\sym{polyHash}1271$ using 4-way SIMD. \label{tab-poly-avx-storage}}
		}
\end{subtable}
\begin{subtable}{1.0\textwidth}
	{\scriptsize
	\centering

	\begin{tabular}{|l|c|c|c|c|c|c|c|}
		\cline{2-8}
		\multicolumn{1}{c|}{} & \multicolumn{7}{|c|}{\# of blocks $\ell$} \\ \cline{2-8}
		\multicolumn{1}{c|}{} & 1-4 & 5-12 & 13-28 & 29-60 & 61-124 & 125-252 & 253-508  \\ \hline
			\# bytes      & 0   & 160  & 320   & 480   & 640    & 800     & 960 \\ \hline

\end{tabular}
	\caption{The number of bytes of key material required to be stored for computing either $\sym{decBRWHash}1305$ or $\sym{decBRWHash}1271$. \label{tab-decBRWHash-storage} }
	}
\end{subtable}
	{\scriptsize \caption{Key storage requirements. \label{tab-st}} }
\end{table}

For computing $4\mbox{-}\sym{decBRWHash}$, the key power vectors $\bm{\tau}^{2^j}$, $j=0,\ldots, \lfloor\lg\mathfrak{n} \rfloor=\lfloor\lg\lceil \ell/4 \rceil\rfloor$
are required. Each of these key power vector requires $5\cdot 32=160$ bytes to be stored, so for an $\ell$-block message, the total number of bytes required to store all 
the key power vectors is $160\cdot \lfloor\lg\lceil \ell/4 \rceil\rfloor$. The key power vector $\bm{\tau}^d$ is required at the end to combine the four
independent $\sym{BRW}$ computations. Since this vector is required only at the end, it is not pre-computed. Concrete values of the storage requirement
for various values of $\ell$ are shown in Table~\ref{tab-decBRWHash-storage}.

\subsection{Efficiency Trade-Off Between $2^{130}-5$ and $2^{127}-1$ \label{subsec-prime-comp}}
There is a generic disadvantage of $2^{127}-1$ in comparison to $2^{130}-5$. 
From Table~\ref{tab-params}, the block size $n$ is 120 for $2^{127}-1$ and 128 for $2^{130}-5$. So for any given message, the number of blocks
for $2^{127}-1$ will be about 16/15 times the number of blocks for $2^{130}-5$. Being required to process more blocks, suggests that hash functions
based on $2^{127}-1$ will be slower than the corresponding hash functions based on $2^{130}-5$. The other aspect to consider is the speed of an individual
multiplication. If the individual multiplication for $2^{127}-1$ is faster than the individual multiplication for $2^{130}-5$, then this may compensate
the requirement of processing more blocks. For {\tt maax} implementations, it is indeed the case that the individual multiplication for $2^{127}-1$
is substantially faster than the individual multiplication for $2^{130}-5$, resulting in the hash functions based on $2^{127}-1$ being faster than the
hash functions based on $2^{130}-5$. See~\cite{BNS2025,cryptoeprint:2025/1224} for the timing results of {\tt maax} implementation which supports this statement.
However, for {\tt avx2} implementation, an individual multiplication for $2^{127}-1$ has efficiency similar to an individual multiplication
for $2^{130}-5$ (see Section~\ref{sec-fld-arith}). As a result, due to the requirement of processing more blocks, for {\tt avx2} implementations, the hash functions 
based on $2^{127}-1$ are slower than the corresponding hash functions based on $2^{130}-5$. 

\begin{remark}\label{rem-read}
In the context of {\tt avx2} implementation, there is another disadvantage for $2^{127}-1$. For vectorised processing, it is advantageous to use 16-byte block sizes 
(or block sizes which are multiples of 16 bytes). With a 16-byte block size, it is possible to use two {\tt vmovdqa(u)} instructions to read 512 bits of the input into 
two {\tt ymm} registers. The first {\tt vmovdqa(u)} instruction reads 32 bytes which brings two input blocks into an {\tt ymm} register, and so does the second. If 
the block size is not 16 bytes, 
then such a smooth read operation will not be possible. Reading the input and allocating it to two {\tt ymm} registers will be more
complicated and hence will require more operations. Since the block size for $2^{127}-1$ is 15 bytes, while the block size for $2^{130}-5$ is 16 bytes, reading the message
bytes and allocating to {\tt ymm} registers require more operations for $2^{127}-1$ than for $2^{130}-5$.
\end{remark}
Primes other than $2^{127}-1$ and $2^{130}-5$ were considered in~\cite{DBLP:conf/sp/DegabrieleGGP24,DBLP:journals/tches/PegolottiGPP26};
see Section~\ref{subsec-oth-prm}. SIMD implementation based on any prime for which the block size is not a power of 2 will 
suffer from the disadvantage mentioned in Remark~\ref{rem-read}. 
For {\tt avx2} implementations, the performance model proposed in~\cite{DBLP:journals/tches/PegolottiGPP26} includes the cost of loading data, but does 
not consider the effect of Remark~\ref{rem-read}.

\subsection{Other Primes \label{subsec-oth-prm}}
Several possible primes were considered in~\cite{DBLP:conf/sp/DegabrieleGGP24,DBLP:journals/tches/PegolottiGPP26} for different security
levels. For about the 128-bit security level, as alternatives to $2^{130}-5$, the primes $2^{116}-3$, $2^{122}-3$, and $2^{150}-3$ were considered
in~\cite{DBLP:conf/sp/DegabrieleGGP24}. On the other hand,
in~\cite{DBLP:journals/tches/PegolottiGPP26} all primes of the form $2^m-\delta$, with $100\leq m\leq 400$ and the smallest admissible $\delta$ were 
considered as part of an automated search. For about the 128-bit security level, individual primes other than $2^{130}-5$ that were
mentioned in the performance analysis in~\cite[Table~2]{DBLP:journals/tches/PegolottiGPP26} are $2^{116}-3$, and $2^{150}-3$.

For about the 128-bit security level, apart from $2^{127}-1$ and $2^{130}-5$, we considered the following other primes: 
$2^{116}-3$, $2^{122}-3$, $2^{137}-13$, $2^{140}-27$, $2^{141}-9$, and $2^{150}-3$. 
Based on our analysis of implementation using 4-way SIMD we provide the reasons for not implementing hash functions for these primes.
\begin{compactenum}
\item 
For the primes $2^{137}-13$, $2^{140}-27$, and $2^{141}-9$ the elements of the corresponding fields can be represented using 5 limbs. 
These primes would support 17-byte blocks instead of the 16-byte blocks supported by the prime $2^{130}-5$, and also provide a few extra bits of security. 
However, with these three primes there will not be sufficiently many free bits left after multiplication to support delayed reduction. So the overall hash 
computation will be slower than that for $2^{130}-5$. 
Further, while using 17-byte blocks instead of 16-byte blocks will reduce the number of blocks, this advantage will be more than offset due to
the problem of reading the input (see Remark~\ref{rem-read}). 
\item One may use an even greater prime such as $2^{150}-3$; the problem in this case is that a 5-limb representation will cause overflow during the 
	multiplication procedure, while a 6-limb representation will be slower than $2^{130}-5$ (though it will support delayed reduction). 
	The block size will be 18 bytes, which will reduce the number of blocks, but also run into the problem of inefficiency of reading the
	input (see Remark~\ref{rem-read}).
\item The prime $2^{122}-3$ is smaller than $2^{127}-1$. 
With $2^{122}-3$, a 4-limb representation will lead to an overflow in the first limb of the product, and a 5-limb representation gains nothing over $2^{127}-1$. 
	The block size will be 15 bytes, and will share the reading inefficiency problem of $2^{127}-1$ (see Remark~\ref{rem-read}).
\end{compactenum}

\paragraph{The case of the prime $2^{116}-3$.}
The hash function Poly1163 based on the prime $2^{116}-3$ was put forward as an alternative to Poly1305 in~\cite{DBLP:conf/sp/DegabrieleGGP24}. 
The argument provided in~\cite{DBLP:conf/sp/DegabrieleGGP24} in favour of Poly1163 is that it is faster than Poly1305, 
and at the same time provides security comparable to Poly1305. 
In the terminology of~\cite{BNS2025}, the hash function Poly1163 would correspond to $\sym{polyHash}1163$.
We compare $\sym{polyHash}1163$ with $\sym{polyHash}1271$ and $\sym{polyHash}1305$ in details.

We first consider the security argument. The argument that Poly1163 provides comparable security to Poly1305 crucially depends on the fact 
that Poly1305 uses clamping to reduce a 128-bit key to 
a 106-bit key. This argument, however, does not provide the complete picture. The key clamping procedure of Poly1305 is not a security requirement. Key
clamping was introduced in~\cite{DBLP:conf/fse/Bernstein05} only to support high speed floating point computation. One can build a hash function based on the 
prime $2^{130}-5$ with 128-bit keys. 
The function $\sym{polyHash}1305$ that we introduced is in fact based on the prime $2^{130}-5$ and uses 128-bit keys. In particular,
we do not propose any key clamping for $\sym{polyHash}1305$. 
The key and tag for Poly1163 are both 14 bytes long, while the key and tag for $\sym{polyHash}1305$ are 16 bytes long (see Table~\ref{tab-params}). 
Applying the AXU bound for $\sym{polyHash}$ (see Theorem~\ref{thm-AXUbndPolyHash}), the loss in security of Poly1163 in
comparison to $\sym{polyHash}1305$ is 18 bits. Similarly, the loss in security of Poly1163 in comparison to $\sym{polyHash}1271$ is 17 bits
(using Table~\ref{tab-params} and Theorem~\ref{thm-AXUbndPolyHash}).

Next we consider the speed argument for different implementations. The issue which is common for all the implementations is that the block sizes for
$\sym{polyHash}1163$, $\sym{polyHash}1271$, and $\sym{polyHash}1305$ are 14-byte, 15-byte, and 16-byte respectively. Correspondingly, for a fixed 
message length, the number of blocks will be the smallest for $\sym{polyHash}1305$, largest for $\sym{polyHash}1163$, with $\sym{polyHash}1271$ in the middle. 
\begin{compactdesc}
\item {64-bit {\tt maax} implementation.} 
	For $\sym{polyHash}1163$, $\sym{polyHash}1271$, and $\sym{polyHash}1305$ the limb representations are 2-limb, 2-limb, and 3-limb respectively. 
	So individual multiplications for
	$\sym{polyHash}1163$ and $\sym{polyHash}1271$ will have the same speed, which will be faster than an individual multiplication
	for $\sym{polyHash}1305$. Overall efficiency depends on the relative trade-off between the decrease in the time for an individual multiplication
	versus the increase in the number of blocks. This trade-off favours $\sym{polyHash}1271$ over 
	$\sym{polyHash}1305$ as reported in~\cite{BNS2025}. Whether this trade-off also favours $\sym{polyHash}1163$ over $\sym{polyHash}1305$ is not known, 
	since~\cite{DBLP:conf/sp/DegabrieleGGP24} did not perform a comparison of a {\tt maax} implementation of Poly1163 with the
	{\tt maax} implementation of $\sym{polyHash}1305$ in~\cite{BNS2025}.
\item {{\tt avx2} implementation.} The limb representations for $\sym{polyHash}1163$, $\sym{polyHash}1271$, and
        $\sym{polyHash}1305$ are 4-limb, 5-limb, and 5-limb respectively. So individual multiplications for $\sym{polyHash}1271$ and $\sym{polyHash}1305$
	will have the same speed which will be slower than an individual multiplication for $\sym{polyHash}1163$. Overall efficiency again
	depends on the relative trade-off between the decrease in the time for an individual multiplication versus the increase in the number of blocks.
	Further, there is the issue of inefficiency in loading registers from memory when the block size is not a power of 2 (see Remark~\ref{rem-read}),
	which is a disadvantage for both $\sym{polyHash}1163$, $\sym{polyHash}1271$ in comparison to $\sym{polyHash}1305$. 
	Our {\tt avx2} implementations show that
	the overall trade-off favours $\sym{polyHash}1305$ over $\sym{polyHash}1271$. Timing results reported in~\cite{pegolotti2026poly1305benchmarks} 
	for {\tt avx2} implementations of Poly1163 and Poly1305 show that the trade-off is in favour of Poly1163. However, the speed
	reported for {\tt avx2} implementation of Poly1163 in~\cite{pegolotti2026poly1305benchmarks} on the Intel Sapphire Rapids processor
	is about 0.45 cycles/byte for 22456 bytes, while for a similar number of bytes, the speed of our {\tt avx2} implementation of $\sym{polyHash}1305$ 
	on the Intel Tiger Lake processor (an earlier generation processor, see Section~\ref{sec-times} for a comparison of 
	Sapphire Rapids and Tiger Lake processors) is between 0.43--0.45 cycles/byte
	(see Table~\ref{tab-comparison-p1305-long} in Appendix~\ref{sec-detailed-timings}). So at present
	there is no definitive evidence that Poly1163 will necessarily be faster than $\sym{polyHash}1305$.
\item {{\tt avx512} implementation.} The comparison between $\sym{polyHash}1163$, $\sym{polyHash}1271$, and $\sym{polyHash}1305$ is the same as that 
	for {\tt avx2} implementations.
\item {{\tt avx512\_ifma} implementation.} With {\tt avx512\_ifma}, 4 simultaneous full 52-bit multiplications can be performed (using two instructions, 
	rather than a single instruction). All three of $\sym{polyHash}1163$, $\sym{polyHash}1271$, and $\sym{polyHash}1305$ will have 3-limb 
	representations. So the speed of individual multiplication will be the same in all the three cases. Since $\sym{polyHash}1305$ has
	the least number of blocks, and has the block size to be a power of 2 (see Remark~\ref{rem-read}), it follows that $\sym{polyHash}1305$ 
	will be fastest of the three hash functions. Between $\sym{polyHash}1163$, $\sym{polyHash}1271$, since $\sym{polyHash}1271$ will have
	less number of blocks, it will be faster than $\sym{polyHash}1163$.
\end{compactdesc}
The overall summary of the efficiency comparison is the following.
For {\tt maax} implementation, $\sym{polyHash}1271$ will be faster than both $\sym{polyHash}1163$ and $\sym{polyHash}1305$, while the comparative
speed of $\sym{polyHash}1163$ and $\sym{polyHash}1305$ is not clear. 
For {\tt avx2} and {\tt avx512} implementations, $\sym{polyHash}1305$ will be faster than $\sym{polyHash}1271$, but the comparative speed of
$\sym{polyHash}1163$ and $\sym{polyHash}1305$ is not clear.
For {\tt avx512\_ifma} implementations, $\sym{polyHash}1305$ will be the fastest, $\sym{polyHash}1163$ will be the slowest, with $\sym{polyHash}1271$
somewhere in the middle. 
Considering the efficiency comparison together with the 17-bit and 18-bit loss of security with respect to $\sym{polyHash}1271$ and $\sym{polyHash}1305$
respectively, in our opinion, the security/efficiency trade-off does not favour $\sym{polyHash}1163$ over either $\sym{polyHash}1271$
or $\sym{polyHash}1305$.

\section{Implementation and Timing Results \label{sec-times}}
We have made {\tt avx2} implementations of $\sym{polyHash}$ and $4\mbox{-}\sym{decBRWHash}$ for both the primes $2^{130}-5$ and $2^{127}-1$ in hand optimised
assembly language. For the $\sym{polyHash}$ implementations we considered the delayed reduction parameter $g$ to take the values $1,2,3$ and $4$. Using higher values
of $g$ leads to a loss of speed. For the $4\mbox{-}\sym{decBRWHash}$ implementations we considered the parameter $t$ in Algorithm~\ref{algo-B} to take the 
values $t=2,3,4$ and $5$. Higher values of $t$ would lead to a very large code size (since $2^t-1$ fragments of straight line code are required to implement
Step~\ref{B19}, as $r$ can take $2^t-1$ values).

Recall that if the key clamping mandated by Poly1305 is implemented, then $\sym{polyHash}1305$ is the same as Poly1305 for messages whose lengths are 
multiples of 8.
Since $\sym{Poly}1305$ is an extensively used hash function, different implementations of the hash function targeted at various platforms are publicly 
available; see for 
example~\cite{DBLP:conf/ches/BernsteinS12,Goll-Gueron,DBLP:conf/sp/Delignat-Lavaud17,DBLP:conf/sp/AlmeidaBBGKL0S20,DBLP:journals/iet-ifs/BhattacharyyaS20,BNS2025,DBLP:conf/sp/DegabrieleGGP24,DBLP:journals/tches/PegolottiGPP26,moon-poly1305-opt,poly1305-x86-64-cryptogams,linux-poly1305-x86,lib1305}.
As far as we are aware, our implementations of $\sym{polyHash}1305$ provide the first hand optimised {\tt avx2} assembly language implementations
of Poly1305 with different values of the delayed reduction parameter $g$.

The codes for our implementations of the hash functions are available from the following links.
\begin{tabbing}
	\ \ \ \ \= \kill
	\> \url{https://github.com/kn-cs/dec-BRWHash} \\
	\> \url{https://github.com/kn-cs/vec-polyHash}
\end{tabbing}
All the implementations were verfied with test vectors of various sizes generated by separate SAGE implementations of the all hash functions.

We recorded an extensive set of timings for all the hash functions.
The timing measurements were taken on a single core of 11th Gen Intel Core i7-1185G7 @ 3.00GHz $\times$ 4 Tiger Lake processor using 31.1 GiB memory.
During the experiments, turbo boost and hyperthreading options were turned off. The OS was Ubuntu 20.04.3 LTS and the code was compiled using
gcc version 9.4.0. The following flags were used during compilation.
\begin{center}
	-march=native -mtune=native -m64 -O3 -funroll-loops -fomit-frame-pointer
\end{center}
We counted CPU cycles using the microlibrary ``libcpucycles'' (see \url{https://cpucycles.cr.yp.to/}) through the {\tt amd64-pmc} counter
(see \url{https://cpucycles.cr.yp.to/counters.html}) which requires a 64-bit Intel/AMD platform and Linux perf\_event interface. The {\tt amd64-pmc} counter
accesses a cycle counter through RDPMC and requires
\url{/proc/sys/kernel/perf\_event\_paranoid} to be at most 2 for user-level RDPMC access. This counter runs at clock frequency of the CPU core.

The timing results show that the {\tt avx2} implementation of $\sym{polyHash}1271$ is slower than the {\tt avx2} implementation of $\sym{polyHash}1305$,
and the {\tt avx2} implementation of $4\mbox{-}\sym{decBRWHash}1271$ is slower than the {\tt avx2} implementation of $4\mbox{-}\sym{decBRWHash}1305$. This confirms
the theoretically predicted slowdown discussed in Section~\ref{subsec-prime-comp}. In view of this slowdown, we do not present the timing results for
$2^{127}-1$.

For purpose of comparison, we provide timing results for $\sym{polyHash}1305$ and $\sym{BRWHash}1305$ from~\cite{BNS2025,cryptoeprint:2025/1224},
which are sequential {\tt maax} implementations. For $\sym{BRWHash}1305$, there are no other previous implementations. For {\tt avx2} implementations
of $\sym{polyHash}1305$, we compare to the measurements reported in~\cite{pegolotti2026poly1305benchmarks}
of the implementation in the recent work~\cite{DBLP:journals/tches/PegolottiGPP26}, and also to the measurements of the performance of Poly1305 code
in OpenSSL~\cite{poly1305-x86-64-cryptogams} and HACL$^*$~\cite{DBLP:conf/ccs/ZinzindohoueBPB17} reported in~\cite{pegolotti2026poly1305benchmarks}.

\pgfplotsset{
  every axis/.append style={grid=both, grid style={gray!25},
    tick label style={font=\footnotesize}, label style={font=\small},
    title style={font=\small}, legend style={font=\scriptsize},
    log ticks with fixed point},
  maax/.style={red!70!black, dashed, thick, mark=square*, mark size=1.3pt},
  avx/.style={blue!70!black, solid, thick, mark=*, mark size=1.3pt},
  gone/.style={black, thick, mark=triangle*, mark size=1.3pt},
  gtwo/.style={orange!85!black, thick, mark=diamond*, mark size=1.5pt},
  gthree/.style={green!50!black, thick, mark=square*, mark size=1.2pt},
  brw/.style={red!70!black, dashed, thick, mark=square*, mark size=1.3pt},
  poly/.style={blue!70!black, solid, thick, mark=*, mark size=1.3pt},
  refline/.style={black, dotted, thick, forget plot},
}

\newcommand{\bestMaaxPre}{min(min(\thisrow{mg1},\thisrow{mg4p}),min(\thisrow{mg8p},min(\thisrow{mg16p},\thisrow{mg32p})))}
\newcommand{\bestMaaxFly}{min(min(\thisrow{mg1},\thisrow{mg4f}),min(\thisrow{mg8f},min(\thisrow{mg16f},\thisrow{mg32f})))}
\newcommand{\bestAvxPre}{min(min(\thisrow{a1p},\thisrow{a2p}),min(\thisrow{a3p},\thisrow{a4p}))}
\newcommand{\bestAvxFly}{min(min(\thisrow{a1f},\thisrow{a2f}),min(\thisrow{a3f},\thisrow{a4f}))}

\pgfplotstableread{
n mg1 mg4p mg4f mg8p mg8f mg16p mg16f mg32p mg32f a1p a1f a2p a2f a3p a3f a4p a4f bm5p bm5f d2p d2f d3p d3f d4p d4f d5p d5f
1 2.31 2.38 2.38 2.38 2.38 2.25 2.25 2.31 2.31 2.38 2.75 2.38 2.81 2.44 2.75 2.38 2.75 2.12 2.12 11.19 13.31 11.25 13.31 11.25 13.31 11.25 13.31
2 1.69 1.84 1.84 1.84 1.84 1.84 1.84 1.84 1.84 1.50 2.12 1.50 2.12 1.50 2.12 1.50 2.12 1.22 1.22 5.62 6.66 5.62 6.66 5.62 6.66 5.62 6.66
3 1.56 1.48 1.88 1.48 1.88 1.48 1.88 1.48 1.88 1.23 2.33 1.23 2.31 1.23 2.31 1.25 2.31 1.02 1.38 3.73 4.44 3.73 4.44 3.75 4.44 3.73 4.44
4 1.52 1.38 2.05 1.28 1.94 1.28 1.94 1.28 1.92 1.55 3.27 1.55 3.28 1.55 3.25 1.52 3.25 1.08 1.62 2.80 3.33 2.81 3.33 2.80 3.33 2.81 3.33
5 1.49 1.40 2.14 1.15 1.94 1.15 1.94 1.15 1.94 1.56 2.94 1.56 2.95 1.52 2.92 1.49 2.90 0.91 1.34 3.17 4.08 3.16 4.06 3.16 4.08 3.16 4.08
6 1.47 1.23 1.86 1.06 1.95 1.06 1.95 1.06 1.98 1.41 2.55 1.42 2.55 1.41 2.53 1.39 2.52 0.80 1.18 2.65 3.40 2.64 3.39 2.64 3.40 2.64 3.40
7 1.46 1.16 1.71 1.01 1.98 1.01 1.98 1.01 1.99 1.30 2.29 1.30 2.29 1.30 2.27 1.28 2.26 0.80 1.12 2.27 2.91 2.26 2.90 2.26 2.91 2.26 2.91
8 1.45 1.15 1.63 1.03 2.09 0.97 2.03 0.97 2.02 1.04 1.99 1.05 2.02 1.06 2.03 1.07 2.03 0.88 1.38 1.98 2.55 1.98 2.55 1.98 2.55 1.98 2.55
9 1.47 1.19 1.60 1.09 2.11 0.94 2.00 0.94 2.01 1.12 1.94 1.12 1.96 1.13 1.97 1.15 1.98 0.81 1.24 1.83 2.32 1.83 2.33 1.82 2.33 1.82 2.32
10 1.43 1.12 1.49 1.02 1.94 0.91 1.98 0.91 1.99 1.06 1.82 1.07 1.84 1.08 1.84 1.08 1.85 0.76 1.14 1.64 2.09 1.64 2.09 1.64 2.09 1.64 2.09
11 1.45 1.07 1.42 0.98 1.83 0.88 2.01 0.88 2.01 1.03 1.71 1.03 1.72 1.04 1.73 1.05 1.74 0.76 1.10 1.49 1.90 1.49 1.90 1.49 1.90 1.49 1.90
12 1.47 1.07 1.40 0.96 1.73 0.87 2.04 0.86 2.03 0.88 1.52 0.81 1.69 0.93 1.78 0.94 1.81 0.80 1.12 1.37 1.74 1.37 1.74 1.36 1.74 1.37 1.74
13 1.48 1.11 1.39 0.94 1.65 0.85 2.06 0.85 2.06 0.95 1.52 0.89 1.68 0.99 1.77 1.00 1.79 0.75 1.05 1.62 1.96 1.61 1.95 1.61 1.95 1.60 1.95
14 1.47 1.06 1.33 0.91 1.57 0.84 2.08 0.84 2.08 0.92 1.46 0.87 1.61 0.96 1.69 0.96 1.71 0.72 0.99 1.50 1.82 1.49 1.81 1.49 1.81 1.48 1.81
15 1.47 1.03 1.28 0.90 1.51 0.82 2.10 0.82 2.10 0.91 1.41 0.86 1.54 0.94 1.62 0.94 1.64 0.73 0.97 1.40 1.70 1.39 1.69 1.39 1.69 1.38 1.69
16 1.46 1.03 1.27 0.92 1.50 0.85 2.15 0.82 2.12 0.80 1.28 0.76 1.41 0.79 1.62 0.80 1.62 0.78 1.12 1.31 1.59 1.30 1.59 1.30 1.58 1.30 1.58
17 1.40 1.06 1.28 0.95 1.49 0.89 2.16 0.81 2.08 0.86 1.30 0.81 1.41 0.79 1.53 0.85 1.62 0.74 1.06 1.26 1.53 1.26 1.52 1.25 1.52 1.25 1.51
18 1.40 1.03 1.24 0.92 1.43 0.86 2.06 0.80 2.10 0.84 1.26 0.81 1.38 0.77 1.48 0.83 1.57 0.72 1.01 1.19 1.44 1.19 1.44 1.18 1.43 1.18 1.43
19 1.41 1.00 1.21 0.90 1.39 0.85 1.99 0.79 2.11 0.84 1.23 0.79 1.33 0.77 1.44 0.82 1.51 0.71 1.00 1.12 1.37 1.12 1.36 1.12 1.36 1.12 1.36
20 1.41 1.01 1.20 0.90 1.36 0.84 1.92 0.79 2.12 0.75 1.14 0.68 1.19 0.69 1.34 0.70 1.49 0.75 1.02 1.07 1.30 1.07 1.29 1.06 1.29 1.06 1.29
21 1.41 1.03 1.21 0.93 1.32 0.85 1.88 0.79 2.12 0.80 1.15 0.74 1.21 0.74 1.35 0.72 1.45 0.72 0.98 1.09 1.31 1.10 1.31 1.10 1.31 1.10 1.31
22 1.41 1.05 1.18 0.87 1.29 0.83 1.84 0.79 2.15 0.79 1.14 0.73 1.19 0.74 1.32 0.71 1.41 0.70 0.95 1.04 1.25 1.05 1.25 1.05 1.25 1.05 1.25
23 1.41 0.99 1.15 0.86 1.26 0.88 1.76 0.78 2.17 0.79 1.11 0.73 1.16 0.73 1.29 0.71 1.38 0.70 0.93 1.00 1.19 1.00 1.20 1.00 1.20 1.00 1.20
24 1.40 0.99 1.15 0.88 1.26 0.87 1.71 0.88 2.19 0.72 1.04 0.67 1.09 0.69 1.23 0.64 1.29 0.73 0.96 0.96 1.14 0.96 1.15 0.96 1.15 0.96 1.15
25 1.40 1.01 1.16 0.90 1.27 0.89 1.68 0.78 2.17 0.76 1.06 0.70 1.10 0.73 1.25 0.69 1.30 0.70 0.93 0.94 1.12 0.95 1.13 0.95 1.13 0.95 1.13
26 1.40 0.99 1.14 0.88 1.24 0.88 1.63 0.79 2.20 0.75 1.05 0.71 1.09 0.72 1.22 0.69 1.28 0.69 0.90 0.91 1.08 0.91 1.09 0.91 1.09 0.91 1.08
27 1.41 0.98 1.12 0.87 1.22 0.85 1.60 0.76 2.19 0.75 1.03 0.70 1.07 0.72 1.20 0.68 1.25 0.69 0.90 0.88 1.04 0.88 1.05 0.88 1.05 0.88 1.04
28 1.40 0.98 1.12 0.87 1.20 0.80 1.57 0.88 2.23 0.70 0.97 0.62 0.99 0.65 1.11 0.65 1.21 0.72 0.92 0.84 1.00 0.85 1.01 0.85 1.01 0.85 1.01
29 1.40 1.00 1.13 0.86 1.18 0.79 1.54 0.81 2.24 0.73 0.99 0.66 1.01 0.64 1.09 0.69 1.22 0.70 0.89 0.97 1.20 0.98 1.20 0.97 1.20 0.97 1.20
30 1.40 0.98 1.11 0.85 1.16 0.78 1.52 0.79 2.24 0.73 0.98 0.66 1.00 0.64 1.07 0.69 1.20 0.69 0.87 0.94 1.16 0.95 1.16 0.94 1.16 0.94 1.16
31 1.40 0.97 1.09 0.85 1.14 0.85 1.51 0.80 2.25 0.73 0.97 0.66 0.98 0.64 1.06 0.69 1.18 0.69 0.87 0.91 1.12 0.92 1.12 0.91 1.12 0.91 1.12
32 1.40 0.97 1.09 0.86 1.14 0.82 1.47 0.92 2.31 0.68 0.92 0.62 0.94 0.60 1.01 0.62 1.11 0.69 0.92 0.88 1.09 0.89 1.09 0.88 1.09 0.88 1.09
50 1.38 0.95 1.03 0.84 1.02 0.79 1.22 0.86 1.79 0.66 0.81 0.59 0.79 0.59 0.85 0.58 0.89 0.65 0.80 0.67 0.80 0.67 0.80 0.67 0.81 0.67 0.80
100 1.37 0.94 0.97 0.83 0.91 0.77 0.97 0.84 1.29 0.60 0.68 0.51 0.61 0.50 0.63 0.48 0.64 0.66 0.75 0.51 0.60 0.51 0.60 0.51 0.59 0.51 0.60
150 1.37 0.93 0.95 0.81 0.87 0.76 0.89 0.84 1.14 0.59 0.65 0.51 0.57 0.48 0.57 0.47 0.57 0.65 0.72 0.47 0.55 0.47 0.54 0.47 0.54 0.47 0.54
200 1.36 0.93 0.94 0.81 0.85 0.75 0.85 0.81 1.03 0.58 0.62 0.49 0.54 0.47 0.53 0.45 0.53 0.65 0.71 0.43 0.49 0.43 0.49 0.43 0.48 0.43 0.48
250 1.36 0.93 0.94 0.81 0.85 0.76 0.84 0.81 0.98 0.59 0.61 0.49 0.53 0.47 0.52 0.45 0.51 0.65 0.69 0.42 0.46 0.41 0.46 0.41 0.45 0.41 0.45
300 1.36 0.93 0.94 0.81 0.84 0.76 0.82 0.81 0.95 0.58 0.60 0.48 0.52 0.46 0.51 0.45 0.50 0.65 0.69 0.41 0.45 0.40 0.45 0.40 0.44 0.40 0.44
350 1.36 0.92 0.93 0.81 0.83 0.75 0.81 0.81 0.94 0.58 0.60 0.49 0.51 0.47 0.50 0.45 0.49 0.64 0.68 0.41 0.44 0.40 0.43 0.39 0.43 0.39 0.43
400 1.36 0.92 0.93 0.81 0.83 0.74 0.80 0.79 0.91 0.57 0.59 0.48 0.51 0.46 0.49 0.44 0.48 0.65 0.68 0.39 0.42 0.39 0.42 0.38 0.41 0.38 0.41
450 1.36 0.92 0.93 0.81 0.82 0.75 0.79 0.79 0.89 0.58 0.59 0.49 0.51 0.46 0.49 0.44 0.47 0.64 0.67 0.39 0.41 0.38 0.41 0.37 0.40 0.37 0.40
500 1.36 0.92 0.93 0.80 0.82 0.74 0.78 0.79 0.88 0.57 0.58 0.48 0.50 0.45 0.48 0.44 0.47 0.65 0.67 0.38 0.40 0.38 0.40 0.37 0.39 0.37 0.39
}\tab

\pgfplotstableread{
n a4p a4f d5p d5f
500 0.44 0.47 0.37 0.39
1000 0.430 0.445 0.350 0.364
1500 0.429 0.439 0.341 0.353
2000 0.430 0.438 0.340 0.349
2500 0.431 0.437 0.338 0.345
3000 0.428 0.433 0.340 0.345
3500 0.428 0.432 0.341 0.346
4000 0.427 0.431 0.339 0.343
4500 0.427 0.431 0.338 0.343
5000 0.427 0.430 0.337 0.342
32768 0.425 0.426 0.332 0.333
}\tablong

\pgfplotstableread{
n a4p a4f bm5p bm5f
1 2.38 2.75 2.12 2.12
2 1.50 2.12 1.22 1.22
3 1.25 2.31 1.02 1.38
4 1.52 3.25 1.08 1.62
5 1.49 2.90 0.91 1.34
6 1.39 2.52 0.80 1.18
7 1.28 2.26 0.80 1.12
8 1.07 2.03 0.88 1.38
9 1.15 1.98 0.81 1.24
10 1.08 1.85 0.76 1.14
11 1.05 1.74 0.76 1.10
12 0.94 1.81 0.80 1.12
13 1.00 1.79 0.75 1.05
14 0.96 1.71 0.72 0.99
15 0.94 1.64 0.73 0.97
16 0.80 1.62 0.78 1.12
17 0.85 1.62 0.74 1.06
18 0.83 1.57 0.72 1.01
19 0.82 1.51 0.71 1.00
20 0.70 1.49 0.75 1.02
21 0.72 1.45 0.72 0.98
22 0.71 1.41 0.70 0.95
23 0.71 1.38 0.70 0.93
24 0.64 1.29 0.73 0.96
25 0.69 1.30 0.70 0.93
26 0.69 1.28 0.69 0.90
27 0.68 1.25 0.69 0.90
28 0.65 1.21 0.72 0.92
29 0.69 1.22 0.70 0.89
30 0.69 1.20 0.69 0.87
31 0.69 1.18 0.69 0.87
32 0.62 1.11 0.69 0.92
50 0.58 0.89 0.65 0.80
100 0.48 0.64 0.66 0.75
150 0.47 0.57 0.65 0.72
200 0.45 0.53 0.65 0.71
250 0.45 0.51 0.65 0.69
300 0.45 0.50 0.65 0.69
350 0.45 0.49 0.64 0.68
400 0.44 0.48 0.65 0.68
450 0.44 0.47 0.64 0.67
500 0.44 0.47 0.65 0.67
}\tabbrwpoly

We provide the detailed timing results in Appendix~\ref{sec-detailed-timings}. In the case of $\sym{polyHash}1305$, timing results are provided
for different values of the group reduction parameter $g$ -- for {\tt maax} implementation $g\in\{1,4,8,16,32\}$, and for {\tt avx2} implementation
$g\in\{1,2,3,4\}$. For $\sym{BRWHash}$ and $4\mbox{-}\sym{decBRWHash}1305$, timing results are provided for different values of the look-ahead
parameter $t$, with $t\in\{2,3,4,5\}$. Further, for each case two sets of timings are provided, one for the situation where the key powers are
pre-computed, and the other for the situation where the key powers are computed on the fly.
Based on the timing results in Appendix~\ref{sec-detailed-timings}, here we provide some plots which illuminate the various types of comparisons.

\paragraph{Comparison between {\tt maax} and {\tt avx2} implementations of $\sym{polyHash}1305$.} This comparison is shown in Figure~\ref{fig-polyHash}.
From the figure we note that for the case of pre-computed key powers, {\tt avx2} implementation is faster than the {\tt maax} implementations 
for 16 or more blocks.
For the case of key powers computed on the fly, the two implementations are within $0.02$ cycles/byte for 16 to 19 blocks, and the {\tt avx2} 
implementation is faster for 20 or more blocks. 

\begin{figure}[H]
\centering
\begin{tikzpicture}
\begin{groupplot}[
  group style={group size=2 by 1, horizontal sep=1.3cm, y descriptions at=edge left},
  width=0.42\textwidth, height=6.4cm, xmin=1, xmax=50, ymin=0.4, ymax=2.6,
  xtick={1,8,16,24,32,50},
  xlabel={blocks (16 bytes each)}, ylabel={cycles/byte}]
\nextgroupplot[title={(a) key powers pre-computed}]
 \addplot[maax] table[x=n, y expr={\bestMaaxPre}]{\tab};
 \addplot[avx]  table[x=n, y expr={\bestAvxPre}]{\tab};
 \addplot[refline] coordinates {(16,0.4) (16,2.6)};
 \legend{\texttt{maax} (best $g$), \texttt{avx2} (best $g$)}
\nextgroupplot[title={(b) key powers computed on the fly}]
 \addplot[maax] table[x=n, y expr={\bestMaaxFly}]{\tab};
 \addplot[avx]  table[x=n, y expr={\bestAvxFly}]{\tab};
 \addplot[refline] coordinates {(16,0.4) (16,2.6)};
 \legend{\texttt{maax} (best $g$), \texttt{avx2} (best $g$)}
\end{groupplot}
\end{tikzpicture}
\caption{Comparison of {\tt maax} and {\tt avx2} implementations of $\sym{polyHash}1305$. The dotted line denotes 16 blocks. \label{fig-polyHash}}
\end{figure}

\paragraph{Comparison between {\tt avx2} implementation of $\sym{polyHash}1305$ and previous {\tt avx2} implementations of Poly1305.} 
Implementations of Poly1305 using {\tt avx2}, {\tt avx512}, and {\tt avx512\_ifma} instruction sets were reported in~\cite{DBLP:journals/tches/PegolottiGPP26}.  
Measurements of the number of cycles required were obtained on the Intel Sapphire Rapids processor and are available 
from~\cite{pegolotti2026poly1305benchmarks}. Measurements of the number of cycles required by HACL$^*$~\cite{DBLP:conf/ccs/ZinzindohoueBPB17} and 
OpenSSL~\cite{poly1305-x86-64-cryptogams} implementations
of Poly1305 on the Intel Sapphire Rapids processor are also available from~\cite{pegolotti2026poly1305benchmarks}.
For comparison with our {\tt avx2} implementations, we focus on the {\tt avx2} measurements reported at~\cite{pegolotti2026poly1305benchmarks}.

From the measurements in~\cite{pegolotti2026poly1305benchmarks} of each of {\tt avx2} implementations of SPHGen~\cite{DBLP:journals/tches/PegolottiGPP26}, 
HACL$^*$~\cite{DBLP:conf/ccs/ZinzindohoueBPB17} and OpenSSL~\cite{poly1305-x86-64-cryptogams}, for every value of the number of bytes, we took the minimum 
value of cycles/byte. We converted the number of bytes to the number of blocks by dividing with 16. The plots of the obtained cycles/byte against the number of
blocks are shown in Figure~\ref{fig-poly-avx2}. For the {\tt avx2} implementation of $\sym{polyHash}1305$, from the tables in 
Appendix~\ref{sec-detailed-timings}, for every value of the number of blocks, we took the minimum value of cycles/byte. The corresponding plot of the 
cycles/byte against the number of blocks is shown in Figure~\ref{fig-poly-avx2}.

Our measurements were made on the Tiger Lake processor, while the measurements reported in~\cite{pegolotti2026poly1305benchmarks} were
made on the later generation Sapphire Rapids processor. 
Before commenting on the nature of the comparison in Figure~\ref{fig-poly-avx2}, we would like to clarify the issue of comparing measurements on two 
different processors. It is important to note that both our measurements and those in~\cite{pegolotti2026poly1305benchmarks} are of the
\textit{number of cycles} and not of the actual time.
For the comparison to be meaningful, it is required to consider the latency and throughput of the
various individual instructions. The {\tt avx2} version (working on {\tt ymm} registers) of each of the instructions {\tt vpmuludq}, {\tt vpaddq}, 
{\tt vpand}, {\tt vpxor}, {\tt vpsllq}, {\tt vpsrlq}, has the 
same latency and throughput, and can be executed on the same number of ports on both Tiger Lake and Sapphire Rapids processors.
The {\tt avx2} version of {\tt vmovdqa} has the same latency on both Tiger Lake and Sapphire Rapids processors for both register-to-register,
and memory-to-register moves; however, the throughput of {\tt vmovdqa} for both register-to-register, and memory-to-register moves is better
in Sapphire Rapids than in Tiger Lake. So our {\tt avx2} implementation of $\sym{polyHash}1305$ will require the same or less
number of cycles on Sapphire Rapids than on Tiger Lake. In view of this, it is meaningful to compare the measurements of our {\tt avx2}
implementations of $\sym{polyHash}1305$ on Tiger Lake with the measurements of Poly1305 made on Sapphire Rapids.

Considering Figure~\ref{fig-poly-avx2}, we see that our {\tt avx2} implementation of $\sym{polyHash}1305$ outperforms the {\tt avx2}
implementations of Poly1305 in SPHGen~\cite{DBLP:journals/tches/PegolottiGPP26}, HACL$^*$~\cite{DBLP:conf/ccs/ZinzindohoueBPB17} and 
OpenSSL~\cite{poly1305-x86-64-cryptogams}. The lowest value of cycles/byte among the 300 data points for the {\tt avx2} implementation of Poly1305
as part of SPHGen~\cite{DBLP:journals/tches/PegolottiGPP26} which was 
reported in~\cite{pegolotti2026poly1305benchmarks} is 0.497894 and was obtained for 24704 bytes, i.e. 1544 blocks. In comparison, our {\tt avx2} 
implementation of $\sym{polyHash}1305$ runs at about 0.43 cycles/byte for 1500 block messages, which is about a 14\% speed-up over the 
SPHGen~\cite{DBLP:journals/tches/PegolottiGPP26} implementation of Poly1305.

%


\begin{figure}[H]
\centering
\begin{tikzpicture}
\begin{axis}[
    width=14cm,
    height=9cm,
    xlabel={Block size (16 bytes)},
    ylabel={Cycles/byte},
    xmode=log,
    log basis x=10,
    xmin=1,
    xmax=5000,
    ymin=0.35,
    ymax=3.1,
    xtick={1,2,4,8,16,32,50,100,200,500,1000,2000,5000},
    xticklabels={1,2,4,8,16,32,50,100,200,500,1000,2000,5000},
    grid=both,
    minor grid style={gray!20},
    major grid style={gray!40},
    tick label style={font=\small},
    label style={font=\small},
    legend style={
        font=\small,
        at={(0.98,0.98)},
        anchor=north east,
        draw=black,
        fill=white
    },
    legend cell align=left,
]

\addplot+[thick, mark=none] coordinates {
  (1,2.38) (2,1.50) (3,1.23) (4,1.52)
  (5,1.49) (6,1.39) (7,1.28) (8,1.04)
  (9,1.12) (10,1.06) (11,1.03) (12,0.81)
  (13,0.89) (14,0.87) (15,0.86) (16,0.76)
  (17,0.79) (18,0.77) (19,0.77) (20,0.68)
  (21,0.72) (22,0.71) (23,0.71) (24,0.64)
  (25,0.69) (26,0.69) (27,0.68) (28,0.62)
  (29,0.64) (30,0.64) (31,0.64) (32,0.60)
  (50,0.58) (100,0.48) (150,0.47) (200,0.45)
  (250,0.45) (300,0.45) (350,0.45) (400,0.44)
  (450,0.44) (500,0.44)
  (1000,0.430) (1500,0.429) (2000,0.430)
  (2500,0.431) (3000,0.428) (3500,0.428)
  (4000,0.427) (4500,0.427) (5000,0.427)
};
	\addlegendentry{$\sym{polyHash}1305$}

\addplot+[thick, dashed, mark=none] coordinates {
  (8,1.816561)
  (40,0.895865)
  (72,0.737031)
  (104,0.667241)
  (136,0.622958)
  (168,0.599262)
  (200,0.580001)
  (232,0.570563)
  (264,0.560977)
  (296,0.551222)
  (328,0.547873)
  (360,0.542844)
  (392,0.538561)
  (424,0.536100)
  (456,0.532862)
  (488,0.530836)
  (520,0.524901)
  (552,0.525378)
  (584,0.521371)
  (616,0.521635)
  (648,0.517491)
  (680,0.519551)
  (712,0.515094)
  (744,0.517273)
  (776,0.512847)
  (808,0.515207)
  (840,0.511505)
  (872,0.510694)
  (904,0.510044)
  (936,0.510186)
  (968,0.508364)
  (1000,0.507782)
  (1032,0.507230)
  (1064,0.506810)
  (1096,0.506310)
  (1128,0.505523)
  (1160,0.505214)
  (1192,0.504379)
  (1224,0.502220)
  (1256,0.503574)
  (1288,0.502183)
  (1320,0.501595)
  (1352,0.501250)
  (1384,0.501306)
  (1416,0.499014)
  (1448,0.500210)
  (1480,0.499885)
  (1512,0.500641)
  (1544,0.497894)
  (1576,0.499381)
};
\addlegendentry{SPHGen~\cite{DBLP:journals/tches/PegolottiGPP26}, reported at~\cite{pegolotti2026poly1305benchmarks}}

\addplot+[thick, dotted, mark=none] coordinates {
  (8,2.962474)
  (40,1.049523)
  (72,0.829923)
  (104,0.744191)
  (136,0.700853)
  (168,0.672320)
  (200,0.652119)
  (232,0.639491)
  (264,0.629866)
  (296,0.620531)
  (328,0.614620)
  (360,0.609285)
  (392,0.604691)
  (424,0.602032)
  (456,0.597938)
  (488,0.595001)
  (520,0.592944)
  (552,0.593382)
  (584,0.589643)
  (616,0.588839)
  (648,0.585424)
  (680,0.583767)
  (712,0.582647)
  (744,0.581833)
  (776,0.580391)
  (808,0.580444)
  (840,0.579939)
  (872,0.579472)
  (904,0.577691)
  (936,0.578332)
  (968,0.576733)
  (1000,0.576961)
  (1032,0.575216)
  (1064,0.576207)
  (1096,0.575495)
  (1128,0.575283)
  (1160,0.574231)
  (1192,0.573720)
  (1224,0.572437)
  (1256,0.572121)
  (1288,0.571424)
  (1320,0.571143)
  (1352,0.570856)
  (1384,0.570881)
  (1416,0.570003)
  (1448,0.570052)
  (1480,0.569578)
  (1512,0.569128)
  (1544,0.569330)
  (1576,0.568899)
};
\addlegendentry{HACL$^*$~\cite{DBLP:conf/ccs/ZinzindohoueBPB17}, reported at~\cite{pegolotti2026poly1305benchmarks}}

\addplot+[thick, dashdotted, mark=none] coordinates {
  (8,2.080158)
  (40,0.824595)
  (72,0.684705)
  (104,0.630225)
  (136,0.605704)
  (168,0.589796)
  (200,0.575299)
  (232,0.565984)
  (264,0.558068)
  (296,0.555360)
  (328,0.548873)
  (360,0.545360)
  (392,0.542650)
  (424,0.539890)
  (456,0.537819)
  (488,0.536330)
  (520,0.534585)
  (552,0.532960)
  (584,0.534249)
  (616,0.530909)
  (648,0.529535)
  (680,0.528856)
  (712,0.527564)
  (744,0.527700)
  (776,0.526550)
  (808,0.529028)
  (840,0.525603)
  (872,0.525112)
  (904,0.527365)
  (936,0.524744)
  (968,0.525745)
  (1000,0.524855)
  (1032,0.523411)
  (1064,0.522243)
  (1096,0.522455)
  (1128,0.521426)
  (1160,0.521035)
  (1192,0.521492)
  (1224,0.521378)
  (1256,0.523124)
  (1288,0.519953)
  (1320,0.520362)
  (1352,0.520370)
  (1384,0.520001)
  (1416,0.519478)
  (1448,0.519374)
  (1480,0.519633)
  (1512,0.519380)
  (1544,0.519196)
  (1576,0.519040)
};
\addlegendentry{OpenSSL~\cite{poly1305-x86-64-cryptogams}, reported at~\cite{pegolotti2026poly1305benchmarks}}

\end{axis}
\end{tikzpicture}
\caption{Comparison of the {\tt avx2} implementation of $\sym{polyHash}1305$ with the {\tt avx2} implementations of Poly1305 
in SPHGen~\cite{DBLP:journals/tches/PegolottiGPP26}, HACL$^*$~\cite{DBLP:conf/ccs/ZinzindohoueBPB17}, and OpenSSL~\cite{poly1305-x86-64-cryptogams}. 
The measurements for $\sym{polyHash}1305$ are from the tables in Appendix~\ref{sec-detailed-timings},
while the measurements for the three implementations of Poly1305 are from~\cite{pegolotti2026poly1305benchmarks}. \label{fig-poly-avx2} }
\end{figure}

\paragraph{Comparison between the {\tt maax} implementation of $\sym{BRWHash}1305$ and the {\tt avx2} implementation of $4\mbox{-}\sym{decBRWHash}1305$.}
We chose $t=5$ for both the hash functions. Figure~\ref{fig-BRWHash-comp} shows the comparison. 
For both the cases of pre-computed and on the fly computation of key powers, from 100 blocks onwards $4\mbox{-}\sym{decBRWHash}1305$ is clearly faster.
At around 50 blocks, for the case of pre-computed key powers, $\sym{BRWHash}1305$ is about 3\% faster, while for the case of on the fly computation of
key powers, both are level. 

\begin{figure}[H]
\centering
\begin{tikzpicture}
\begin{groupplot}[
  group style={group size=2 by 1, horizontal sep=1.3cm, y descriptions at=edge left},
  width=0.42\textwidth, height=6.4cm, xmode=log, xmin=1, xmax=500, ymin=0.3, ymax=2.4,
  xtick={1,4,16,50,100,500}, xticklabel style={rotate=45, anchor=east},
  xlabel={blocks (16 bytes each)}, ylabel={cycles/byte}]
\nextgroupplot[title={(a) key powers pre-computed}]
 \addplot[maax] table[x=n, y=bm5p]{\tab};
 \addplot[avx]  table[x=n, y=d5p]{\tab};
 \addplot[refline] coordinates {(100,0.3) (100,2.4)};
 \legend{BRWHash \texttt{maax} ($t{=}5$), 4-decBRWHash \texttt{avx2} ($t{=}5$)}
\nextgroupplot[title={(b) key powers computed on the fly}]
 \addplot[maax] table[x=n, y=bm5f]{\tab};
 \addplot[avx]  table[x=n, y=d5f]{\tab};
 \addplot[refline] coordinates {(100,0.3) (100,2.4)};
 \legend{BRWHash \texttt{maax} ($t{=}5$), 4-decBRWHash \texttt{avx2} ($t{=}5$)}
\end{groupplot}
\end{tikzpicture}
\caption{Comparison between the {\tt maax} implementation of $\sym{BRWHash}1305$ and the {\tt avx2} implementation of $4\mbox{-}\sym{decBRWHash}1305$
The dotted line denotes 100 blocks. \label{fig-BRWHash-comp}}
\end{figure}


\paragraph{Comparison between {\tt maax} implementation of $\sym{BRWHash}1305$ and {\tt avx2} implementation of $\sym{polyHash}1305$.}
For $\sym{BRWHash}1305$, we chose $t=5$ and for the {\tt avx2} implementation of $\sym{polyHash}1305$ we chose $g=4$. The comparison is shown in 
Figure~\ref{fig-brw-poly}. The {\tt avx2} implementation of $\sym{polyHash}1305$ is eventually faster. For the case of pre-computed key powers,
$\sym{polyHash}1305$ becomes faster somewhere between 32 and 50 blocks, while for the case of key powers computed
on the fly, $\sym{polyHash}1305$ becomes faster somewhere between 50 and 100 blocks.

\begin{figure}[H]
\centering
\begin{tikzpicture}
\begin{groupplot}[
  group style={group size=2 by 1, horizontal sep=1.3cm, y descriptions at=edge left},
  width=0.42\textwidth, height=6.4cm, xmode=log, xmin=1, xmax=500, ymin=0.3, ymax=2.6,
  xtick={1,4,16,50,100,500}, xticklabel style={rotate=45, anchor=east},
  xlabel={blocks (16 bytes each)}, ylabel={cycles/byte}]
\nextgroupplot[title={(a) key powers pre-computed}]
 \addplot[brw]  table[x=n, y=bm5p]{\tabbrwpoly};
 \addplot[poly] table[x=n, y=a4p]{\tabbrwpoly};
 \legend{BRWHash \texttt{maax} ($t{=}5$), polyHash \texttt{avx2} ($g{=}4$)}
\nextgroupplot[title={(b) key powers computed on the fly}]
 \addplot[brw]  table[x=n, y=bm5f]{\tabbrwpoly};
 \addplot[poly] table[x=n, y=a4f]{\tabbrwpoly};
 \legend{BRWHash \texttt{maax} ($t{=}5$), polyHash \texttt{avx2} ($g{=}4$)}
\end{groupplot}
\end{tikzpicture}
	\caption{Comparison between {\tt maax} implementation of $\sym{BRWHash}1305$ and {\tt avx2} implementation of 
	$\sym{polyHash}1305$. \label{fig-brw-poly}}
\end{figure}

\paragraph{Comparison among choices of $g$ for $\sym{polyHash}1305$, and choices of $t$ for $4\mbox{-}\sym{decBRWHash}1305$, both 
for {\tt avx2} implementations.}
The comparison is shown in Figure~\ref{fig-g-t}. For $\sym{polyHash}1305$ the situation is the following: with pre-computed key powers, $g=4$ is 
always faster than the other options for 50 blocks or more,
and for on the fly computation of key powers, $g=4$ is slower than $g=2$ or $g=3$ up to 100 blocks, tied at 150--200 blocks, and is faster from 
250 or more blocks. 
For $4\mbox{-}\sym{decBRWHash}1305$, the situation is the following: the other values of $t$ provide speeds which are at most 5.4\% slower than
$t=5$; the values of $t=4$ and $t=5$ are almost the same for messages with more than 100 blocks. 

\begin{figure}[H]
\centering
\begin{tikzpicture}
\begin{groupplot}[
  group style={group size=2 by 2, horizontal sep=1.4cm, vertical sep=2.3cm},
  width=0.42\textwidth, height=5.2cm, xmode=log, xmin=16, xmax=500,
  xtick={16,32,50,100,200,500}, xticklabel style={rotate=45, anchor=east},
  xlabel={blocks}]
\nextgroupplot[title={(a) polyHash, pre-computed}, ymin=0.85, ymax=1.35,
  ylabel={cycles/byte relative to $g{=}4$}]
 \addplot[gone]   table[x=n, y expr={\thisrow{a1p}/\thisrow{a4p}}]{\tab};
 \addplot[gtwo]   table[x=n, y expr={\thisrow{a2p}/\thisrow{a4p}}]{\tab};
 \addplot[gthree] table[x=n, y expr={\thisrow{a3p}/\thisrow{a4p}}]{\tab};
 \addplot[refline] coordinates {(16,1) (500,1)};
 \legend{$g{=}1$,$g{=}2$,$g{=}3$}
\nextgroupplot[title={(b) polyHash, on the fly}, ymin=0.75, ymax=1.35]
 \addplot[gone]   table[x=n, y expr={\thisrow{a1f}/\thisrow{a4f}}]{\tab};
 \addplot[gtwo]   table[x=n, y expr={\thisrow{a2f}/\thisrow{a4f}}]{\tab};
 \addplot[gthree] table[x=n, y expr={\thisrow{a3f}/\thisrow{a4f}}]{\tab};
 \addplot[refline] coordinates {(16,1) (500,1)};
 \legend{$g{=}1$,$g{=}2$,$g{=}3$}
\nextgroupplot[title={(c) 4-decBRWHash, pre-computed}, ymin=0.98, ymax=1.06,
  ylabel={cycles/byte relative to $t{=}5$}]
 \addplot[gone]   table[x=n, y expr={\thisrow{d2p}/\thisrow{d5p}}]{\tab};
 \addplot[gtwo]   table[x=n, y expr={\thisrow{d3p}/\thisrow{d5p}}]{\tab};
 \addplot[gthree] table[x=n, y expr={\thisrow{d4p}/\thisrow{d5p}}]{\tab};
 \addplot[refline] coordinates {(16,1) (500,1)};
 \legend{$t{=}2$,$t{=}3$,$t{=}4$}
\nextgroupplot[title={(d) 4-decBRWHash, on the fly}, ymin=0.98, ymax=1.06]
 \addplot[gone]   table[x=n, y expr={\thisrow{d2f}/\thisrow{d5f}}]{\tab};
 \addplot[gtwo]   table[x=n, y expr={\thisrow{d3f}/\thisrow{d5f}}]{\tab};
 \addplot[gthree] table[x=n, y expr={\thisrow{d4f}/\thisrow{d5f}}]{\tab};
 \addplot[refline] coordinates {(16,1) (500,1)};
 \legend{$t{=}2$,$t{=}3$,$t{=}4$}
\end{groupplot}
\end{tikzpicture}
\caption{For {\tt avx2} implementations, comparison among the choices of $g$ for $\sym{polyHash}1305$, and the choices of $t$ for 
$4\mbox{-}\sym{decBRWHash}1305$.  For $\sym{polyHash}1305$, the curves for the different values of $g$ are relative to the reference value $g=4$.
For $4\mbox{-}\sym{decBRWHash}1305$, the curves for the different values of $t$ are relative to the reference value $t=5$. In both cases, a value
above $1$ means that the corresponding value is slower than the reference.
	\label{fig-g-t} }
\end{figure}

\paragraph{Comparison between the timings with pre-computed key powers and on the fly computation of key powers.}
We consider the {\tt avx2} implementations. The comparison is shown in Figure~\ref{fig-key-power}. 
The y-axis represents $100\,(t_{\mathrm{fly}}-t_{\mathrm{pre}})/t_{\mathrm{pre}}$.
The overhead is about $5$\%--$7\%$ at 500 blocks, $4\%$ or less from 1000 blocks, and about $0.2$\%--$0.3\%$ at $2^{15}$ blocks.

\begin{figure}[H]
\centering
\begin{tikzpicture}
\begin{axis}[width=0.78\textwidth, height=7cm, xmode=log, ymode=log,
  xmin=16, xmax=40000, ymin=0.1, ymax=150,
  xtick={16,50,100,500,1000,5000,32768}, xticklabel style={rotate=45, anchor=east},
  xlabel={blocks (16 bytes each)}, ylabel={on-the-fly overhead (\%)}, legend pos=north east]
 \addplot[avx]  table[x=n, y expr={100*(\thisrow{a4f}-\thisrow{a4p})/\thisrow{a4p}}]{\tab};
 \addplot[avx, forget plot]  table[x=n, y expr={100*(\thisrow{a4f}-\thisrow{a4p})/\thisrow{a4p}}]{\tablong};
 \addplot[maax] table[x=n, y expr={100*(\thisrow{d5f}-\thisrow{d5p})/\thisrow{d5p}}]{\tab};
 \addplot[maax, forget plot] table[x=n, y expr={100*(\thisrow{d5f}-\thisrow{d5p})/\thisrow{d5p}}]{\tablong};
 \addplot[refline] coordinates {(500,0.1) (500,150)};
 \legend{polyHash ($g{=}4$), 4-decBRWHash ($t{=}5$)}
\end{axis}
\end{tikzpicture}
\caption{
Comparison between the timings with pre-computed key powers and on the fly computation of key powers. The dotted line represents 500 blocks. 
	\label{fig-key-power} }
\end{figure}

\paragraph{Speed-up of $4\mbox{-}\sym{decBRWHash}1305$ over $\sym{polyHash1305}$ for \texttt{avx2} implementations.}
This is the main point of comparison and the central focus of the present paper. 
We choose $g=4$ for the {\tt avx2} implementation of $\sym{polyHash1305}$, and $t=5$ for $4\mbox{-}\sym{decBRWHash}1305$.
The comparison is shown in Figure~\ref{fig-brw-over-poly}. The speed-up percentage is measured as $(t_1-t_2)/t_1$, where $t_1$ is the
cycles/byte for $\sym{polyHash1305}$, and $t_2$ is the cycles/byte for $4\mbox{-}\sym{decBRWHash}1305$. The summary is as follows.
\begin{compactdesc}
\item{\em Pre-computed key powers.} For 150 or more blocks (equivalently, 2400 or more bytes), the timing results show that 
	$4\mbox{-}\sym{decBRWHash}1305$ is faster than 
	$\sym{polyHash1305}$. For messages having 200 to 500 blocks (equivalently, 3200 to 8000 bytes), the speed-up is in the range of 4\% to 16\%;
	for messages having 1000 to 5000 blocks (equivalently, 16 KB to 80 KB), the speed-up is in the range of 18\% to 21\%; and for
	messages with 32768 blocks (equivalently, $2^{19}$ bytes), the speed-up is about 22\%.
\item{\em Key powers computed on the fly.} For 16 or more blocks (equivalently, 256 or more bytes) $4\mbox{-}\sym{decBRWHash}1305$ is faster than
	$\sym{polyHash1305}$. For messages having 50 to 500 blocks (equivalently, 800 to 8000 bytes), the speed-up is mostly in the range of
	10\% to 17\%; for messages having 1000 to 5000 blocks (equivalently, 16 KB to 80 KB), the speed-up is in the range of 18\% to 20\%; and for
	messages with 32768 blocks (equivalently, $2^{19}$ bytes), the speed-up is about 22\%.
\end{compactdesc}
\begin{figure}[H]
\centering
\begin{tikzpicture}
\begin{groupplot}[
  group style={group size=2 by 1, horizontal sep=1.4cm},
  width=0.42\textwidth, height=6.6cm, xmode=log, xmin=16, xmax=40000,
  xtick={16,50,150,500,1000,5000,32768}, xticklabel style={rotate=45, anchor=east},
  xlabel={blocks (16 bytes each)}, ylabel={speed-up (\%)}]
\nextgroupplot[title={(a) key powers pre-computed}, ymin=-70, ymax=30]
 \addplot[avx] table[x=n, y expr={100*(\thisrow{a4p}-\thisrow{d5p})/\thisrow{a4p}}]{\tab};
 \addplot[avx, forget plot] table[x=n, y expr={100*(\thisrow{a4p}-\thisrow{d5p})/\thisrow{a4p}}]{\tablong};
 \addplot[refline] coordinates {(16,0) (40000,0)};
 \addplot[refline] coordinates {(150,-70) (150,30)};
\nextgroupplot[title={(b) key powers computed on the fly}, ymin=0, ymax=30]
 \addplot[avx] table[x=n, y expr={100*(\thisrow{a4f}-\thisrow{d5f})/\thisrow{a4f}}]{\tab};
 \addplot[avx, forget plot] table[x=n, y expr={100*(\thisrow{a4f}-\thisrow{d5f})/\thisrow{a4f}}]{\tablong};
 \addplot[refline] coordinates {(16,0) (40000,0)};
\end{groupplot}
\end{tikzpicture}
\caption{
Speed-up of $4\mbox{-}\sym{decBRWHash}1305$ over $\sym{polyHash1305}$ for \texttt{avx2} implementations.
	\label{fig-brw-over-poly}
	}
\end{figure}

To summarise, for {\tt avx2} implementations with key powers computed on-the-fly, timing results show that 
$4\mbox{-}\sym{decBRWHash}1305$ is faster than $\sym{polyHash}1305$
for messages of lengths 256 bytes or more, achieves a speed-up of about 16\% for messages 
which are a few kilobytes long, and the speed-up improves to about 22\% for messages which are a few megabytes long. 
In typical file systems~\cite{DBLP:conf/asist/DinneenN21}, text files are usually a few kilobytes long while media files such as 
pictures, audio and video files, are about a few megabytes long. For such files, i.e. both text files and media files, hashing using 
$4\mbox{-}\sym{decBRWHash}1305$ will be substantially faster than hashing using $\sym{polyHash}1305$.


\section{Conclusion\label{sec-conclu}}
We proposed a new AXU hash function based on BRW polynomials. The hash function is a generalisation of the hash function based on BRW polynomials, with the
generalisation permitting efficient SIMD implementations. For the prime $2^{130}-5$, SIMD implementations of the new hash function using {\tt avx2} instructions
on modern Intel processors show that the new hash function is faster than the well known Poly1305 hash function for messages longer than a few hundred bytes achieving
a speed-up of about 16\% for message lengths in kilobyte range to about 22\% for message lengths in the megabyte range. This makes the new hash function an attractive alternative
to Poly1305 for use in authentication and authenticated encryption systems for general files found in typical file systems.


\appendix

\section{Algorithm for Computing BRW Polynomials \label{sec-BRW-algo}}
Algorithm~\ref{algo-B} describes the algorithm given in~\cite{BNS2025} to compute $\sym{BRW}(\tau;M_1,\ldots,M_l)$, where $l$ is a non-negative integer.
The call $\sym{unreducedBRW}$ in the algorithm performs the following computation. There are two calls to $\sym{unreducedBRW}$ in Algorithm~\ref{algo-B},
in Steps~\ref{B9} and~\ref{B19}. The call in Step~\ref{B9} is on exactly $2^t-1$ blocks, while the call in Step~\ref{B19} is on at most $2^t-1$ blocks.
Since $t$ is a fixed value, both of these calls to $\sym{unreducedBRW}$ are implemented using straight line codes. 
\begin{tabbing}
\ \ \ \ \= $\bullet$ \= \ \ \ \ \=\ \ \ \ \= \ \ \ \ \kill
 \> $\bullet$ \> $\sym{unreducedBRW}(\tau;) = 0$; \\
 \> $\bullet$ \> $\sym{unreducedBRW}(\tau;M_1) = M_1$; \\
 \> $\bullet$ \> $\sym{unreducedBRW}(\tau;M_1,M_2) = \sym{unreducedMult}(M_1,\tau)+M_2$; \\
 \> $\bullet$ \> $\sym{unreducedBRW}(\tau;M_1,M_2,M_3) = \sym{unreducedMult}((\tau + M_1),(\tau^2 + M_2)) + M_3$; \\
 \> $\bullet$ \> $\sym{unreducedBRW}(\tau;M_1,M_2, \ldots, M_{k})$ \\
 \>           \> \> $= \sym{unreducedMult}(\sym{reduce}(\sym{unreducedBRW}(\tau;M_1,\ldots,M_{k-1})),(\tau^k + M_k))$, \\
 \> \> \> if  $k \in \{4,8,16,32,\ldots\}$; \\
 \> $\bullet$ \> $\sym{unreducedBRW}(\tau;M_1,M_2, \ldots, M_{l})$ \\
 \>           \> \> $= \sym{unreducedBRW}(\tau;M_1,\ldots,M_{k})+\sym{unreducedBRW}(\tau;M_{k+1},\ldots,M_{l})$, \\
 \> \> \> if  $k \in \{4,8,16,32,\ldots\}$ and $k < l < 2k$.
\end{tabbing}

\begin{algorithm}
\caption{Evaluation of $\sym{BRW}(\tau;M_1,\ldots,M_l)$, $l\geq 0$. In the algorithm $t\geq 2$ is a parameter. \label{algo-B}}
\begin{algorithmic}[1]
\Function{\sym{ComputeBRW}}{$\tau,M_1,\ldots,M_l$} \label{B1}
  \State ${\sf keyPow}[0]\leftarrow \tau$ \label{B2}
  \If {$l>2$} \label{B3}
    \For {$j\leftarrow 1$ to $\lfloor\lg l\rfloor$}  \label{B4}
      \State ${\sf keyPow}[j]\leftarrow {\sf keyPow}[j-1]^2$  \label{B5}
    \EndFor \label{B6}
  \EndIf \label{B7}
  \State $\sym{top}\leftarrow -1$ \label{B7a}
  \For {$i\leftarrow 1$ to $\lfloor l/2^t\rfloor$} \label{B8}
	\State ${\sf tmp} \leftarrow \sym{unreducedBRW}(\tau;M_{2^t(i-1)+1},\ldots,M_{2^t\cdot i-1})$;  \label{B9}
	\State $k\leftarrow \sym{ntz}(i)$ \label{B10}
	\For {$j\leftarrow 0$ to $k-1$} \label{B11}
		\State ${\sf tmp}\leftarrow {\sf tmp}+\sym{stack}[\sym{top}]$; $\sym{top}\leftarrow \sym{top}-1$ \label{B12}
	\EndFor \label{B14}
	\State $\sym{tmp}\leftarrow \sym{unreducedMult}(\sym{reduce}(\sym{tmp}), M_{2^t\cdot i} + \sym{keyPow}[t+k])$  \label{B15}
	\State $\sym{top}\leftarrow \sym{top}+1$; $\sym{stack}[\sym{top}]\leftarrow \sym{tmp}$ \label{B16}
  \EndFor; \label{B17}
  \State $r\leftarrow l\bmod 2^t$;  \label{B18}
  \State ${\sf tmp}\leftarrow \sym{unreducedBRW}(\tau;M_{l-r+1},\ldots,M_l)$;  \label{B19}
  \State $i \leftarrow \sym{wt}(\lfloor l/2^t\rfloor)$ \label{B20}
	\For {$j\leftarrow 0$ to $i-1$} \label{B21}
	\State $\sym{tmp}\leftarrow \sym{tmp}+\sym{stack}[\sym{top}]$; $\sym{top}\leftarrow \sym{top}-1$  \label{B22}
  \EndFor \label{B23}
  \State return ${\sf reduce}({\sf tmp})$; \label{B24}
\EndFunction. \label{B25}
\end{algorithmic}
\end{algorithm}

\section{Detailed Timing Results \label{sec-detailed-timings}}
The timing results for the {\tt avx2} implementation of $\sym{polyHash}1305$ and $4\mbox{-}\sym{decBRWHash}1305$ are shown in 
Tables~\ref{tab-p130-5-1to8} to~\ref{tab-p130-5-25to32} for messages having 1 to 32 blocks, and in Table~\ref{tab-comparison-p1305} for messages having 50 to 500 blocks.
For comparison, in these tables we also present timing results from~\cite{cryptoeprint:2025/1224} for the {\tt maax} implementation of $\sym{polyHash}1305$ and 
$\sym{BRWHash}1305$ obtained on the
above platform. 
In Table~\ref{tab-comparison-p1305-long}, we present the timing results for the {\tt avx2} implementation of $\sym{polyHash}1305$ with $g=4$
and $4\mbox{-}\sym{decBRWHash}1305$ with $t=5$ for messages having 1000 to 5000 blocks. 

For all the tables, the figures in the cells of the tables denote the number of cycles per byte required to compute the digest by the corresponding hash function
with the stated value of the parameter. Each cell has two figures, the figure on the top denotes the number of cycles per byte when the required key powers are pre-computed
and stored (i.e. the time for generating the key powers are not included in the time for hashing), while the number on the bottom denotes the number of cycles per byte when 
the required key powers are computed on the fly (i.e. the time for generating the key powers are included in the time for hashing).
Recall from Table~\ref{tab-params} that the block size $n$ is 128 bits for the prime $2^{130}-5$, and so the number of blocks mentioned in the tables
can be converted to number of bytes by multiplying with 16.

\begin{table}
	\centering
	{\scriptsize
		\begin{tabular}{|l|r|c|c|c|c|c|c|c|c|} \cline{3-10}
\multicolumn{2}{c}{} & \multicolumn{8}{|c|}{\# msg blks} \\ \cline{3-10}
\multicolumn{2}{c|}{} & 1 & 2 & 3 & 4 & 5 & 6 & 7 & 8 \\ \hline

\multirow{9}{*}{ \begin{tabular}{c} $\sym{polyHash}1305$ \\ ({\tt maax}) \end{tabular} } & $g=1$ & 2.31 & 1.69 & 1.56 & 1.52 & 1.49 & 1.47 & 1.46 & 1.45 \\ \cline{2-10}
& $g=4$ & 2.38 & 1.84 & 1.48 & 1.38 & 1.40 & 1.23 & 1.16 & 1.15 \\
& & 2.38 & 1.84 & 1.88 & 2.05 & 2.14 & 1.86 & 1.71 & 1.63 \\ \cline{2-10}
& $g=8$ & 2.38 & 1.84 & 1.48 & 1.28 & 1.15 & 1.06 & 1.01 & 1.03  \\
& & 2.38 & 1.84 & 1.88 & 1.94 & 1.94 & 1.95 & 1.98 & 2.09  \\ \cline{2-10}
& $g=16$ & 2.25 & 1.84 & 1.48 & 1.28 & 1.15 & 1.06 & 1.01 & 0.97  \\
& & 2.25 & 1.84 & 1.88 & 1.94 & 1.94 & 1.95 & 1.98 & 2.03  \\ \cline{2-10}
& $g=32$ & 2.31 & 1.84 & 1.48 & 1.28 & 1.15 & 1.06 & 1.01 & 0.97  \\
& & 2.31 & 1.84 & 1.88 & 1.92 & 1.94 & 1.98 & 1.99 & 2.02  \\ \hline
\multirow{7}{*}{ \begin{tabular}{c} $\sym{polyHash}1305$ \\ ({\tt avx2}) \end{tabular} } & $g=1$ & 2.38 & 1.50 & 1.23 & 1.55 & 1.56 & 1.41 & 1.30 & 1.04  \\
& & 2.75 & 2.12 & 2.33 & 3.27 & 2.94 & 2.55 & 2.29 & 1.99  \\ \cline{2-10}
& $g=2$ & 2.38 & 1.50 & 1.23 & 1.55 & 1.56 & 1.42 & 1.30 & 1.05  \\
& & 2.81 & 2.12 & 2.31 & 3.28 & 2.95 & 2.55 & 2.29 & 2.02  \\ \cline{2-10}
& $g=3$ & 2.44 & 1.50 & 1.23 & 1.55 & 1.52 & 1.41 & 1.30 & 1.06  \\
& & 2.75 & 2.12 & 2.31 & 3.25 & 2.92 & 2.53 & 2.27 & 2.03  \\ \cline{2-10}
& $g=4$ & 2.38 & 1.50 & 1.25 & 1.52 & 1.49 & 1.39 & 1.28 & 1.07  \\
& & 2.75 & 2.12 & 2.31 & 3.25 & 2.90 & 2.52 & 2.26 & 2.03  \\ \hline
\multirow{8}{*}{ \begin{tabular}{c} $\sym{BRWHash}1305$ \\ ({\tt maax}) \end{tabular}  } & $t=2$ & 2.06 & 1.22 & 0.96 & 1.16 & 0.93 & 0.81 & 0.81 & 0.91  \\
& & 2.06 & 1.22 & 1.38 & 1.73 & 1.40 & 1.21 & 1.14 & 1.40  \\ \cline{2-10}
& $t=3$ & 2.12 & 1.19 & 0.98 & 1.09 & 0.90 & 0.80 & 0.80 & 0.89  \\
& & 2.12 & 1.19 & 1.35 & 1.62 & 1.35 & 1.17 & 1.11 & 1.38  \\ \cline{2-10}
& $t=4$ & 2.19 & 1.22 & 0.98 & 1.08 & 0.90 & 0.81 & 0.80 & 0.89  \\
& & 2.19 & 1.22 & 1.35 & 1.62 & 1.34 & 1.18 & 1.11 & 1.38  \\ \cline{2-10}
& $t=5$ & 2.12 & 1.22 & 1.02 & 1.08 & 0.91 & 0.80 & 0.80 & 0.88  \\
& & 2.12 & 1.22 & 1.38 & 1.62 & 1.34 & 1.18 & 1.12 & 1.38  \\ \hline
\multirow{8}{*}{\begin{tabular}{c} $4\mbox{-}\sym{decBRWHash}1305$ \\ ({\tt avx2}) \end{tabular} } & $t=2$ & 11.19 & 5.62 & 3.73 & 2.80 & 3.17 & 2.65 & 2.27 & 1.98  \\
& & 13.31 & 6.66 & 4.44 & 3.33 & 4.08 & 3.40 & 2.91 & 2.55  \\ \cline{2-10}
& $t=3$ & 11.25 & 5.62 & 3.73 & 2.81 & 3.16 & 2.64 & 2.26 & 1.98  \\
& & 13.31 & 6.66 & 4.44 & 3.33 & 4.06 & 3.39 & 2.90 & 2.55  \\ \cline{2-10}
& $t=4$ & 11.25 & 5.62 & 3.75 & 2.80 & 3.16 & 2.64 & 2.26 & 1.98  \\
& & 13.31 & 6.66 & 4.44 & 3.33 & 4.08 & 3.40 & 2.91 & 2.55  \\ \cline{2-10}
& $t=5$ & 11.25 & 5.62 & 3.73 & 2.81 & 3.16 & 2.64 & 2.26 & 1.98  \\
& & 13.31 & 6.66 & 4.44 & 3.33 & 4.08 & 3.40 & 2.91 & 2.55  \\ \hline

		\end{tabular}
	}
\caption{Cycles/byte measurements for 1 to 8 blocks for the different hash functions based on the prime $2^{130}-5$. \label{tab-p130-5-1to8} }
\end{table}

\begin{table}
\centering
{\scriptsize
\begin{tabular}{|l|r|c|c|c|c|c|c|c|c|}
	\cline{3-10}
	\multicolumn{2}{c}{} & \multicolumn{8}{|c|}{\# msg blks} \\ \cline{3-10}
	\multicolumn{2}{c|}{} & 9 & 10 & 11 & 12 & 13 & 14 & 15 & 16 \\ \hline

\multirow{9}{*}{ \begin{tabular}{c} $\sym{polyHash}1305$ \\ ({\tt maax}) \end{tabular} } & $g=1$ & 1.47 & 1.43 & 1.45 & 1.47 & 1.48 & 1.47 & 1.47 & 1.46 \\ \cline{2-10}
& $g=4$ & 1.19 & 1.12 & 1.07 & 1.07 & 1.11 & 1.06 & 1.03 & 1.03 \\
& & 1.60 & 1.49 & 1.42 & 1.40 & 1.39 & 1.33 & 1.28 & 1.27 \\ \cline{2-10}
& $g=8$ & 1.09 & 1.02 & 0.98 & 0.96 & 0.94 & 0.91 & 0.90 & 0.92  \\
& & 2.11 & 1.94 & 1.83 & 1.73 & 1.65 & 1.57 & 1.51 & 1.50  \\ \cline{2-10}
& $g=16$ & 0.94 & 0.91 & 0.88 & 0.87 & 0.85 & 0.84 & 0.82 & 0.85  \\
& & 2.00 & 1.98 & 2.01 & 2.04 & 2.06 & 2.08 & 2.10 & 2.15  \\ \cline{2-10}
& $g=32$ & 0.94 & 0.91 & 0.88 & 0.86 & 0.85 & 0.84 & 0.82 & 0.82  \\
& & 2.01 & 1.99 & 2.01 & 2.03 & 2.06 & 2.08 & 2.10 & 2.12  \\ \hline
\multirow{7}{*}{ \begin{tabular}{c} $\sym{polyHash}1305$ \\ ({\tt avx2}) \end{tabular} } & $g=1$ & 1.12 & 1.06 & 1.03 & 0.88 & 0.95 & 0.92 & 0.91 & 0.80  \\
& & 1.94 & 1.82 & 1.71 & 1.52 & 1.52 & 1.46 & 1.41 & 1.28  \\ \cline{2-10}
& $g=2$ & 1.12 & 1.07 & 1.03 & 0.81 & 0.89 & 0.87 & 0.86 & 0.76  \\
& & 1.96 & 1.84 & 1.72 & 1.69 & 1.68 & 1.61 & 1.54 & 1.41  \\ \cline{2-10}
& $g=3$ & 1.13 & 1.08 & 1.04 & 0.93 & 0.99 & 0.96 & 0.94 & 0.79  \\
& & 1.97 & 1.84 & 1.73 & 1.78 & 1.77 & 1.69 & 1.62 & 1.62  \\ \cline{2-10}
& $g=4$ & 1.15 & 1.08 & 1.05 & 0.94 & 1.00 & 0.96 & 0.94 & 0.80  \\
& & 1.98 & 1.85 & 1.74 & 1.81 & 1.79 & 1.71 & 1.64 & 1.62  \\ \hline
\multirow{8}{*}{ \begin{tabular}{c} $\sym{BRWHash}1305$ \\ ({\tt maax}) \end{tabular}  } & $t=2$ & 0.81 & 0.78 & 0.76 & 0.83 & 0.77 & 0.74 & 0.74 & 0.79  \\
& & 1.25 & 1.16 & 1.12 & 1.16 & 1.08 & 1.03 & 1.01 & 1.15  \\ \cline{2-10}
& $t=3$ & 0.81 & 0.74 & 0.74 & 0.80 & 0.76 & 0.72 & 0.71 & 0.77  \\
& & 1.23 & 1.14 & 1.10 & 1.13 & 1.05 & 1.00 & 0.97 & 1.12  \\ \cline{2-10}
& $t=4$ & 0.81 & 0.76 & 0.76 & 0.80 & 0.76 & 0.72 & 0.72 & 0.76  \\
& & 1.24 & 1.14 & 1.10 & 1.12 & 1.05 & 0.99 & 0.97 & 1.11  \\ \cline{2-10}
& $t=5$ & 0.81 & 0.76 & 0.76 & 0.80 & 0.75 & 0.72 & 0.73 & 0.78  \\
& & 1.24 & 1.14 & 1.10 & 1.12 & 1.05 & 0.99 & 0.97 & 1.12  \\ \hline
\multirow{8}{*}{\begin{tabular}{c} $4\mbox{-}\sym{decBRWHash}1305$ \\ ({\tt avx2}) \end{tabular} } & $t=2$ & 1.83 & 1.64 & 1.49 & 1.37 & 1.62 & 1.50 & 1.40 & 1.31  \\
& & 2.32 & 2.09 & 1.90 & 1.74 & 1.96 & 1.82 & 1.70 & 1.59  \\ \cline{2-10}
& $t=3$ & 1.83 & 1.64 & 1.49 & 1.37 & 1.61 & 1.49 & 1.39 & 1.30  \\
& & 2.33 & 2.09 & 1.90 & 1.74 & 1.95 & 1.81 & 1.69 & 1.59  \\ \cline{2-10}
& $t=4$ & 1.82 & 1.64 & 1.49 & 1.36 & 1.61 & 1.49 & 1.39 & 1.30  \\
& & 2.33 & 2.09 & 1.90 & 1.74 & 1.95 & 1.81 & 1.69 & 1.58  \\ \cline{2-10}
& $t=5$ & 1.82 & 1.64 & 1.49 & 1.37 & 1.60 & 1.48 & 1.38 & 1.30  \\
& & 2.32 & 2.09 & 1.90 & 1.74 & 1.95 & 1.81 & 1.69 & 1.58  \\ \hline

\end{tabular}
}
\caption{Cycles/byte measurements for 9 to 16 blocks for the various hash functions based on the prime $2^{130}-5$. \label{tab-p130-5-9to16} }
\end{table}

\begin{table}
\begin{center}
{\scriptsize 
\begin{tabular}{|l|r|c|c|c|c|c|c|c|c|}
\cline{3-10}
\multicolumn{2}{c}{} & \multicolumn{8}{|c|}{\# msg blks} \\ \cline{3-10}
\multicolumn{2}{c|}{} & 17 & 18 & 19 & 20 & 21 & 22 & 23 & 24 \\ \hline

\multirow{9}{*}{ \begin{tabular}{c} $\sym{polyHash}1305$ \\ ({\tt maax}) \end{tabular} } & $g=1$ & 1.40 & 1.40 & 1.41 & 1.41 & 1.41 & 1.41 & 1.41 & 1.40 \\ \cline{2-10}
& $g=4$ & 1.06 & 1.03 & 1.00 & 1.01 & 1.03 & 1.05 & 0.99 & 0.99 \\
& & 1.28 & 1.24 & 1.21 & 1.20 & 1.21 & 1.18 & 1.15 & 1.15 \\ \cline{2-10}
& $g=8$ & 0.95 & 0.92 & 0.90 & 0.90 & 0.93 & 0.87 & 0.86 & 0.88  \\
& & 1.49 & 1.43 & 1.39 & 1.36 & 1.32 & 1.29 & 1.26 & 1.26  \\ \cline{2-10}
& $g=16$ & 0.89 & 0.86 & 0.85 & 0.84 & 0.85 & 0.83 & 0.88 & 0.87  \\
& & 2.16 & 2.06 & 1.99 & 1.92 & 1.88 & 1.84 & 1.76 & 1.71  \\ \cline{2-10}
& $g=32$ & 0.81 & 0.80 & 0.79 & 0.79 & 0.79 & 0.79 & 0.78 & 0.88  \\
& & 2.08 & 2.10 & 2.11 & 2.12 & 2.12 & 2.15 & 2.17 & 2.19  \\ \hline
\multirow{7}{*}{ \begin{tabular}{c} $\sym{polyHash}1305$ \\ ({\tt avx2}) \end{tabular} } & $g=1$ & 0.86 & 0.84 & 0.84 & 0.75 & 0.80 & 0.79 & 0.79 & 0.72  \\
& & 1.30 & 1.26 & 1.23 & 1.14 & 1.15 & 1.14 & 1.11 & 1.04  \\ \cline{2-10}
& $g=2$ & 0.81 & 0.81 & 0.79 & 0.68 & 0.74 & 0.73 & 0.73 & 0.67  \\
& & 1.41 & 1.38 & 1.33 & 1.19 & 1.21 & 1.19 & 1.16 & 1.09  \\ \cline{2-10}
& $g=3$ & 0.79 & 0.77 & 0.77 & 0.69 & 0.74 & 0.74 & 0.73 & 0.69  \\
& & 1.53 & 1.48 & 1.44 & 1.34 & 1.35 & 1.32 & 1.29 & 1.23  \\ \cline{2-10}
& $g=4$ & 0.85 & 0.83 & 0.82 & 0.70 & 0.72 & 0.71 & 0.71 & 0.64  \\
& & 1.62 & 1.57 & 1.51 & 1.49 & 1.45 & 1.41 & 1.38 & 1.29  \\ \hline
\multirow{8}{*}{ \begin{tabular}{c} $\sym{BRWHash}1305$ \\ ({\tt maax}) \end{tabular}  } & $t=2$ & 0.75 & 0.72 & 0.74 & 0.78 & 0.74 & 0.72 & 0.73 & 0.76  \\
& & 1.08 & 1.05 & 1.03 & 1.06 & 1.01 & 0.98 & 0.97 & 0.99  \\ \cline{2-10}
& $t=3$ & 0.73 & 0.70 & 0.70 & 0.74 & 0.71 & 0.70 & 0.70 & 0.73  \\
& & 1.06 & 1.02 & 1.00 & 1.02 & 0.98 & 0.95 & 0.94 & 0.97  \\ \cline{2-10}
& $t=4$ & 0.72 & 0.70 & 0.69 & 0.73 & 0.70 & 0.68 & 0.68 & 0.71  \\
& & 1.04 & 1.00 & 0.99 & 1.01 & 0.96 & 0.94 & 0.93 & 0.95  \\ \cline{2-10}
& $t=5$ & 0.74 & 0.72 & 0.71 & 0.75 & 0.72 & 0.70 & 0.70 & 0.73  \\
& & 1.06 & 1.01 & 1.00 & 1.02 & 0.98 & 0.95 & 0.93 & 0.96  \\ \hline
\multirow{8}{*}{\begin{tabular}{c} $4\mbox{-}\sym{decBRWHash}1305$ \\ ({\tt avx2}) \end{tabular} } & $t=2$ & 1.26 & 1.19 & 1.12 & 1.07 & 1.09 & 1.04 & 1.00 & 0.96  \\
& & 1.53 & 1.44 & 1.37 & 1.30 & 1.31 & 1.25 & 1.19 & 1.14  \\ \cline{2-10}
& $t=3$ & 1.26 & 1.19 & 1.12 & 1.07 & 1.10 & 1.05 & 1.00 & 0.96  \\
& & 1.52 & 1.44 & 1.36 & 1.29 & 1.31 & 1.25 & 1.20 & 1.15  \\ \cline{2-10}
& $t=4$ & 1.25 & 1.18 & 1.12 & 1.06 & 1.10 & 1.05 & 1.00 & 0.96  \\
& & 1.52 & 1.43 & 1.36 & 1.29 & 1.31 & 1.25 & 1.20 & 1.15  \\ \cline{2-10}
& $t=5$ & 1.25 & 1.18 & 1.12 & 1.06 & 1.10 & 1.05 & 1.00 & 0.96  \\
& & 1.51 & 1.43 & 1.36 & 1.29 & 1.31 & 1.25 & 1.20 & 1.15  \\ \hline

\end{tabular}
}
\caption{Cycles/byte measurements for 17 to 24 blocks for the various hash functions based on the prime $2^{130}-5$. \label{tab-p130-5-71to24} }
\end{center}
\end{table}

\begin{table}
\begin{center}
{\scriptsize 
\begin{tabular}{|l|r|c|c|c|c|c|c|c|c|}
\cline{3-10}
\multicolumn{2}{c}{} & \multicolumn{8}{|c|}{\# msg blks} \\ \cline{3-10}
\multicolumn{2}{c|}{} & 25 & 26 & 27 & 28 & 29 & 30 & 31 & 32 \\ \hline

\multirow{9}{*}{ \begin{tabular}{c} $\sym{polyHash}1305$ \\ ({\tt maax}) \end{tabular} } & $g=1$ & 1.40 & 1.40 & 1.41 & 1.40 & 1.40 & 1.40 & 1.40 & 1.40 \\ \cline{2-10}
& $g=4$ & 1.01 & 0.99 & 0.98 & 0.98 & 1.00 & 0.98 & 0.97 & 0.97 \\
& & 1.16 & 1.14 & 1.12 & 1.12 & 1.13 & 1.11 & 1.09 & 1.09 \\ \cline{2-10}
& $g=8$ & 0.90 & 0.88 & 0.87 & 0.87 & 0.86 & 0.85 & 0.85 & 0.86  \\
& & 1.27 & 1.24 & 1.22 & 1.20 & 1.18 & 1.16 & 1.14 & 1.14  \\ \cline{2-10}
& $g=16$ & 0.89 & 0.88 & 0.85 & 0.80 & 0.79 & 0.78 & 0.85 & 0.82  \\
& & 1.68 & 1.63 & 1.60 & 1.57 & 1.54 & 1.52 & 1.51 & 1.47  \\ \cline{2-10}
& $g=32$ & 0.78 & 0.79 & 0.76 & 0.88 & 0.81 & 0.79 & 0.80 & 0.92  \\
& & 2.17 & 2.20 & 2.19 & 2.23 & 2.24 & 2.24 & 2.25 & 2.31  \\ \hline
\multirow{7}{*}{ \begin{tabular}{c} $\sym{polyHash}1305$ \\ ({\tt avx2}) \end{tabular} } & $g=1$ & 0.76 & 0.75 & 0.75 & 0.70 & 0.73 & 0.73 & 0.73 & 0.68  \\
& & 1.06 & 1.05 & 1.03 & 0.97 & 0.99 & 0.98 & 0.97 & 0.92  \\ \cline{2-10}
& $g=2$ & 0.70 & 0.71 & 0.70 & 0.62 & 0.66 & 0.66 & 0.66 & 0.62  \\
& & 1.10 & 1.09 & 1.07 & 0.99 & 1.01 & 1.00 & 0.98 & 0.94  \\ \cline{2-10}
& $g=3$ & 0.73 & 0.72 & 0.72 & 0.65 & 0.64 & 0.64 & 0.64 & 0.60  \\
& & 1.25 & 1.22 & 1.20 & 1.11 & 1.09 & 1.07 & 1.06 & 1.01  \\ \cline{2-10}
& $g=4$ & 0.69 & 0.69 & 0.68 & 0.65 & 0.69 & 0.69 & 0.69 & 0.62  \\
& & 1.30 & 1.28 & 1.25 & 1.21 & 1.22 & 1.20 & 1.18 & 1.11  \\ \hline
\multirow{8}{*}{ \begin{tabular}{c} $\sym{BRWHash}1305$ \\ ({\tt maax}) \end{tabular}  } & $t=2$ & 0.73 & 0.72 & 0.72 & 0.75 & 0.72 & 0.72 & 0.72 & 0.75  \\
& & 0.96 & 0.94 & 0.93 & 0.96 & 0.92 & 0.91 & 0.90 & 0.97  \\ \cline{2-10}
& $t=3$ & 0.71 & 0.69 & 0.69 & 0.73 & 0.69 & 0.69 & 0.68 & 0.71  \\
& & 0.93 & 0.91 & 0.90 & 0.92 & 0.89 & 0.87 & 0.87 & 0.95  \\ \cline{2-10}
& $t=4$ & 0.69 & 0.68 & 0.68 & 0.71 & 0.68 & 0.67 & 0.67 & 0.70  \\
& & 0.92 & 0.89 & 0.88 & 0.90 & 0.87 & 0.85 & 0.85 & 0.93  \\ \cline{2-10}
& $t=5$ & 0.70 & 0.69 & 0.69 & 0.72 & 0.70 & 0.69 & 0.69 & 0.69  \\
& & 0.93 & 0.90 & 0.90 & 0.92 & 0.89 & 0.87 & 0.87 & 0.92  \\ \hline
\multirow{8}{*}{\begin{tabular}{c} $4\mbox{-}\sym{decBRWHash}1305$ \\ ({\tt avx2}) \end{tabular} } & $t=2$ & 0.94 & 0.91 & 0.88 & 0.84 & 0.97 & 0.94 & 0.91 & 0.88  \\
& & 1.12 & 1.08 & 1.04 & 1.00 & 1.20 & 1.16 & 1.12 & 1.09  \\ \cline{2-10}
& $t=3$ & 0.95 & 0.91 & 0.88 & 0.85 & 0.98 & 0.95 & 0.92 & 0.89  \\
& & 1.13 & 1.09 & 1.05 & 1.01 & 1.20 & 1.16 & 1.12 & 1.09  \\ \cline{2-10}
& $t=4$ & 0.95 & 0.91 & 0.88 & 0.85 & 0.97 & 0.94 & 0.91 & 0.88  \\
& & 1.13 & 1.09 & 1.05 & 1.01 & 1.20 & 1.16 & 1.12 & 1.09  \\ \cline{2-10}
& $t=5$ & 0.95 & 0.91 & 0.88 & 0.85 & 0.97 & 0.94 & 0.91 & 0.88  \\
& & 1.13 & 1.08 & 1.04 & 1.01 & 1.20 & 1.16 & 1.12 & 1.09  \\ \hline

\end{tabular}
}
\caption{Cycles/byte measurements for 25 to 32 blocks for the various hash functions based on the prime $2^{130}-5$. \label{tab-p130-5-25to32} }
\end{center}
\end{table}

\begin{table}
\centering
{\scriptsize
\begin{tabular}{|l|r|r|c|c|c|c|c|c|c|c|c|c|}
		\cline{3-12}
		\multicolumn{2}{c|}{} & \multicolumn{10}{|c|}{\# msg blks} \\ \cline{3-12}
		\multicolumn{2}{c|}{} & 50 & 100 & 150 & 200 & 250 & 300 & 350 & 400 & 450 & 500 \\ \hline

\multirow{9}{*}{ \begin{tabular}{c} $\sym{polyHash}1305$ \\ ({\tt maax}) \end{tabular} } & $g=1$ & 1.38 & 1.37 & 1.37 & 1.36 & 1.36 & 1.36 & 1.36 & 1.36 & 1.36 & 1.36 \\ \cline{2-12}
& $g=4$ & 0.95 & 0.94 & 0.93 & 0.93 & 0.93 & 0.93 & 0.92 & 0.92 & 0.92 & 0.92 \\
& & 1.03 & 0.97 & 0.95 & 0.94 & 0.94 & 0.94 & 0.93 & 0.93 & 0.93 & 0.93 \\ \cline{2-12}
& $g=8$ & 0.84 & 0.83 & 0.81 & 0.81 & 0.81 & 0.81 & 0.81 & 0.81 & 0.81 & 0.80  \\
& & 1.02 & 0.91 & 0.87 & 0.85 & 0.85 & 0.84 & 0.83 & 0.83 & 0.82 & 0.82  \\ \cline{2-12}
& $g=16$ & 0.79 & 0.77 & 0.76 & 0.75 & 0.76 & 0.76 & 0.75 & 0.74 & 0.75 & 0.74  \\
& & 1.22 & 0.97 & 0.89 & 0.85 & 0.84 & 0.82 & 0.81 & 0.80 & 0.79 & 0.78  \\ \cline{2-12}
& $g=32$ & 0.86 & 0.84 & 0.84 & 0.81 & 0.81 & 0.81 & 0.81 & 0.79 & 0.79 & 0.79  \\
& & 1.79 & 1.29 & 1.14 & 1.03 & 0.98 & 0.95 & 0.94 & 0.91 & 0.89 & 0.88  \\ \hline
\multirow{7}{*}{\begin{tabular}{c} $\sym{polyHash}1305$ \\ ({\tt avx2}) \end{tabular} } & $g=1$ & 0.66 & 0.60 & 0.59 & 0.58 & 0.59 & 0.58 & 0.58 & 0.57 & 0.58 & 0.57  \\
& & 0.81 & 0.68 & 0.65 & 0.62 & 0.61 & 0.60 & 0.60 & 0.59 & 0.59 & 0.58  \\ \cline{2-12}
& $g=2$ & 0.59 & 0.51 & 0.51 & 0.49 & 0.49 & 0.48 & 0.49 & 0.48 & 0.49 & 0.48  \\
& & 0.79 & 0.61 & 0.57 & 0.54 & 0.53 & 0.52 & 0.51 & 0.51 & 0.51 & 0.50  \\ \cline{2-12}
& $g=3$ & 0.59 & 0.50 & 0.48 & 0.47 & 0.47 & 0.46 & 0.47 & 0.46 & 0.46 & 0.45  \\
& & 0.85 & 0.63 & 0.57 & 0.53 & 0.52 & 0.51 & 0.50 & 0.49 & 0.49 & 0.48  \\ \cline{2-12}
& $g=4$ & 0.58 & 0.48 & 0.47 & 0.45 & 0.45 & 0.45 & 0.45 & 0.44 & 0.44 & 0.44  \\
& & 0.89 & 0.64 & 0.57 & 0.53 & 0.51 & 0.50 & 0.49 & 0.48 & 0.47 & 0.47  \\ \hline
\multirow{8}{*}{ \begin{tabular}{c} $\sym{BRWHash}1305$ \\ ({\tt maax}) \end{tabular} } & $t=2$ & 0.70 & 0.71 & 0.71 & 0.71 & 0.70 & 0.71 & 0.70 & 0.70 & 0.70 & 0.70  \\
& & 0.85 & 0.80 & 0.77 & 0.76 & 0.75 & 0.75 & 0.74 & 0.73 & 0.73 & 0.73  \\ \cline{2-12}
& $t=3$ & 0.68 & 0.67 & 0.66 & 0.67 & 0.67 & 0.67 & 0.66 & 0.67 & 0.66 & 0.66  \\
& & 0.82 & 0.77 & 0.74 & 0.72 & 0.71 & 0.71 & 0.70 & 0.70 & 0.69 & 0.69  \\ \cline{2-12}
& $t=4$ & 0.66 & 0.66 & 0.65 & 0.65 & 0.65 & 0.65 & 0.65 & 0.65 & 0.65 & 0.65  \\
& & 0.81 & 0.75 & 0.72 & 0.71 & 0.69 & 0.69 & 0.68 & 0.68 & 0.68 & 0.68  \\ \cline{2-12}
& $t=5$ & 0.65 & 0.66 & 0.65 & 0.65 & 0.65 & 0.65 & 0.64 & 0.65 & 0.64 & 0.65  \\
& & 0.80 & 0.75 & 0.72 & 0.71 & 0.69 & 0.69 & 0.68 & 0.68 & 0.67 & 0.67  \\ \hline
\multirow{8}{*}{ \begin{tabular}{c} $4\mbox{-}\sym{decBRWHash}1305$ \\ ({\tt avx2}) \end{tabular} } & $t=2$ & 0.67 & 0.51 & 0.47 & 0.43 & 0.42 & 0.41 & 0.41 & 0.39 & 0.39 & 0.38  \\
& & 0.80 & 0.60 & 0.55 & 0.49 & 0.46 & 0.45 & 0.44 & 0.42 & 0.41 & 0.40  \\ \cline{2-12}
& $t=3$ & 0.67 & 0.51 & 0.47 & 0.43 & 0.41 & 0.40 & 0.40 & 0.39 & 0.38 & 0.38  \\
& & 0.80 & 0.60 & 0.54 & 0.49 & 0.46 & 0.45 & 0.43 & 0.42 & 0.41 & 0.40  \\ \cline{2-12}
& $t=4$ & 0.67 & 0.51 & 0.47 & 0.43 & 0.41 & 0.40 & 0.39 & 0.38 & 0.37 & 0.37  \\
& & 0.81 & 0.59 & 0.54 & 0.48 & 0.45 & 0.44 & 0.43 & 0.41 & 0.40 & 0.39  \\ \cline{2-12}
& $t=5$ & 0.67 & 0.51 & 0.47 & 0.43 & 0.41 & 0.40 & 0.39 & 0.38 & 0.37 & 0.37  \\
& & 0.80 & 0.60 & 0.54 & 0.48 & 0.45 & 0.44 & 0.43 & 0.41 & 0.40 & 0.39  \\ \hline

\end{tabular}
}
\caption{Cycles/byte measurements for 50 to 500 blocks for the various hash functions based on the prime $2^{130}-5$. \label{tab-comparison-p1305} }
\end{table}

\begin{table}
\centering
{\scriptsize
\begin{tabular}{|l|r|r|c|c|c|c|c|c|c|c|c|}
		\cline{3-11}
		\multicolumn{2}{c|}{} & \multicolumn{9}{|c|}{\# msg blks} \\ \cline{3-11}
		\multicolumn{2}{c|}{} & 1000 & 1500 & 2000 & 2500 & 3000 & 3500 & 4000 & 4500 & 5000 \\ \hline

\multirow{2}{*}{\begin{tabular}{c} $\sym{polyHash}1305$ \\ ({\tt avx2}) \end{tabular} } & $g=4$ & 0.430 & 0.429 & 0.430 & 0.431 & 0.428 & 0.428 & 0.427 & 0.427 & 0.427 \\
& & 0.445 & 0.439 & 0.438 & 0.437 & 0.433 & 0.432 & 0.431 & 0.431 & 0.430 \\ \hline
\multirow{2}{*}{ \begin{tabular}{c} $4\mbox{-}\sym{decBRWHash}1305$ \\ ({\tt avx2}) \end{tabular} } & $t=5$ & 0.350 & 0.341 & 0.340 & 0.338 & 0.340 & 0.341 & 0.339 & 0.338 & 0.337 \\
& & 0.364 & 0.353 & 0.349 & 0.345 & 0.345 & 0.346 & 0.343 & 0.343 & 0.342 \\ \hline

\end{tabular}
}
	\caption{Cycles/byte measurements for 1000 to 5000 blocks for {\tt avx2} implementation of $\sym{polyHash}1305$ with $g=4$ and
	$4\mbox{-}\sym{decBRWHash}1305$ with $t=5$. \label{tab-comparison-p1305-long} }
\end{table}

\end{document}